\documentclass{article}
\usepackage[utf8]{inputenc}
\usepackage{amsmath}
\usepackage{amsfonts}
\usepackage{hyperref}
\usepackage{derivative}
\usepackage{xcolor}
\usepackage{breakcites}
\usepackage{appendix}
\usepackage{comment}
\usepackage{caption}
\usepackage{graphicx}               % Necessary to use \scalebox
\usepackage{amsmath,amssymb}
\usepackage{url}
\usepackage{subfloat}
\usepackage[top=0.5in]{geometry}
\usepackage{fullpage}
\usepackage{epstopdf}
\usepackage{float}
\usepackage{subfloat}
\usepackage{caption}

\usepackage{bm}
\usepackage{longtable}
\usepackage{tabularx}
\usepackage[figuresleft]{rotating}
\usepackage{authblk}
\newsavebox\CBox
\def\textBF#1{\sbox\CBox{#1}\resizebox{\wd\CBox}{\ht\CBox}{\textbf{#1}}}

\usepackage{graphicx}
\usepackage{subfigure}

\usepackage{array}
\usepackage{booktabs}
\usepackage{multirow}

\usepackage{hyperref}
\usepackage{amsmath,amssymb,amsfonts,bm}
\usepackage{amsthm,amsfonts}
\usepackage{mathrsfs}
\usepackage{empheq}
\newtheorem{theorem}{Theorem}

\usepackage{mathtools}

\usepackage{algorithm}
\usepackage{algorithmic}
\newtheorem{lemma}{Lemma}
\newtheorem{assump}{Assumption}
\newtheorem{coro}{Corollary}
\newtheorem{prop}{Proposition}

\theoremstyle{remark}
\newtheorem{example}{\textbf{Example}}
\newtheorem{remark}{Remark}

\theoremstyle{definition}
\newtheorem{definition}{Definition}

\title{Do price trajectory data increase the efficiency of market impact estimation? }

\author[1]{Fengpei Li\footnote{Corresponding author. Email: fengpei.li@mogranstanley.com}}

\author[2,3]{Vitalii Ihnatiuk}
\author[1]{Yu Chen}
\author[1]{Jiahe Lin}
\author[3]{Ryan Kinnear}
\author[1]{Anderson Schneider}
\author[1]{Yuriy Nevmyvaka}
\author [5]{Henry Lam}
\affil[1]{Machine Learning Research, Morgan Stanley}
\affil[2]{Quantitative Research, Morgan Stanley }
\affil[3]{Department of Electrical Engineering, University of Waterloo}
\affil[4]{Department of Economics, Taras Shevchenko National University of Kyiv}
\affil[5]{Department of Industrial Engineering and Operations Research, Columbia University}
\date{}
\begin{document}
\maketitle

\abstract{Market impact is an important problem faced by large institutional investors and active market participants. In this paper, we rigorously investigate whether price trajectory data from the metaorder increases the efficiency of estimation, from the view of the Fisher information, which is directly related to the asymptotic efficiency of statistical estimation. We show that, for popular market impact models, estimation methods based on partial price trajectory data, especially those containing early trade prices, can outperform established estimation methods (e.g., VWAP-based) asymptotically. We discuss theoretical and empirical implications of such phenomenon, and how they could be readily incorporated into practice. }

\section{Introduction}
Market impact is a crucial feature of the market microstructure faced by large traders, and it manifest itself through the adverse effect on the price of the underlying asset generated from the execution of an order. In other words, upon completion of the trade, aside from direct costs (i.e., commissions/fees), slippage from effective bid-ask spread or delay/timing risk, investors are also subject to the transaction cost generated from the price impact of their own actions \cite{robert2012measuring}. For example, a trader who needs to liquidate a large number of shares will take liquidity from the Limit-Order Book (LOB) and push the price down, resulting in a \textit{implementation shortfall} (see \cite{perold1988implementation}) or \textit{liquidation cost} (see \cite{gatheral2013dynamical}), i.e. the difference between the realized revenue and the initial asset value. Besides the short-term correlation between price changes and trades, or the statistical effect of order flow fluctuations (see \cite{bouchaud2010price}), one notable explanation for this dynamics of market impact relates to the reveal of new, private information through trades, which dates back to the seminal work of \cite{kyle1985continuous}. As shown in \cite{kyle1985continuous}, for an investor with private information, to minimize execution cost or the revelation of information, the optimal strategy is to trade incrementally through time. In fact, in modern electronic markets, a decision to trade a large volume (i.e., the full order, also referred to as  $\textit{metaorder}$) is typically implemented by a sequence of incremental executions of the smaller orders (referred to as  $\textit{child orders}$) over a scheduled time window. As discussed in \cite{zarinelli2015beyond,curato2017optimal}, the optimal way to split \textit{metaorders} depends on the cost criterion and the specific market impact model.

The functional form of the price impact along with its corresponding parameters critically characterize these market impact models. First, a large fraction of the market microstructure literature is dedicated to expressing the optimal execution strategies under different risk criteria as a functional of the model parameters (analytically \cite{gatheral2012transient,adrian2020intraday} or implicitly \cite{dang2017optimal,curato2017optimal}). Moreover, the functional forms/parameters also determine the viability/stability of a market impact model, as models with certain functional forms/parameters admit different types of price manipulation strategies (see \cite{huberman2004price,alfonsi2012order,gatheral2013dynamical}) which potentially lead to risky/unstable trading behavior and mathematically preclude the existence of unique optimal execution strategies \cite{gatheral2012transient}. Second, the estimation of model parameters has important implications in explaining and understanding various stylized facts/empirical findings including the \textit{square-root impact law} (\cite{grinold2000active,bucci2019crossover}), the logarithmic dependence in market impact surface (\cite{zarinelli2015beyond}), the power law decay (\cite{bouchaud2003fluctuations,gatheral2010no}) and so on. Finally,  the design of trading strategies that minimize execution cost, as well as pre-trade analytic software that delivers a reliable pre-trade estimate of the expected trade cost, relies crucially on an accurate and efficient estimation of the model parameters. Indeed, following the prevalence of automated trading algorithms, both of the above have become standard considerations for active market participants, especially large institutional investors.

Despite its theoretical importance in market microstructure literature and its practical importance among institutional traders, there are relatively few studies (both empirically and, to a greater extent, theoretically) on the parameter estimation problem for market impact models and we aim to fill this gap in this paper. One main reason is the limited access to the metaorder data for academics and practitioners, as these data  are typically proprietary data of brokerage firms or funds. Consequently, most empirical studies on the market impact estimation are based on publicly available datasets, which collectively suffer from certain inherent drawbacks (e.g., limited information on trades being initiated by buyer vs seller, an unknown number of metaorders, unknown trading duration  \cite{almgren2005direct}) and can only provide a ``partial view" of the market (\cite{zarinelli2015beyond}). Another possible reason is the limited consensus on the appropriate model for price impact (e.g., linear vs non-linear, permanent vs temporary, transient vs instantaneous, see \cite{bouchaud2010price,cont2014price}). Notably, a few exceptions include (\cite{almgren2005direct,moro2009market,zarinelli2015beyond}) which conducted empirical investigations with large metaorder datasets, but the model fitting procedures typically relied only on some summary statistics during the execution. For example, in \cite{almgren2005direct}, the authors tried to determine the exponent of the power law functional form in price impact, for which a nonlinear least square regression is carried out with two statistics: the \textit{realized impact} and the\textit{ permanent impact} (which we explore in detail later). In \cite{zarinelli2015beyond}, in order to fit the \textit{temporary market impact surface}, the market impact measured at the moment when the metaorder is completed was regressed on the metaorder duration and the metaorder participation rate. Moreover, it was also suggested in \cite{curato2017optimal} that, as ``one of the major attractions of
the propagator model to practitioners", the historical execution data on the cost of VWAP executions\footnote{To be specified in Section 2.} (Volume Weighted Average Price, can be seen as another summary statistics, similar to the realized impact in \cite{almgren2005direct}) can be easily used to estimate the functional forms/parameters of the instantaneous market impact function. 

While these summary statistics contain important information about the price trajectory during the order (e.g. cost of VWAP involves averaging of price along the trajectory), it seems uneconomical to simply discard a large part of the price trajectory data during model fitting, especially considering how, from an intuitive standpoint, besides the trade cost itself, the trajectory taken by the price movement to arrive at that cost (referred to as the ``master curve" shape \cite{lillo2003master}) could potentially reveal extra information for model fitting (see also \cite{briere2020modeling}). For example, it has been empirically observed that the market impact of metaorders is concave with respect to the order size \cite{zarinelli2015beyond,toth2011anomalous}, from which one might conjecture the price movements around the early stages of the trade can be especially informative for predicting the total order cost or the market impact shape. Meanwhile, for the owner of metaorder data (i.e. asset managers or brokerage firms), compared with modeling approaches based on LOB, modeling approaches based on the price dynamics would also be feasible, as the additional collection and storage of these extra price data during the life of order should generally not come at a high cost. On the other hand, it is not unusual to see financial data intentionally discarded (due to structural noise or data corruption) for more accurate estimates. For example, it is common practice (see \cite{ait2009estimating,zhang2005tale}) when estimating the volatility of an asset return process to throw away a large fraction of the high-frequency data as a way to avoid the contamination of market microstructure noise (e.g., bid-ask spread). In particular, the realized volatility is typically computed by the sum of less frequently sampled squared returns, i.e. selected on time intervals much larger (i.e. 5 or 10 minutes) than the original ones where data are collected (i.e. every a couple of seconds or less), thus effectively discarding a substantial portion of data. Although the bid-ask spread should not have a substantial effect on the market impact model, as trades within a metaorder mostly have the same sign (i.e. a large buy program usually does not contain many sell orders) and most market impact models in the literature do not include a bid-ask spread (see detailed discussion in \cite{alfonsi2010optimal1} regarding a two-sided versus one-sided LOB model), it remains largely unclear/undiscussed whether the extra price trajectory information should benefit the estimation of market impact models.

In this paper, using the Fisher information, which is directly related to the asymptotic efficiency of statistical estimation, we investigate whether the additional price trajectory data increase the efficiency of market impact estimation. We compare the Fisher information matrix (FI) of different statistical experiments constructed from the same underlying stochastic process, and quantify the \textit{relative efficiency} of their respective \textit{maximum likelihood estimators} (MLE). The validity of this approach in assessing the optimality of experimental design is rooted in the asymptotic optimality of the MLE estimator in \textit{regular} parametric models among the class of \textit{regular} estimators (see, e.g.,\cite{van2000asymptotic} or sections below). Specifically, among the popular existing market impact models, we compare different estimators by their respective Fisher information matrix and observe when one information matrix would dominate another, as this implies an asymptotically smaller variance for estimating any quantity of interest under that parametric model, e.g., the impact of metaorder or the cost of execution. To ensure the broad applicability of our findings, we separately investigate two broad categories of market impact models: the Almgren-Chriss model and the propagator models, which cover a large portion of the parametric models that are currently adopted or studied. Whether the price trajectory data could increase the efficiency of estimation is directly related to whether the current statistic is sufficient. Perhaps surprisingly (or even puzzling), we observe that, even when one does not have access to the full price trajectory data, it does not take many price points at all to achieve a more efficient estimation than well-established (also highly intuitive) methods, e.g., VWAP-based estimation method. For example, we show that in the Almgren-Chriss model, even substituting the realized impact data (the terminology from \cite{almgren2005direct}, equivalent to the cost of VWAP) with just two price points, one in a sufficiently early stage (within the first quarter of trade duration) of the order and the other one at the end of the order, we would get a strictly more asymptotically efficient estimation. One possible, intuitive explanation could be the concavity of the market impact function (see, e.g., \cite{zarinelli2015beyond}), where two carefully chosen points can leverage information more efficiently than VWAP. Similar results can also be observed for the family of propagator models, except more price trajectories are typically more useful. We will explore our findings in detail in Section 3 and 4.

\subsection{Related Works and Motivation}
\subsubsection{Market Impact Modeling}
As one of the central themes in quantitative finance and market microstructure literature, the modeling of market impact is of great interest to practitioners and academic researchers. There is a wide range of literature on market impact modeling for which we only give a partial review here (for a more comprehensive review, see e.g. \cite{gatheral2013dynamical, zarinelli2015beyond}). As one of the most well-known and widely adopted models on market impact models, the Almgren–Chriss model in the influential papers of \cite{almgren1999value,almgren2001optimal} quantifies two distinctive components of price impact: a temporary impact induced by and affecting only an individual's ongoing trade, and a permanent impact affecting all current and future trades equally. Under the Almgren–Chriss model, \cite{almgren2003optimal} extended the analysis in \cite{grinold2000active,loeb1983trading} and solved the optimal execution problem by framing it as a \textit{mean-variance optimization} between the expected execution cost and its variance (representing the uncertainty/liquidity risk during execution). The optimal execution problem was also investigated, from an optimal control perspective, by  \cite{bertsimas1998optimal,forsyth2012optimal,forsyth2011hamilton,gatheral2011optimal} under the geometric Brownian motion assumption for the unaffected stock price process, rather than the arithmetic Brownian motion (ABM) assumption in the Almgren-Chriss model. The discrete and continuous time variants of these models (\cite{bank2004hedging,almgren2007adaptive,brunnermeier2005predatory,carlin2007episodic,cetin2010liquidity}) are collectively termed by \cite{curato2017optimal} as ``first-generation" market impact models, as opposed to the ``second-generation" models (\cite{bouchaud2003fluctuations, bouchaud2006random, gatheral2012transient, gatheral2010no,obizhaeva2013optimal,  alfonsi2010optimal}). The ``second-generation" models, among which the propagator model is perhaps the most notable representative, postulate that the price impact should be neither permanent nor temporary, but \textit{transient}, as it decays over time \cite{moro2009market,lehalle2010rigorous}. As one of the pioneering ``second-generation" models, \cite{obizhaeva2013optimal} proposed a model with linear transient price impact and exponential decay, by modeling an exponential resilience for LOB. The model based on dynamics of the LOB was further developed by \cite{alfonsi2010optimal,alfonsi2008constrained,gatheral2012transient,curato2017optimal,avellaneda2008high,bayraktar2014liquidation,cont2014price,gueant2012optimal}, which include non-linear price impact, as well as LOB with general shape function and time-dependent resilience. On the other hand, instead of modeling the dynamics of LOB, the discrete-time and continuous-time propagator models developed by \cite{bouchaud2003fluctuations,gatheral2010no} directly model the dynamics of price, using decay kernels to reflect the transient nature of the market impact. Detailed discussions on the connection and comparison between these two approaches, as well as further generalizations of propagator models, can be found in \cite{bacry2015market,gatheral2012transient,alfonsi2010optimal1,donier2015fully,toth2011anomalous,curato2017optimal}. Finally, aside from the aforementioned approaches focusing on the interactions between large trades and price changes, other approaches from alternative perspectives also provided many valuable insights on the price impact dynamics. For example, \cite{cont2014price} investigates how price changes are driven by order flow imbalance in the order book events (e.g., quote events); \cite{cardaliaguet2018mean} also investigates the optimal execution problem in a mean-field game setting, where the trader strategically deals with the uncertainty in price and other market participants.

\subsubsection{Stochastic Process and Asymptotic Statistical Inference}

In this section, we briefly review some of the basic concepts from asymptotic statistical inference and discuss how they are related to our setting where the underlying price evolves as a continuous-time stochastic process. In the market impact literature (see e.g.,\cite{gatheral2012transient,almgren2001optimal}), it is typical to assume the ``unaffected'' stock price follows an arithmetic Brownian motion (ABM). The term \textit{unaffected price process} is used in \cite{gatheral2012transient} which refers to the price-determined market participants other than ourselves, i.e., diffusion excluding the drift component. As mentioned in \cite{almgren2001optimal}, while it could be beneficial to consider geometric Brownian motion (GBM) or correlation in price movements for longer investment horizons, ABM remains a suitable model for the unaffected stock price (no drift term or no information on the direction of movement) over the short-term horizon, and its difference with GBM in this regime is almost negligible. In fact, it was investigated in \cite{gatheral2011optimal} that, under the Almgren-Chriss model, the \textit{cost-risk efficient frontier} under GBM and ABM assumption is ``almost identical" and in \cite{almgren2001optimal} that, the improvement gained by incorporating short-term serial correlation in price movement is also small. 

Indeed, as spelled out in \cite{merton1992continuous}, most theoretical models in finance use a continuous-time diffusion process (or general continuous-time Markov process) driven by stochastic differential equations (SDE). However, as documented in \cite{grenander1950stochastic}, the extensive literature on stochastic process ``rarely touched upon" questions of inference. A notable reference for applying an asymptotic approach in statistical estimation/inference to SDE is \cite{bishwal2007parameter}, where the author established the asymptotic property of MLE using the Radon-Nikodym derivative (likelihood) when the whole continuous sample path from the SDE can be observed and sampled. Yet, in practice, the observed data can only possibly be discretely sampled in time and most estimation methods also require
discrete input \cite{ait2004estimators,ait2008closed}. Thus, the parameter estimation and inference problem for discretely, or sometimes non-synchronously/randomly observed diffusion processes (e.g., high-frequency trading) is of much more practical importance. 

The parameter estimation for drift parameters in diffusion processes based on discrete observations has been studied by a few authors (see \cite{bishwal2007parameter} for a comprehensive review). Although a variety of
estimating methods for such data exists, a large portion of such studies utilized (quasi-) likelihood-based estimation/inference \cite{bibby2010estimating}. While the maximum likelihood estimation (MLE) can be a natural choice, the key difficulty is that despite the Markovian nature of diffusion processes allows one to readily calculate the log-likelihood function of discretely sampled
observations simply as the
sum of successive pairs of log-transition function, the transition densities themselves from one time point
to another generally do not have any finite or known analytic forms, except in some special cases. In order to implement the efficient MLE-based methods, many attempts have been made to approximate the likelihood function which leads to approximate maximum likelihood estimators (AMLE). Notably, the ground-breaking series of works in \cite{ait2002maximum,ait2004estimators,ait2008closed} proposed asymptotic expansions of the transition density for the diffusion process which could be used for approximation. Following this line of work by A{\"\i}t-Sahalia, various subsequent analyses and noteworthy applications have been developed (see \cite{li2013maximum,chang2011approximate} for a review), as well as other numerical, simulation-based approaches to approximate likelihood (see, e.g., \cite{pedersen1995consistency}). 

Fortunately, in market impact models, by invo the ABM assumption for the price process, one can evade the technical difficulty of transition density approximation for likelihood-based estimation of impact (i.e., drift term) because the joint distribution of discrete price observations follows a multivariate Gaussian. Within this canonical statistical model, a wide range of basic results in asymptotic inference, especially estimation for \textit{regular} parametric models, are readily available. Precise definitions on \textit{regular} experiments or \textit{regular} parametric models can be found in \cite{bickel1993efficient}, which we also specify in the appendix for completeness. As termed in \cite{bickel1993efficient}, a regular parametric statistical experiment has a ``nice" parameter space in $\theta\in\Theta\subseteq\mathbb R^{K}$  and a density ``smooth" in $\theta$. Most importantly, a regular statistical experiment possesses a non-singular Fisher information matrix at each point $\theta\in\Theta$. As we shall see, the relationship between MLE and Fisher information matrix plays a crucial role in our understanding of asymptotic optimality.

\subsubsection{Asymptotic Optimality and Experimental Design }

MLE is a ubiquitous method in statistical inference and it has many desirable properties in terms of efficiency, feasibility, and generality. In fact, it has been argued that MLE attains asymptotically optimality among the classes of \textit{regular} estimators (precise definitions on \textit{regular} estimators can be seen in \cite{van2000asymptotic}, which we specify in the appendix for completeness. Intuitively, a regular estimator admits a considerable regularity where a small change in parameters does not change the distribution of the estimator too much.). For example (see section 7\&8 in \cite{van2000asymptotic} for results stated below), in this regime, the local asymptotic normality (LAN) and Lipschitzness of log-likelihood can be used to establish the $\sqrt{n}$ - convergence of the MLE estimator to the true parameter under a Gaussian distribution with the inverse of Fisher information matrix as its covariance matrix. This limiting property attained by the MLE, as shown in the  H\'ajek-LeCam convolution theorem and its variant, is the ``best" limiting distribution asymptotically for any \textit{regular} estimator, in the sense that, it is 1) locally asymptotically minimax for any bowl-shaped loss function, i.e., non-negative function with level sets convex and symmetric around the origin, 2) achieves the lowest possible variance (i.e., a quadratic form based on the inverse of the Fisher information matrix) for any asymptotically regular sequence of the estimator and 3) any improvement over this limit distribution can only be made on a Lebesgue null set of parameters. The asymptotic efficiency of MLE has also been discussed in the sense of Bahadur's asymptotic efficiency or C.R.Rao's efficiency (see \cite{ibragimov2013statistical}). A more well-known result, regarded as a simpler version of the H\'ajek-LeCam convolution theorem, is the
Cram\'er-Fr\'echet-Rao information lower bound, which also establishes the asymptotic variance lower bound as the inverse of Fisher information under unbiasedness. Notice the various prerequisites one must declare before one claims the asymptotic optimality of MLE. This is not a mere technicality because, aside from the fact the optimality criterion is not singular in nature, various counter-examples exist outside the confine of such conditions. For example, it is well-known that James-Stein's shrinkage estimator (\cite{lehmann2006theory}) achieves strictly smaller risk for estimating the mean of a $K\geq 3$-dimension multivariate Gaussian with identity covariance matrix under quadratic loss when compared to the MLE (i.e., sample mean). However, the James-Stein estimator is not regular and the improvement over MLE here is for finite sample scenarios, not asymptotic ones. For a counter-example with asymptotic improvement on a Lebesgue null set, one can check the famous Hodges' estimator \cite{van2000asymptotic}.

The discussion above does not aim to debate whether one should necessarily use MLE for the estimation of market impact models. Rather, one can make the observation that, the asymptotic variance attained by MLE, being the ``best" (or ``lowest") possible as the inverse of the Fisher information matrix, quantifies an upper limit on how efficiently one can learn the parameter from a given statistical experiment. As a result, one naturally questions whether one can, by designing statistical experiments that  maximize Fisher information in some sense (more about it below), 
reduces uncertainty in parameter estimation. Indeed, this line of work is pursued extensively in experimental design literature, where the Fisher information matrix has been used to measure the amount of information gained and to design optimal experiments. For example, recently \cite{durant2021determining} uses Fisher information to optimize the experimental design in neutron reflectometry. In this paper, we investigate whether the statistical experiments based on price trajectory are more efficient than the ones based on certain summary statistics. In the experimental design literature (see, e.g., \cite{fedorov2010optimal,whittle1973some,wolkenhauer2008parameter,chaloner1995bayesian}), a unifying, single optimality-criteria for designing experiments has been studied. Traditional methods include maximizing expected trace, minimal eigenvalue, or determinant of Fisher information, corresponding to so-called A-optimality, E-optimality, or D-optimality. However, we note that in the case of estimating multiple parameters, there is an inherent difficulty in estimating all of them accurately, as an optimal way to estimate one particular parameter may not be optimal for the other ones, especially in the context of market impact models where the scales of parameters (or their variances) are vastly different. As a consequence, the meaning of the traditional criteria becomes less clear, unless the Fisher information from one experiment strictly dominates the other (i.e., their difference matrix is positive semi-definite), which we focus on showing in this paper.

\subsection{Roadmap and Outline}
In this section we clarify the roadmap and main contributions of this paper. As mentioned in the previous paragraph, although we discuss and quantify the asymptotic efficiency of MLE (based on which numerical simulation is conducted), the main discourse is concentrated on comparing the Fisher information of different statistical experiments or designs. The MLE serves as natural candidate for numerical verification and its asymptotic optimality is quantified by the Fisher information; this also justifies the comparison among statistical experiments based on dominance of Fisher information. %From the point of view of designing experiment based on Fisher information, we propose efficient way to sample price trajectory data from continuous-time stochastic process in market impact models which outperforms statistical experiments based on common summary statistics. 
To this end, we propose an efficient method---that outperforms those based on common summary statistics---for sampling price trajectory data from the underlying continuous-time stochastic process assumed in the market impact models, where efficiency is measured by the Fisher information.
The sampling scheme and the analysis revolving around it
do not rely on any discretization of the continuous process, although we do consider, for theoretical purpose, discretized price trajectory with fixed increments approaching 0 (similar to the canonical model in \cite{ait2002maximum}, with equal spaced discretization). In particular, in the discussion about maximizing the Fisher information, we borrow insight from the concept of sufficiency in an idealized discretization of price trajectory. The flexibility regarding sampling scheme is an important practical consideration because, although experiment design with samples along a single trajectory approaching infinity could theoretically improve the Fisher information, in reality the price trajectory at a scale finer than a certain threshold would not behave as a continuous-time diffusion process and one would practically want to design robust experiments based on fewer representative trajectory data (e.g., 3 or 4 price points). We focus our discussion in this spirit.

The remainder of the paper is organized as follows. In Section \ref{sec:basic_setup}, we present the basic market impact estimation framework given price trajectory data, including technical lemmas about conditions.  In Section \ref{s3}, we investigate two popular market impact models:  Almgren-Chriss and the family of propagator models. The main result for the Almgren-Chriss model, connecting asymptotic efficiency and sampling of three trajectory points, is presented in Theorem  \ref{acfinal}. The main result for propagator models is presented in Theorem \ref{minisuff}. In sharp contrast with the Almgren-Chriss result, it states that the only sufficient statistic is the full trajectory data when considering general instantaneous and kernel functions. Section \ref{s3} also shows a numerical study comparing different sampling strategies against VWAP-based estimation methods where the importance of early price data is explored. The last section concludes with discussions on some limitations of the theorems, specifically on both model misspecification and model selection. Simulation and empirical results are placed within each section. Reviews of basic concepts, proofs, and technical conditions are left in the Appendix.

\section{Basic Setup and Framework}\label{sec:basic_setup}

\subsection{Background and Model}
Throughout this paper, we assume the metaorder is a buy program so that we do not need to specify the \textit{sign} of a trade. The case for a sell program can be derived analogously. We first focus our discussion on the VWAP (\textit{Volume Weighted Average Price}) execution strategy, where the trading rate $\dot{x}_t=v$ is constant in volume time, where, as a standard assumption in market impact literature, time units are measured by traded volume or \textit{volume time} instead of \textit{physical time}. Typically \textit{volume time} is scaled to adjust for different levels of trading activity during the day, but for this paper, we do not actively distinguish the two times. Consequently, a VWAP execution strategy aims to trade equally/evenly in volume time, which implies trading at a constant proportion against the current traded volume of the stock. In this setting, VWAP strategies can be identified as strategies with a constant trading rate. See \cite{almgren2005direct, gatheral2013dynamical}.  Note that this is not a restriction on the order types, since we are considering the estimation rather than the optimal execution problem. In particular, different execution strategies can be approximated by sequences of interval VWAP strategies with different trading rates (see Remark 3.1 of \cite{curato2017optimal}) and a provably reliable estimation procedure for VWAP execution provides insight for non-VAWP orders as well. Thus, the discussion for non-VWAP strategies shall be deferred to Section \ref{nonvwappp}.

We consider a continuous-time model for the evolution of the underlying stock price during the execution of a metaorder during $0\leq t\leq T$:
\begin{equation}\label{main}
    S_t=S_0+\mu_{\theta}(t,v)+ \sigma\int_0^t dW_s, \text{ for } 0\leq t \leq T
\end{equation}
where $t$ represents time. The trade duration $T$ and trading rate $v$ are given beforehand, i.e., the SDE in \eqref{main} is conditional on $(v,T)$. For a buy program, the drift term $\mu_\theta(t,v)$ in \eqref{main} is initialized with $\mu_\theta(0,v)=0$ and typically a concave function in $t$ \cite{nadtochiy2022simple,lillo2003master} (various forms of $\mu_\theta(t,v)$ will be discussed, see also \cite{almgren2005direct,gatheral2010no}) representing the price impact generated from trading at rate $v$ for a period of $t$. Morevoer,  $\{\mu_\theta\}_{\theta\in\Theta}$ is a family of impact functions parameterized by $\theta\in\Theta\subseteq \mathbb R^K$ as the parameter space. Given fixed $(v,T)$ and $\theta\in\Theta$, $\{S_t\}_{t\in\mathbb T}$ is defined on some filtered probability space $(\Omega,\{\mathcal F_t\}_{t\in\mathbb T}, \mathbb P_\theta)$ where $W_t$ is a standard Brownian motion, $\sigma$ is the volatility and $\mathbb T$ is the time span (typically $\mathbb T=[0,T]$ for some $T$). As it is typically the case for market impact models (and in practice), we only aim to calibrate function $\mu_\theta(t,v)$ while we assume $\sigma$ is fixed or estimated separately. An execution strategy is represented by a continuous function of time $\{x_t\}_{t\in\mathbb T}$ (here we write $x_t$ instead of $X_t$, but in general $x_t$ can also be random and adaptive w.r.t $\mathcal F_t$) with $x_0=0$ and $x_T=X$ indicating the units of shares bought by time $t$, where $X$ is the metaorder size and $T$ is order duration measured by volume time. In general, impact function $\mu$ depends on the entire trajectory of trading rate $\{\dot{x}_t\}_{0\leq t \leq T}$, but for VWAP strategy with constant $\dot{x}_t=v$, we can simplify this dependence into only one variable $v$. 
As in \cite{zarinelli2015beyond}, we characterize two metaorder features (or hyper-parameters): participation rate $v$ and duration $T$. In particular, if $V_D$ is the daily traded volume and $V_M$ is the volume traded by the whole market during the order execution period, then we define 
\begin{equation}\label{tandv}
    T={V_M}/{V_D} \text{ , } v={X}/{V_M},
\end{equation}
where both quantities are unitless and are in $[0,1]$. Notice the definition in \eqref{tandv} would suggest that $T\cdot v=X/V_D$, making the order size $X$ effectively scaled by a factor of $1/{V_D}$. For ease of notation, henceforth we assume $V_D=1$ and treat $X$ as has been scaled. Moreover, as we do not actively distinguish between physical time and volume time so that we retain $X=T v$. Below are our two canonical classes models from \eqref{main}.

\begin{example}
    As a first example, in the Almgren-Chriss model \cite{almgren2005direct}, price dynamics follow as: 
\begin{equation*}
    S_t=S_0+S_0(g(v)t+h(v))+S_0\sigma\int_0^t dW_s.
\end{equation*}
where, in view of \eqref{main}, one can see $\mu(v,t)=S_0(g(v)t+h(v))$ and $\sigma=\sigma S_0$ scaled by $S_0$. As we shall discuss in detail, $g$ represents the permanent impact and $h$ the temporary impact. As $S_0$ can be viewed as fixed and execution cost is typically expressed in basis points (bps), an affine transformation of $P_t=\frac{S_t-S_0}{S_0}$ is carried out in \cite{almgren2005direct} which reduces the dynamics of $P_t$ to the canonical form \eqref{main} with $\mu(t,v)=g(v)t+h(v)$ and volatility $\sigma$. 
\end{example}

\begin{example}
    A second example under the framework of \eqref{main}, is the continuous-time propagator model (see, e.g., \cite{bouchaud2009markets,bouchaud2003fluctuations,gatheral2012transient,curato2017optimal}) developed to quantitatively reflect the \textit{transient} nature of price impact:
\begin{equation}\label{propagator}
    S_t=S_0+\int_0^t f(\dot{x}_s)G(t-s)ds+\sigma\int_0^t dW_s,
\end{equation}
where $f(\cdot)$ is referred to as the \textit{instantaneous market impact function} and $G(\cdot)$ as the $\textit{decay kernel}$ (\cite{gatheral2010no}). For a VWAP strategy, \eqref{propagator} reduces to \eqref{main} with $\mu(t,v)=f(v)\int_0^t G(s)ds$. We shall discuss \eqref{propagator} in detail in Section 3.2.
\end{example}

\subsection{Statistical Experiment and Asymptotic Inference}\label{theorystat}

In the context of this paper, rather than a statistical model, we speak of a statistical experiment constructed from discrete observations or statistics sampled from \eqref{main}, as we cannot observe the whole path under the SDE. We briefly review some basic concepts in the statistical experiment and asymptotic inference, with the technical proofs and additional related material provided in the Appendix. We first lay out some assumptions.

\begin{assump}[\textbf{Model Specification and Identifiability}]\label{modelspec}
$\Theta$ is an open subset of $\mathbb R^K$ and
there exists a unique $\theta^\star\in\Theta$ as the true parameter of the SDE in \eqref{main} for all $(v,T)$ almost surely (a.s.).
\end{assump}

\begin{assump}[\textbf{Metaorder Characteristics}]\label{orderindepG}

The distribution of $(v,T)$ in metaorder data does not depend on $\theta$ and follows some exogenous distribution function $G_{\text{order}}$. We write it as $(v,T)\sim G_{\text{order}}$ with pdf (or p.m.f) $g_{\text{order}}$, i.e., $G_{\text{order}}(dv,dT)=g_{\text{order}}(v,T)dvdT$. We further assume $(v,T)\sim G_{\text{order}}$ satisfies
\begin{equation*}
    v_L\leq v\leq v_U \text { and } T_L\leq T\leq T_U
\end{equation*}
for some constants $0<v_L\leq v_U<\infty$ and $0<T_L\leq T_U<\infty$ a.s.
    
\end{assump}
\begin{remark}\label{tandv}
We make a couple of remarks on the assumptions above. We first note that assumption \ref{modelspec} states \eqref{main} is well-specified within the parametric family. The discussion for model misspecification is deferred to Section \ref{limitation}. The uniqueness of $\theta^\star$ implicitly requires certain identifiability conditions within the parametric family. This is relatively easy to satisfy when the support of $G_{\text{order}}$ is not too narrow, because, as long as one parametrizes $\mu_{\theta}$ carefully, it would then be difficult to find $\theta_1$ and $\theta_2$ such that they $\mu_{\theta_1}(t,v)=\mu_{\theta_2}(t,v)$ for all $0\leq t \leq T$ and all $(v,T)$ pair almost surely. 
We shall also see later that such ``sufficient variability" condition on $G_{\text{order}}$ is crucial in our discussion regarding the nonsingular Fisher information matrix. Such condition on $G_{\text{order}}$ is also not restrictive as the metaorder data typically contains various execution styles and order characteristics reflecting the demands or specifications of clients, resulting in a broad range of values for $(v,T)$ in practice (see the figures on the empirical distribution of $(v,T)$ from real metaorder data in \cite{zarinelli2015beyond}). For assumption \ref{orderindepG}, the lower and upper bounds on $(v,T)$ are reasonable, as one typically does not trade too slowly or too fast. On the other hand, the assumption \ref{orderindepG} on the exogenousity of $G_{\text{order}}$ may not be as straightforward. Indeed, although the distribution $G_{\text{order}}$ undoubtedly depends on many exogenous factors such as requests of clients or trading styles, it is not immediately clear whether one can claim $G_{\text{order}}$ has no dependence on $\theta$. For example, it is conceivable that the traders would, over time, estimate the parameter up to some accuracy and adapt their trading strategy accordingly for all the exogenous requirements, to satisfy certain optimal trading schedules (e.g., \cite{almgren2003optimal}), resulting in a ``shift" of $G_{\text{order}}$ towards their acquired knowledge of $\theta$. In this paper, we assume such dependence is negligible, but we do note that assumption \ref{orderindepG} could be a source of bias and should be a subject of future deliberation.  
\end{remark}
\subsubsection{Statistical Experiment \& Regularity Conditions}
We are now ready to discuss the experiment design derived from \eqref{main}. In this paper, we primarily consider statistical experiments consisting of discrete observations or summary statistics of the following form: given $N, M\in\mathbb N^+$, let\footnote{We use $[N]$ to denote the set $\{1,2,...,N\}$.} $\boldsymbol S =\{S_{t_i}\}_{i\in[N]}$ and $\boldsymbol J=\{\int_{t_j^1}^{t_j^2} S_tdt\}_{j\in[M]}$, the statistical experiment on $\mathbb R^{2+N+M}$ consists of the triplet,
\begin{equation}\label{statexp}
   (\boldsymbol X, \mathscr B^{N+M}, \{\mathbb P^{\boldsymbol S\cup \boldsymbol J}_\theta\}_{\theta\in\Theta})
\end{equation}
where $\boldsymbol X\triangleq\{v,T\}\cup\boldsymbol S\cup \boldsymbol J$. Here $\boldsymbol S$ are $N$ price trajectory data (i.e., observations from a single trajectory)
with $\tau_i\triangleq t_i/T$ fixed; $\boldsymbol J$ are $M$ summary statistics proportional to the average cost (or price) along a certain time window with $\tau_j^1\triangleq t_j^1/T,\tau_j^2\triangleq t_j^2/T$ fixed. Such selection of $0\leq t\leq T$ based on fixed ratio $\tau$ avoids inconsistencies in choosing $t$ when $T$ is random. Here $\mathscr B^{2+N+M}$ is the Borel-sigma algebra on $\mathbb R^{2+N+M}$ and $\mathbb P^{\boldsymbol S\cup \boldsymbol J}_\theta$ is the product measure of $G_{\text{order}}$ and the probability measure induced from the SDE in \eqref{main} conditional on $(v,T)$ when restricted to $\boldsymbol S\cup \boldsymbol J$. For ease of notation, we omit the upper script in $\mathbb P^{\boldsymbol S\cup \boldsymbol J}_\theta$ and simply write $\mathbb P_\theta$ (and density function $p_{\theta}$ as well) when there is no ambiguity about the sample space in question. One can consider a sample $\boldsymbol X$ under \eqref{statexp} generated as: first sample $(v,T)\sim G_{\text{order}}$, then generate a sample path based on SDE \eqref{main}, and finally record $\boldsymbol S$ and $\boldsymbol J$.

Based on \eqref{main}, it is straightforward to see that, conditional on $(v,T)$, the sample $\boldsymbol S\cup \boldsymbol J$ follows a $N+M$-dimensional multivariate Gaussian distribution
\begin{equation}\label{gassexp}
    \mathcal N(\mu(\theta, T,v), \Sigma(T)),
\end{equation}
with a mean function $\mu(\theta,T,v)$ given by
\begin{equation*}
    \mathbb ES_{t}=\mu_\theta(t,v) \text{ and }\mathbb E\int_{t_1}^{t_2} S_tdt=\int_{t_1}^{t_2}\mu_\theta(t,v)dt,
\end{equation*}
 and a covariance matrix $\Sigma(T)$ given by 
\begin{align*}
    \text{Cov}[(S_{t_1},S_{t_2})]=&\sigma^2\text{Cov}[(W_{t_1},W_{t_2})]\nonumber\\ \text{Cov}[(S_{t_1},\int_{{t_2}}^{t_3}S_tdt)]=&\sigma^2\text{Cov}[(W_{t_1},\int_{{t_2}}^{t_3}W_tdt)]\nonumber\\\text{Cov}[(\int_{t_1}^{t_2}S_tdt,\int_{t_3}^{t_4}S_tdt)]=&\sigma^2\text{Cov}[(\int_{t_1}^{t_2}W_tdt,\int_{t_3}^{t_4}W_tdt)],
\end{align*}
 all of which can be readily computed using elementary It\^{o} calculus on standard Brownian motion (e.g., $\text{Cov}[(W_s,W_t)]=\min(s,t)$). A notable consequences is that $\Sigma(T)$ has no dependence on $\theta$ or $v$.

We assume there is no linear dependence among elements in $\boldsymbol S\cup\boldsymbol J$ so that $\Sigma(T)\succ 0$ and $\Sigma^{-1}(T)$ exists. Thus, the likelihood function of a sample from \eqref{statexp} can be written as 
\begin{equation}\label{oglike}
    p_\theta(\boldsymbol X) = g_{\text{order}}(v,T)\frac{1}{\sqrt{2\pi|\Sigma(T)|}} \exp-\frac{\Big(\mu(\theta,T,v)-(\boldsymbol S,\boldsymbol J)\Big)^T\Sigma^{-1}(T) \Big(\mu(\theta,T,v)-(\boldsymbol S,\boldsymbol J)\Big)}{2},
\end{equation}
and log-likelihood 
\begin{equation}\label{loglike}
    l(\theta|\boldsymbol X) =\log p_\theta(\boldsymbol X).
\end{equation}
Let $\boldsymbol X_1,\boldsymbol X_2,...,\boldsymbol X_n$ be $n$ i.i.d. samples from experiment \eqref{statexp}. Let $\hat\theta_n$ be the maximum likelihood estimator such that
\begin{equation}\label{mlebasic}
    \hat\theta_n =\underset{\theta\in\Theta}{\mathrm{argmax}} \sum_{i=1}^n l(\theta|\boldsymbol X_i).
\end{equation}

\begin{example}
    Typically one has price trajectory data $\boldsymbol{S}=\{S_{t_i}\}_{i\in [N]}$ and $\boldsymbol J=\{\int_{t_j^1}^{t_j^2} S_tdt\}_{j\in[M]}$ from metaorders, in the forms of (\cite{zarinelli2015beyond, almgren2005direct}):
\begin{enumerate}
    \item The cost of execution: $\int S_t \dot{x}_t dt-XS_0$ (divide by $XS_0$ for bps),
    \item The peak impact $S_T$,
    \item The ``permanent" impact $S_{T_{\text{post}}}$ for $T_{\text{post}}>T$ (e.g., 30 mins after trade \cite{almgren2005direct}),
    \item The quantification of price trajectory for $\mu_\theta(t,v)$ for some $t<T$.
\end{enumerate}
For example, conditional $(v,T)$, the cost of VWAP execution $C_{\text{VWAP}}=v\int S_tdt-XS_0$ follows Gaussian as in \eqref{gassexp}
\begin{equation}\label{vwapexp}
   C_{\text{vwap}}\sim\mathcal N (c_{\text{vwap}}(\theta,T,v),\sigma_{\text{vwap}}^2 (T,v)), 
\end{equation}
where the mean $c_{\text{vwap}}(\theta,T,v)=v\int_0^T\mu_{\theta}(t,v)dt$ and variance $\sigma_{\text{vwap}}^2 (T,v)=\frac{v^2T^3}{3}\sigma^2$ (derivation left in Appendix). Another example, given a single sample corresponding to 
price trajectory data $\boldsymbol X = (v,T)\cup \{S_{t_i}\}_{i\in [N]}$, the log-likelihood follows (derived using independent increments of $S_t$): \begin{align}\label{loglike}
    &l(\boldsymbol{X}|\theta)\triangleq\log(p_\theta(\boldsymbol{X}))=\log g_{\text{order}}(v,T) -\frac{N}{2}\log(2\pi\sigma^2)\nonumber\\
    &+\sum_{i\in [N] } -\frac{1}{2}\log(\tau_i T-\tau_{i-1}T) -\frac{1}{2\sigma^2(\tau_i-\tau_{i-1})T}\Big((S_{\tau_i T}-S_{\tau_{i-1}T})-(\mu_{\theta}(\tau_iT,v)-\mu_{\theta}(\tau_{i-1}T,v))\Big)^2.
\end{align}
Then, given $n$ total orders samples $\boldsymbol X_j$, executed by VWAP strategy under $(T_j,v_j)$, with $\{\boldsymbol{S}_j\}_{1\leq j \leq n}$ being the $j$-th price trajectory data $\boldsymbol{S}_j=\{S_{\tau_i T_j}\}_{i\in [N]}$, the MLE is the estimator that maximizes the log-likelihood ratio:
\begin{align}\label{mle}
    \hat\theta_{n}=&\underset{\theta\in\Theta}{\arg\max} \sum_{j=1}^n l(\boldsymbol{X}_j|\theta)\nonumber\\
    =& \underset{\theta\in\Theta}{\arg\min}\sum_{j=1}^M\sum_{i\in \mathcal K_j}\frac{1}{t_{ij}-t_{(i-1)j}} \Big((S_{t_{ij}}-S_{t_{(i-1)j}})-(\mu_{\theta}(t_{ij}, v_j)-\mu_{\theta}(t_{(i-1)j}, v_j))\Big)^2,
\end{align}
where the second equality follows from \eqref{loglike}.
\end{example}

\begin{assump}[\textbf{Gaussian Experiment}]\label{ras3}
        For the statistical experiment in \eqref{statexp}, we assume 
        \begin{enumerate}
              \item  $\mathbb P_{\theta}l(\theta|\boldsymbol X)<\infty$, $\forall \theta\in\Theta$.
              \item For the multivariate Gaussian in \eqref{gassexp}, there exists some $\epsilon$ such that $\Sigma(T)\succcurlyeq\epsilon\sigma^2 \text{\textbf{I}}$ $\forall T$ a.s.
              \item For the multivariate Gaussian in \eqref{gassexp}, $\mu(\theta,T,v)$ is continuously differentiable in $\theta$ $\forall (v,T)$ a.s.
              \item For the multivariate Gaussian in \eqref{gassexp}, there exists a neighborhood around every $\theta$ such that  $\forall \theta'$ in this neighborhood, $\mu(\theta', T,v)\leq B$ and $|\mu(\theta', T,v)-\mu(\theta'', T,v)| \leq L\|\theta'-\theta''\|$ for some $B,L$, $\forall \theta',\theta''$ in this neighborhood and $\forall (v,T)$ a.s.
              \item For the multivariate Gaussian in \eqref{gassexp}, $\theta^\star$ is the unique true parameter.
        \end{enumerate}
        \end{assump}
\begin{remark}
We make a couple of remarks on assumption \ref{ras3}. Assumption 3.1 is standard. Assumption 3.2 is related to the design of the Gaussian experiment, where one must choose $\boldsymbol S \cup \boldsymbol J$ so that they are not linearly dependent (otherwise one can simply delete redundant observations). This would ensure that $\Sigma (T)\succ 0$. The uniform lower bound $\epsilon$ on its eigenvalue for all $T$, as we shall see, hinges on the lower and upper bound of $T$ in assumption \ref{orderindepG}. For assumption 3.3, one simply can check whether $\mu_\theta(t,v)$ in \eqref{main} is continuously differentiable $\forall (v,T)$ in the bounded support of $G_{\text{order}}$. Assumption 3.4 ensures Lipschitz continuity of $\mu_\theta(t,v)$ in \eqref{main} $\forall (v,T)$ in the bounded support of $G_{\text{order}}$. Assumption 3.5 is different from assumption \ref{modelspec} as $\mathbb P_\theta$ is a projection of the probability measure induced from SDE \eqref{main} onto a finite-dimensional space \eqref{statexp}. Intuitively, the higher the dimension for $\theta$, the higher the dimension of $\boldsymbol X$ we need for the model identifiability in \eqref{statexp}.
    \end{remark}

\begin{definition}
    With assumption 3.3 in place, we write the derivative of $\mu(\theta,T,v)$ w.r.t $\theta$ as a Jacobian matrix

\begin{equation}\label{fishJJ}
    \mathcal J(\theta, T,v)=\begin{bmatrix} 
    \pdv{\mu_1}{\theta_1} & \pdv{\mu_1}{\theta_2} & \dots & \pdv{\mu_1}{\theta_K} \\
    \pdv{\mu_2}{\theta_1} & \pdv{\mu_2}{\theta_2} & \dots & \pdv{\mu_2}{\theta_K} \\
    \vdots & \vdots & \ddots& \vdots\\
    \pdv{\mu_{N+M}}{\theta_1} & \pdv{\mu_{N+M}}{\theta_2} & \dots & \pdv{\mu_{N+M}}{\theta_K} 
    \end{bmatrix}.
\end{equation}
Consequently, given $\theta_0\in\Theta$, we can define the Fisher information matrix in $\mathbb R^{K\times K}$ for experiment \eqref{statexp}:
\begin{equation}\label{fisher}
     \mathcal I_{\boldsymbol{X}}(\theta_0)=\mathbb E_{\theta_0}\bigg[\Big(\frac{\partial l(\theta|\boldsymbol{X})}{\partial \theta}\Big)\Big(\frac{\partial l(\theta|\boldsymbol{X})}{\partial \theta}\Big)^T\bigg],
\end{equation}
where the subscript in $\mathbb E_{\theta_0}$ denotes expectation under $\mathbb P_{\theta_0}$ and the subscript in $\mathcal I_{\boldsymbol{X}}$ denotes the experiment design in \eqref{statexp} based on sample $\boldsymbol X$.
\end{definition}

Before we discuss the relevance of the Fisher information and its importance, we first provide a known result which allows convenient computation for it in a Gaussian experiment \eqref{gassexp}.

\begin{lemma}\label{jacob}
Let $X\sim \mathcal N(a(\theta),\Sigma)$ be the $N$-dimensional multivariate Gaussian distribution with known $\Sigma$ where $\theta\in \mathbb R^K$. Then, let $ D(\theta)\in\mathbb R^{N\times K}$ be Jacobian matrix where
\begin{equation*}
     D(\theta)=\begin{bmatrix} 
    \pdv{a_1}{\theta_1} & \pdv{a_1}{\theta_2} & \dots & \pdv{a_1}{\theta_K} \\
    \pdv{a_2}{\theta_1} & \pdv{a_2}{\theta_2} & \dots & \pdv{a_2}{\theta_K} \\
    \vdots & \vdots & \ddots& \vdots\\
    \pdv{a_N}{\theta_1} & \pdv{a_n}{\theta_2} & \dots & \pdv{a_N}{\theta_K} 
    \end{bmatrix},
\end{equation*}
we have
\begin{equation}
    \mathcal I_X(\theta)= D(\theta)^T\Sigma^{-1} 
 D(\theta).
\end{equation}
\end{lemma}
\begin{proof}
    See example 7.7 in \cite{van2000asymptotic}.
\end{proof}

Based on Lemma \ref{jacob}, we have the following Proposition.
\begin{prop}\label{fishtvtv}
Under assumption \ref{orderindepG} and \ref{ras3}, the statistical experiment \eqref{statexp} has its Fisher information in \eqref{fisher} as
\begin{equation}\label{statexpFI}
    \mathcal I_{\boldsymbol X}(\theta) = \mathbb E_{(v,T)\sim G_{\text{order}}} \bigg[\mathcal J^T(\theta, T,v)\Sigma^{-1}(T) \mathcal J(\theta, T,v)\bigg].
\end{equation}
\end{prop}
\begin{proof}
Due to assumption \ref{orderindepG}, we can infer from \eqref{oglike} that
\begin{equation*}
    l(\theta|\boldsymbol{X}) = \log g_{\text{order}}(v,T) + \log p_{\theta}(\boldsymbol X|(v,T)) = \log g_{\text{order}}(v,T) + \log p_{\theta}(\boldsymbol S\cup \boldsymbol J|(v,T)) 
\end{equation*}
Since $g_{\text{order}}(v,T)$ does not depend on $\theta$ by assumption \ref{orderindepG} and $\boldsymbol S\cup \boldsymbol J$ is multivariate Gaussian conditional on $(v,T)$, we can evoke Lemma \ref{jacob} and the tower property of conditional expectation to write
    \begin{align*}
        \mathcal I_{\boldsymbol{X}}(\theta)=&\mathbb E_{(v,T)\sim G_{\text{order}}}\bigg[\mathbb E\bigg[\Big(\frac{\partial l(\theta|\boldsymbol{X})}{\partial \theta}\Big)\Big(\frac{\partial l(\theta|\boldsymbol{X})}{\partial \theta}\Big)^T\bigg|(v,T)\bigg]\bigg]\nonumber\\
        =&\mathbb E_{(v,T)\sim G_{\text{order}}}\bigg[\mathcal I_{\boldsymbol X|(v,T)}(\theta)\bigg]\nonumber\\
        =&\mathbb E_{(v,T)\sim G_{\text{order}}} \bigg[\mathcal J^T(\theta, T,v)\Sigma^{-1}(T) \mathcal J(\theta, T,v)\bigg].
    \end{align*}
    where the $\mathcal I_{\boldsymbol X|(v,T)}$ is computed from conditional likelihood, treating $(v,T)$ as fixed.
\end{proof}
 Lastly, we present our final assumption which is essential for establishing regular parametric models.
\begin{assump}[\textbf{Nonsingular Fisher Information}]\label{nonsigfishh}
    We assume the Fisher information $\mathcal I_{\boldsymbol{X}}(\theta)$ in \eqref{statexpFI} is well-defined, non-singular and continuous in $\theta$.
\end{assump}

As we shall see, assumption \ref{nonsigfishh} is related to $G_{\text{order}}$. Simply put, it would be hard to determine a high-dimensional $\theta$ in the market impact model if trading style $(v,T)$ is too singular, as one might struggle to separate different impact effects from $v$ or $T$. One would need to check assumption \ref{nonsigfishh} on a case-by-case basis. For example, a point mass distribution of $G_{\text{order}}$ (i.e., one pair of $(v,T)$) is almost never capable of fitting a good market impact model.

\subsubsection{Asymptotic Optimality and Fisher Information}

Now we present some consequences of assumptions \ref{modelspec}-\ref{nonsigfishh}. The results discussed here are related to basic concepts in asymptotic inference, such as $\textit{regular}$ parametric model, local asymptotic normality (LAN), and \textit{regular} estimators. We do not dive into the details of these established results and we defer both the proof and the reference for these concepts to the appendix. The purpose of the following proposition is to establish that, under fairly reasonable criteria and a fairly broad class of estimators, the MLE achieves asymptotic optimality given the aforementioned conditions and that optimality is quantified by the Fisher information.

\begin{prop}\label{sdfadsfasdfas}
Let $\boldsymbol X_1,\boldsymbol X_2,...,\boldsymbol X_n$ be $n$ i.i.d. samples from experiment \eqref{statexp}. Let $\hat\theta_n$ be the maximum likelihood estimator such that
\begin{equation*}
    \hat\theta_n =\underset{\theta\in\Theta}{\mathrm{argmax}} \sum_{i=1}^n l(\theta|\boldsymbol X_i).
\end{equation*}
    Then, under assumptions \ref{modelspec}-\ref{nonsigfishh}, the MLE is consistent for $\theta^\star$, i.e., $\hat\theta_{n}\rightarrow \theta^\star$ in probability and $\sqrt{n}$-asymptotically normal:
\begin{equation}\label{mleinv}
    \sqrt{n}(\hat\theta_{n}-\theta^\star) \overset{d}{\rightarrow} \text{ } \mathcal N\big(0, \mathcal I^{-1}_{\boldsymbol{X}}(\theta^\star)\big) \text{ and } \sqrt{n}(\phi(\hat\theta_{n})-\phi(\theta^\star)) \overset{d}{\rightarrow} \text{ } \mathcal N\big(0, \nabla_\theta\phi(\theta^\star)^T\mathcal I^{-1}_{\boldsymbol{X}}(\theta^\star)\nabla_\theta\phi(\theta^\star)\big)
\end{equation}
%for any differentiable in $\phi(\theta)$. 
for any $\phi(\theta)$ differentiable in $\theta$. 
Moreover, for any bowl-shaped loss function $l$ (i.e., $\{x:l(x)\leq c\}$ is convex and symmetric around the origin) and any estimator sequence $\{T_n\}_{n\geq0}$,
\begin{equation}\label{lammmmax}
    \sup_{I}\liminf_{n\rightarrow\infty}\sup_{h\in I}\mathbb E_{\theta^\star+\frac{h}{\sqrt{n}}}l\Big(\sqrt{n}\Big(T_n-\phi(\theta^\star+\frac{h}{\sqrt{n}})\Big)\Big)\geq l(Z_{\theta^\star})
\end{equation}
where $I$ is taken over all finite subset of $\mathbb R^K$ and $Z_{\theta^\star}\sim\mathcal N\big(0, \nabla_\theta\phi(\theta^\star)^T\mathcal I^{-1}_{\boldsymbol{X}}(\theta^\star)\nabla_\theta\phi(\theta^\star)\big)$.
\end{prop}

The proof is presented in the appendix. The argument follows from section 8.7 of \cite{van2000asymptotic} where the asymptotic optimality of the MLE and its Fisher information is argued from the local asymptotic minimax perspective. Many other optimality criteria exist and are discussed in Appendix as well. The main takeaway is that, if \eqref{mleinv} is taken as the basis for asymptotic optimality, one should aim to design an experiment that maximizes the Fisher information. To see why, consider two statistical experiments \eqref{statexp} based on different designs $\boldsymbol X$ and $\boldsymbol Y$ (e.g., price trajectory versus total cost) and one wants to predict $c_{\text{vwap}}(\theta)=\int c_{\text{vwap}}(\theta^\star,T,v)G_{\text{order}}(dv,dT)$ based on $\hat\theta_{n,\boldsymbol X}$ or $\hat\theta_{n,\boldsymbol Y}$. Then, if $c_{\text{vwap}}(\theta)$ is differentiable w.r.t $\theta$ with gradient $\nabla_\theta c_{\text{vwap}}$, one can show that, based on \eqref{mleinv} and the delta method \cite{van2000asymptotic}, the asymptotic variance of the plug-in estimator satisfies 
\begin{align*}
    \sqrt{n}\big(c_{\text{vwap}}(\hat\theta_{n,\boldsymbol X})-c_{\text{vwap}}(\theta^\star)\big)\overset{d}{\rightarrow} \text{ } &  \mathcal N \big(0, \nabla_\theta c^T_{\text{vwap}}(\theta^\star)\cdot\mathcal I^{-1}_{\boldsymbol X}(\theta^\star)\cdot \nabla_\theta c_{\text{vwap}}(\theta^\star)\big) \nonumber\\
    \sqrt{n}\big(c_{\text{vwap}}(\hat\theta_{n,\boldsymbol Y})-c_{\text{vwap}}(\theta^\star)\big)\overset{d}{\rightarrow} \text{ } &\mathcal N \big(0, \nabla_\theta c^T_{\text{vwap}}(\theta^\star)\cdot\mathcal I^{-1}_{\boldsymbol Y}(\theta^\star)\cdot \nabla_\theta c_{\text{vwap}}(\theta^\star)\big).
\end{align*}
Now, if $\mathcal I_{\boldsymbol X}(\theta^\star)\succcurlyeq\mathcal I_{\boldsymbol Y}(\theta^\star)$ (i.e., $A\succcurlyeq B$ denotes the matrix $A-B$ is positive semi-definite), it can be shown that $\mathcal I^{-1}_{\boldsymbol X}(\theta^\star)\preccurlyeq\mathcal I^{-1}_{\boldsymbol Y}(\theta^\star)$ and
\begin{equation}\label{asymvar}
    \nabla_\theta c^T_{\text{vwap}}(\theta^\star)\cdot\mathcal I^{-1}_{\boldsymbol X}(\theta^\star)\cdot \nabla c_{\text{vwap}}(\theta^\star)\leq\nabla_\theta c^T_{\text{vwap}}(\theta^\star)\cdot\mathcal I^{-1}_{\boldsymbol Y}(\theta^\star)\cdot \nabla c_{\text{vwap}}(\theta^\star),
\end{equation}
which implies $\hat\theta_{n,\boldsymbol X}$ has lower asymptotic variance for estimating $c_{\text{vwap}}(\theta^\star)$, or any other differentiable function of $\theta$, simply based on its greater Fisher information.

\subsubsection{Sufficiency and Other Lemmas}
In this section, we introduce several technical lemmas we shall use later. Following the previous discussion, one might want to maximize the Fisher information of an experiment \eqref{statexp}. An important related concept is sufficient statistics. Formally, given $\boldsymbol X$ in some statistical experiment \eqref{statexp} and a function $\phi$, one can create another statistical experiment using the \textit{statistics} $\phi(\boldsymbol X)$. This typically incurs a loss of information in terms of the Fisher information unless the statistic is sufficient. Here, a statistic $\phi(\cdot)$ is sufficient for $\{\mathbb P_\theta\}_{\theta\in\Theta}$ if the conditional distribution of $\boldsymbol{S}$ given $\phi(\boldsymbol{S})$ is free of $\theta$, under $\mathbb P_\theta$ for any $\theta\in\Theta$. We also give a formal characterization here which we shall use later. 

\begin{lemma}[Neyman-Fisher Factorization Theorem]\label{factor} Consider the statistical experiment \eqref{statexp} and its sample $\boldsymbol X$. A statistic $\phi(\boldsymbol X)$ is sufficient if and only if the likelihood function in \eqref{statexp} has the factorization
\begin{equation*}
    p(\boldsymbol{X}|\theta)=h(\boldsymbol{X})g(\phi(\boldsymbol{X}),\theta),
\end{equation*}
for some non-negative $h(\cdot), g(\cdot)$.
\end{lemma}

The following lemmas relate a sufficient statistic with the Fisher information.
\begin{lemma}[Data Processing Inequality \cite{zamir1998proof}]\label{data processing}
Consider the statistical experiment \eqref{statexp} and its sample $\boldsymbol X$. Let $\phi(\boldsymbol X)$ be a statistic of data, then the statistical experiment based on $\phi(\boldsymbol X)$ satisfies,for $\theta\in \Theta$,
\begin{equation}\label{fishinequality}
   \mathcal I_{\boldsymbol{X}}(\theta) \succcurlyeq  \mathcal I_{\phi(\boldsymbol{X})}(\theta),
\end{equation}
 with the equality obtained if and only if $\phi(\boldsymbol X)$ is a sufficient statistic for the original statistical experiment.
\end{lemma}
\begin{lemma}[Data Refinement Inequality]\label{data refine} Consider the statistical experiment \eqref{statexp} and its sample $\boldsymbol X$. Let $\boldsymbol U$ and $\boldsymbol V$ be two statistics of the experiments, respective; then the statistical experiment based on them satisfies, for $\theta\in\Theta$,
\begin{align*}
    \mathcal I_{\boldsymbol U,\boldsymbol V}(\theta)\succcurlyeq\mathcal I_{\boldsymbol U}(\theta).
\end{align*}
Moreover, if $\boldsymbol V=\phi(\boldsymbol U)$ for some $\phi(\cdot)$ which is a bijective mapping, then 
\begin{equation*}
   \mathcal I_{\boldsymbol U}(\theta)=\mathcal I_{\boldsymbol V}(\theta). 
\end{equation*}
\end{lemma}

The proofs for Lemmas \ref{data processing} and \ref{data refine} are in the Appendix for completeness. In other words, the Fisher information of the experiment using $\phi(\boldsymbol{X})$ is generally smaller than the one using $\boldsymbol X$, unless statistic $\phi(\boldsymbol{X})$ is sufficient, in which case  $\boldsymbol{X}$ provides no extra information over $\phi(\boldsymbol X)$ for estimating $\theta$ (thus the equality in \eqref{fishinequality}). Note the original sample $\boldsymbol X$ is always, trivially, a sufficient statistic. However, the whole data $\boldsymbol{X}$ is not always required.
\begin{example}
\textup{Consider a simple model under \eqref{main} which is a random walk with drift expressed by $S_t=S_0+\theta f(v) t +\sigma W_t$ for some $f(\cdot)$. Let $\boldsymbol X=(v,T)\cup\boldsymbol S$ be the statistical experiment in \eqref{statexp} where $\boldsymbol S=\{S_{t_i}\}_{i\in[N]}$. It is straightforward to check that the log-likelihood in \eqref{loglike} is given by 
\begin{align*}
   &l(\theta|\boldsymbol{X})\nonumber\\
   =&\log g_{\text{order}}(v,T)-\sum_{i=1}^N\frac{1}{2}\log(2\pi\sigma^2(t_i-t_{i-1}))-\frac{1}{2\sigma^2(t_i-t_{i-1})}\sum_{i=1}^N \Big((S_{t_i}-S_{t_{i-1}})-\theta f(v) (t_i-t_{i-1}))\Big)^2 \nonumber\\
\end{align*}
which can be reduced to $h(\boldsymbol{S},T,v)+g(S_T-S_0, T,v, \theta)$ for some $h,g$ (i.e., isolate the term with $\theta$). Using Theorem \ref{factor}, we can see that $(S_T-S_0,v,T)$ is sufficient. Thus, for this model, price trajectory $\boldsymbol S$ is not needed for the estimation of $\theta$, the difference between the end price and start price would suffice (in fact $S_T$ would suffice since we assume $S_0$ is known)}.
\end{example}

%\begin{remark}\label{rmk-multiple}
%Lemma~\ref{factor} and the example above effectively consider sufficient statistics for the probability distribution of data from a single trajectory. In the setting of Lemma ~\ref{factor}, one can easily extend to the setting where one is given multiple i.i.d. samples $\boldsymbol{X}_1,...,\boldsymbol{X}_n$ from statistical experiment \eqref{statexp}, the following factorization readily holds:
%\begin{equation}\label{factormul}
%p(\boldsymbol{X}_1,\cdots,\boldsymbol{X}_n|\theta)= \prod_{i=1}^n h(\boldsymbol{X}_i) \prod_{i=1}^n g(\phi(\boldsymbol{X}_i),\theta) =: \tilde{h}(\boldsymbol{X}_i,\cdots,\boldsymbol{X}_n)\tilde{g}(\phi(\boldsymbol{X}_i),\cdots,\phi(\boldsymbol{X}_n),\theta).
%\end{equation}
%Consequently, $(\phi(\boldsymbol{X}_i),\cdots,\phi(\boldsymbol{X}_n))$ is sufficient for $\theta$ when given $n$ samples. Leveraging Lemma~\ref{data refine} and Remark~\ref{rmk-multiple}, in the next two sections, we examine the sufficient statistics based on a single price trajectory.
%\end{remark}

In the next two sections, we examine sampling schemes from price trajectory in two popular market impact models: namely the Almgren Chriss Model and the Propagator Model. In particular, given Lemma~\ref{data processing}, which establishes the connection between sufficient statistics and the Fisher information, we are able to find some insight in designing efficient statistical experiments.

\section{The Almgren-Chriss Model}\label{s3}
%In this section, we investigate popular heuristic market impact models from the asymptotic view of statistical estimation. 

The Almgren-Chriss model remains one of the most popular and widely adopted market impact models since its introduction \cite{almgren2001optimal} and the paper \cite{almgren2005direct} is dedicated to parameter estimation problems for the model:
\begin{align}\label{acdynamic}
     S_t=&S_0+S_0(g(v)t+h(v))+S_0\sigma \int_0^t dW_s, \text { when $t\leq T$} \nonumber\\
     S_t=&S_0+S_0g(v)T+S_0\sigma \int_0^t dW_s, \text{ when $t>T$. }
\end{align}
In Almgren-Chriss model, $g$ is the permanent impact function and $h$ is the temporary impact function. In \cite{almgren2005direct}, $h$ and $g$ are parameterized by power law:
\begin{align}\label{power}
    g(v;\theta)=&\gamma v^\alpha\nonumber\\
    h(v;\theta)=&\eta v^\beta,
\end{align}
for $\theta=(\gamma,\eta,\alpha,\beta)$. Although $\alpha$ can generally vary between $[0,1]$, it was argued in \cite{almgren2005direct,huberman2004price} that the linearity of permanent impact function $g$ (i.e., $\alpha=1$) is desirable for a model to preclude price manipulation strategy. Moreover, $\beta=0.5$ corresponds to a plausible square-root temporary impact function for $h$. However, during the fitting process, these are not taken as restrictions. There are several practical modifications in the model estimation section of \cite{almgren2005direct}:
\begin{enumerate}
    \item In order to facilitate the estimation of the permanent impact in the Almgren-Chriss model, the price trajectory $\boldsymbol{S}$ includes an additional price point $S_{T_{\text{post}}}$ where $T_{\text{post}}>T$; this can also be seen from the ``piece-wise" dynamic in \eqref{acdynamic}. \cite{almgren2005direct} indicated that the choice of $T_{\text{post}}$ being 30 minutes after $T$ works well. In this way, in the experiment design \eqref{statexp}, we can fix a constant $\tau_{\text{delay}}$ so that $$T_{\text{post}}= T+\tau_{\text{delay}},$$ for $\tau_{\text{delay}}=0.077$ (i.e., 30 minutes delay in a day with 6.5 trading hours). In such manner, we remove the randomness of choosing $T_{\text{post}}$ given $T$, similar to the way we use fixed $\tau_i$ to choose $t_i=\tau_i T$.
    \item The impact functions $g$ and $h$ are further scaled by volatility $\sigma$. For a single stock, this can be subsumed in the coefficient $\gamma,\eta$ in \eqref{power}. 
    \item A \textit{liquidity factor} in the form of shares outstanding with power law exponent is inserted to characterize the fact that the market impact is not just based on $(T,v)$, but also on the stock's liquidity condition. For a single stock, this can also be subsumed in $\gamma,\eta$ in \eqref{power}.
\end{enumerate}

We will again discuss these modifications in detail in the simulation section where we reproduce and compare with the simulation results in \cite{almgren2005direct}. In summary, these modifications facilitate a cross-sectional description of market impact, although we focus on the single stock analysis in this paper. Finally, for our discussion on \eqref{acdynamic}, aside from assumptions in section \ref{theorystat}, we do not put restrictions on the exact form of $h(\cdot ;\theta)$, $g(\cdot ;\theta)$, i.e., we allow them to take in forms other than the power law parametric form \eqref{power}. With the affine transformation $P_t=\frac{S_t-S_0}{S_0}$, \cite{almgren2005direct} estimated $\theta$ using two quantities termed as \textit{permanent impact} $I$ and \textit{realized impact} $J$:
\begin{align}
    I=&\frac{S_{T_\text{post}}-S_0}{S_0}=P_{T_\text{post}}, \nonumber\\
    J=&\frac{\int_0^T S_{t}dt/T-S_0}{S_0}=\frac{1}{T}\int_0^T P_tdt,
\end{align}
where, under the model \eqref{acdynamic}, conditional on $(v,T)$ the joint distribution of $I$ and $J$ follow a Gaussian one  \eqref{gassexp}:
\begin{align}\label{iandj}
   & \begin{pmatrix}I\\J-I/2\end{pmatrix} \sim\mathcal N
        \begin{pmatrix}
        \mu_{I,J}(\theta,T,v),
          \Sigma_{I,J}(T)
    \end{pmatrix}\nonumber\\
    &\text { where } \mu_{I,J}(\theta,T,v)\triangleq\begin{pmatrix}Tg(v;\theta)\\h(v;\theta)\end{pmatrix},  \Sigma_{I,J}(T)\triangleq\sigma^2\begin{bmatrix}
 T_\text{post}& -\frac{T_\text{post}-T}{2} \\
-\frac{T_\text{post}-T}{2}  & \frac{T_\text{post}}{4}-\frac{T}{6}
\end{bmatrix}.
  \end{align}
 A partial derivation of the above can be found in \cite{almgren2005direct}, but we also include a full derivation in the Appendix for completeness. Then, in the parameter estimating phase (section 4.2 of \cite{almgren2005direct}), \cite{almgren2005direct} uses the metaorder data to fit the normalized residuals of $(I,J-I/2)$ using Gauss-Newton optimization algorithm, since this is a non-linear least square problem for the $\mu$ given by \eqref{power}. This fits exactly in our frame of \eqref{statexp}, as the maximum likelihood estimation \eqref{mlebasic} for the Gaussian experiment based on $I,J$ in \eqref{gassexp} is equivalent to solving the following non-linear least squares problem:
\begin{align}\label{sumofsquare}
    &\max_{\theta\in\Theta}\sum_{i=1}^n l(\theta|I_i, J_i)\ratio\Leftrightarrow \nonumber\\
    &\min_{\theta\in\Theta}\sum_{i=1}^n\Big(\begin{bmatrix}
           I_i \\
           J_i-I_i/2
         \end{bmatrix}-\mu_{I,J}(\theta, T_i, v_i)\Big)^T\Sigma^{-1}_{I,J}(T_i)\Big(\begin{bmatrix}
           I_i \\
           J_i-I_i/2
         \end{bmatrix}-\mu_{I,J}(\theta, T_i, v_i)\Big),
\end{align}
based on $n$ total orders samples $(I_i, J_i)_{i\in[n]}$, executed under $(T_i,v_i)_{i\in[n]}$. This equivalence allows us to connect the parameter estimation method in \cite{almgren2005direct} with the asymptotic theory of statistical estimation in Section \ref{theorystat}. It follows from Lemma \ref{data refine}, Lemma \ref{jacob} and Proposition \ref{fishtvtv}, if the Jacobian of $\mu_{I,J}(\theta, T,v)$ w.r.t $\theta$ is $\mathcal J_{I,J}(\theta, T,v)$, then the Fisher information for experiments based on $I,J$ is
\begin{align}\label{FI-IJ}
    \mathcal I_{I,J}(\theta)=\mathcal I_{I,J-I/2}(\theta)=&\mathbb E_{(v,T)\sim G_{\text{order}}}\bigg[\mathcal J^T_{I,J}(\theta, T,v)\cdot\Sigma_{I,J}^{-1}(T)\cdot\mathcal J_{I,J}(\theta, T,v)\bigg]\nonumber\\
    =&\mathbb E_{(v,T)\sim G_{\text{order}}}\bigg[\mathcal J^T_{I,J}(\theta, T,v)\cdot\sigma^{-2}\begin{bmatrix}
 \frac{1}{T}\frac{T_\text{post}/4-T/6}{T_\text{post}/3-T/4}& \frac{1}{T}\frac{T_\text{post}/2-T/2}{T_\text{post}/3-T/4} \\
\frac{1}{T}\frac{T_\text{post}/2-T/2}{T_\text{post}/3-T/4}  &  \frac{1}{T}\frac{T_\text{post}}{T_\text{post}/3-T/4}
\end{bmatrix}\cdot\mathcal J_{I,J}(\theta, T,v)\bigg].
\end{align}
\begin{example}\label{almgrenfit}
 For fitting the power law model \eqref{power} in Almgren-Chriss model with $\theta=(\gamma,\eta,\alpha,\beta)$, one can calculate $$\mathcal J_{I,J}(\theta, T,v)=\begin{bmatrix}
 Tv^\alpha& 0 &T\gamma v^{\alpha}\ln(v)&0\\
0  &  v^\beta&0&\eta v^{\beta}\ln(v)
\end{bmatrix},$$ which means that, based on \eqref{FI-IJ}, $\mathcal I_{I,J}(\theta)\in\mathbb R^{4\times4}$ does not have full rank (hence singular) if $(v,T)\sim \delta_{(v_0,T_0)}$ (i.e., a point mass). In such case, the optimization in \eqref{sumofsquare} could also be under-determined. However, if one only tries to determine the exponent of the power law, then $\theta=(\alpha,\beta)$ and $\mathcal J_{I,J}(\theta, T,v)=\begin{bmatrix}
  T\gamma_0v^{\alpha}\ln(v)&0\\
0&\eta_0v^{\beta}\ln(v)
\end{bmatrix}$ for some given $(\gamma_0,\eta_0)$. In summary, the non-singularity of $ \mathcal I_{I,J}(\theta)$ largely depends on the variability of $(v,T)\sim G_{\text{order}}$ for multi-dimensional $\theta$.
\end{example}

In a case-by-case fashion, we prove the following in the appendix, so that our results from previous sections can be applied. 
\begin{lemma}[Regularity Check]\label{powerregular}
    Consider the statistical experiment \eqref{statexp} based on Almgren-Chriss model with power law \eqref{power} impact and $I,J$ in \eqref{iandj}. %If the assumption \ref{modelspec}-\ref{orderindepG} is satisfied, then assumption \ref{ras3}-\ref{nonsigfishh} is satisfied if the marginal distribution of $G_{\text{order}}$ on $v$ has a support with infinite cardinality.
    {Under Assumptions~\ref{modelspec} and~\ref{orderindepG}, Assumptions \ref{ras3} and \ref{nonsigfishh} are satisfied provided that the marginal distribution of $G_{\text{order}}$ on $v$ has a support with infinite cardinality.}
\end{lemma}
%In particular, we discuss a sufficient condition for variability of $(v,T)\sim G_{\text{order}}$ to make sure assumptions 3.5 and 4. We note that, the condition on $G_{\text{order}}$ is satisfied for any distribution with a continuous density. We also note that this condition is presented for its simplicity, as it is sufficient but not necessary. 
{Lemma~\ref{powerregular} provides a sufficient condition for Assumptions 3 (3.5 in particular) and 4 to hold, although it is not necessary. Note that the condition on $G_{\text{order}}$ is satisfied for any distribution with a continuous density. } By checking the proof, one can also establish the result for discrete distribution $G_{\text{order}}$ with carefully chosen support whose cardinality is larger than the dimension of $\theta$.

\subsection{Sufficient Statistic from a Discrete Price Trajectory }\label{sus}

From Lemma \ref{powerregular}, in view of Section \ref{theorystat} and the regularity check of experiment based on \eqref{power}, we want to see if we can design experiments with Fisher information larger than $\mathcal I_{I,J}$. We first investigate this for the canonical price trajectory data of the form $\boldsymbol S_{\text{full}}=\{S_{t_i}\}_{i\in[N+1]}$ on a given grid $\{\Delta t, 2\Delta t,..., N\Delta t=T, T_{\text{post}}\}$ with $\Delta t=T/N$ and $N$ is a fixed number. Notably, this is also the setup for \cite{almgren2001optimal}, when they consider the optimal execution problem with $N$-trading period.

In this example, again taking the transform $P_t=\frac{S_t-S_0}{S_0}$, from \eqref{acdynamic}, we can see that the likelihood for the experiment based on $\boldsymbol{S}_{\text{full}}$, conditional on $(v,T)$, under Almgren-Chriss model is:
\begin{align}\label{likeli}
    &p_\theta(\boldsymbol{S}_{\text{full}}|v,T)\nonumber\\
    \propto& \exp\bigg(-\frac{(P_{t_1}-g(v;\theta)\Delta t-h(v;\theta))^2}{2\sigma^2\Delta t}-\Big(\sum_{i=2}^N\frac{(P_{t_i}-P_{t_{i-1}}-g(v;\theta)\Delta t)^2}{2\sigma^2\Delta t}\Big)-\frac{(P_{T_{\text{post}}}-P_T+h(v;\theta))^2}{2\sigma^2(T_{\text{post}}-T)}\bigg) \nonumber\\
    \propto & f_1(\boldsymbol{S}_{\text{full}},T,v)f_2\Big(g(v;\theta)P_T+h(v;\theta)(\frac{P_{t_1}}{\Delta t}-\frac{P_{T_{\text{post}}}-P_T}{T_{\text{post}-T}}), \theta\Big),
\end{align}
for some $f_1, f_2$. Using Lemma \ref{factor}, we can see that $(P_{t_1}, P_T, P_{T_{\text{post}}})$ is a sufficient statistic for $\theta$. In fact, $(P_{t_1}, P_T, P_{T_{\text{post}}})$ is sufficient, but not \textit{minimal sufficient}, as $(P_T,\frac{P_{t_1}}{\Delta t}-\frac{P_{T_{\text{post}}}-P_T}{T_{\text{post}-T}})$, a non-invertible mapping of $(P_{t_1}, P_T, P_{T_{\text{post}}})$, is also sufficient. For the definition of \textit{minimal sufficient}, see e.g. \cite{keener2010theoretical} or Theorem \ref{minisuff} below. However, the observation $\frac{P_{t_1}}{\Delta t}-\frac{P_{T_{\text{post}}}-P_T}{T_{\text{post}-T}}$ is somewhat artificial to this example. Consequently, from Lemma \ref{data processing} we have:
\begin{equation}\label{333}
    \mathcal I_{\boldsymbol{S}_{\text{full}}}(\theta)=\mathcal I_{P_{t_1}, P_T, P_{T_\text{post}}}(\theta).
\end{equation}
Although the construction of $t_1$ is artificial and related to the choice of $N$. However, since the choice of $N$ can be arbitrarily large, it is reasonable to suspect that the insight from \eqref{333} can be translated to generic price trajectory data $\boldsymbol S$. Moreover, since the post-order price $I=P_{T_\text{post}}$, can we actually then substitute the VWAP cost $J$ for two price points $P_{t}$ (for some $t$ in ``early" stages) and $P_{T}$ along the price trajectory data for a more asymptotically efficient estimate? Indeed, in the next section, we shall see these insights in several important ways.

\subsection{Sampling Strategy along the Price Trajectory}\label{samplestrat}
Based on the previous insight, given a experiment design based on partial trajectory data $\boldsymbol S=\{P_{t_i}\}_{i\in[N]} \cup \{P_{T_{\text{post}}}\}$ with $t_N=T$, one might ask:
\begin{enumerate}
    \item Does sampling more than three points along the trajectory (i.e., $N>2$) increase the Fisher information? 
    \item If sampling more than three points does not improve the asymptotic efficiency, is simply sampling two points $P_T$ and $P_{T_{\text{post}}}$ good enough?
    \item For sampling three points, in addition to $P_T$ and $P_{T_{\text{post}}}$, how should we pick the extra point $P_t$ for some $t\in(0,T)$?
\end{enumerate}
To answer these questions, suppose we have the following grid of observations $$\boldsymbol{S}=\{P_{t_1},P_{t_2},..., P_{t_{N-1}}, P_T, P_{T_{\text{post}}}\}.$$
Then writing out the density $p_\theta(\boldsymbol{S}|v,T)$ as in \eqref{likeli}, we have 
\begin{equation}\label{suff}
     p_\theta(\boldsymbol{S}|v,T)\propto f_1(\boldsymbol{S})f_2\Big(g(v;\theta)P_T+h(v;\theta)(\frac{P_{t_1}}{t_1}-\frac{P_{T_{\text{post}}}-P_T}{T_{\text{post}-T}}), \theta\Big)
\end{equation}
for some $f_1,f_2$. Following the same reasoning as in the previous section, we see that 
\begin{equation*}
 \mathcal I_{\boldsymbol{S}}(\theta)=\mathcal I_{P_{t_1},P_T,P_{T_{\text{post}}}}(\theta),
\end{equation*}
where $t_1=\min\{t : P_t\in\boldsymbol{S}\}$.
Thus, we have the following corollary.
\begin{coro}\label{coromain}
    For the Almgren-Chirss model and experiment based on $\boldsymbol S=\{P_{t_i}\}_{i\in[N]} \cup \{P_{T_{\text{post}}}\}$, sampling more than three points on the trajectory does not increase the Fisher information, as:
    \begin{equation}\label{3point}
 \mathcal I_{\boldsymbol{S}}(\theta)=\mathcal I_{P_{t_1},P_T,P_{T_{\text{post}}}}(\theta).
\end{equation}
\end{coro}
This answers the first question. To answer the second question, one can calculate (the derivation is left in the Appendix):
\begin{equation}\label{twops}
    \mathcal I_{P_T,P_{T_\text{post}}}(\theta)=\mathbb E_{(v,T)\sim G_{\text{order}}}\bigg[\mathcal J^T_{I,J}(\theta, T,v)\cdot\sigma^{-2}\begin{bmatrix}
 \frac{1}{T}& \frac{1}{T} \\
\frac{1}{T} &  \frac{T_\text{post}}{T}\frac{1}{(T_\text{post}-T)}
\end{bmatrix}\cdot\mathcal J_{I,J}(\theta, T,v)\bigg].
\end{equation}
Comparing \eqref{twops} with \eqref{FI-IJ}, we note that, in general, neither $I_{P_T,P_{T_\text{post}}}(\theta)\succcurlyeq I_{I,J}(\theta)$ nor $ I_{I,J}(\theta)\succcurlyeq I_{P_T,P_{T_\text{post}}}(\theta)$ can be established. Thus, one cannot definitively say sampling two points is more efficient than using VWAP cost $J$ and post-order price $I$.

So far, we have shown that sampling two price points $P_T$ and $P_{T_\text{post}}$  is not optimal while sampling more than one point along price trajectory  $\{P_t\}_{t\in(0,T)}$ is not necessary. This naturally brings us to the third question above. We summarize the answer as our first main theorem.
\begin{theorem}\label{acfinal}
Under assumptions \ref{modelspec}-\ref{nonsigfishh}, we sample $(P_{t}, P_T,P_{T_\text{post}})$, as long as $t$ is chosen early enough so that 
\begin{equation*}
    \frac{t}{T} \leq  \frac{1}{4},
\end{equation*}
then, under the Almgren-Chriss model \eqref{acdynamic}, we have
\begin{equation}\label{thmmain}
    \mathcal I_{P_{t},P_T,P_{T_{\text{post}}}}(\theta) \succcurlyeq \mathcal I_{I,J}(\theta).
\end{equation}
Moreover, \eqref{thmmain} holds for generic forms of $g$ and $h$ and all distributions $G_{\text{order}}$ of $(v,T)$.
\end{theorem}
\begin{proof}
The proof of Theorem \ref{acfinal} is left in the Appendix \ref{proofofthmmain}.
\end{proof}

\begin{remark}
Notably, also shown in the proof, the ratio $\frac{1}{4}$ in Theorem \ref{acfinal} is somewhat tight, in the sense that for $\frac{t}{T}>\frac{1}{4}$, there might exist some form of $h$ and $g$ for which \eqref{thmmain} is no longer true. This could be of interest in its own right, as the significance of number $\frac{1}{4}$ is not immediately clear from observing \eqref{acdynamic}.
\end{remark}

 Theorem \ref{acfinal} suggests that, at least for the Almgren-Chriss dynamics \eqref{acdynamic}, as long as we sample $P_t$ early enough (e.g., any price sampled within the first 25 percent of the filled order), then MLE estimated using $(P_{t}, P_T,P_{T_\text{post}})$ would outperform the estimation method using $I,J$ in \cite{almgren2005direct} asymptotically, for any distributions of $(v,T)$, rendering Theorem \ref{acfinal} applicable for metaorders across different magnitudes and for generic form of $h$ and $g$ (specifically the original power law models in \cite{almgren2005direct}). Perhaps interestingly, this result seems to suggest earlier stages of an order are more informative for calibrating impact models. Finally, we use the data based on Section 4.1, 4.2 and Table 3 of \cite{almgren2005direct} to provide an example of the level of improvement based on Theorem \ref{acfinal}. More extensive simulation studies verifying Theorem \eqref{acfinal} are provided at the end of the section.

\begin{example}
\textup{Consider estimating the power law exponent of \eqref{power} in \eqref{acdynamic}. Suppose the underlying true parameters are $\theta^\star=(\gamma^\star,\eta^\star,\alpha^\star,\beta^\star)=(0.314, 0.142, 0.891, 0.600)$. We set the hyper-parameters of the metaorders to be $(X,v,T,T_{\text{post}},\sigma)=(0.1, 0.5, 0.2, 0.275, 0.0157)$. The data is from \cite{almgren2005direct}, which can be viewed as one specific instance from $G_{\text{order}}$, i.e., a point mass. We set $t=0.1\cdot T$ and $(\gamma^\star, \eta^\star)=(0.314, 0.142)$, so we only estimate $(\alpha,\beta)$. This avoids non-singular issues from $G_{\text{order}}$ being a point mass. Then, one can calculate $\mathcal I_{P_t,P_T,P_{T_{\text{post}}}}(\alpha^\star,\beta^\star)$ and $\mathcal I_{I,J}(\alpha^\star, \beta^\star)$, which leads to
\begin{equation}\label{exminv}
    \mathcal I^{-1}_{P_t,P_T,P_{T_{\text{post}}}}(\alpha^\star,\beta^\star)=10^4\begin{bmatrix}
 3.504 \cdot 10^4& -2.259 \cdot 10^2 \\
-2.259 \cdot 10^2  & 1.845 \cdot 10
\end{bmatrix},
\mathcal I^{-1}_{I,J}(\alpha^\star,\beta^\star)=10^4\begin{bmatrix}
 4.438 \cdot 10^4& -4.942 \cdot 10^2 \\
-4.942 \cdot 10^2  & 3.811 \cdot 10
\end{bmatrix}.
\end{equation}
The closed-form expression of $\mathcal I_{P_t,P_T,P_{T_{\text{post}}}}(\theta)$ is calculated in Appendix \ref{proofofthmmain}.
In view of \eqref{mleinv} and \eqref{asymvar}, using $(P_t,P_T,P_{T_{\text{post}}})$ over $(I,J)$ implies an asymptotic variance reduction (i.e., equivalent to relative efficiency, \cite{van2000asymptotic}) of 21\% for estimating $\alpha$ ( calculated by $\frac{4.438-3.504}{4.438}$, or equivalently a 21\% increase in sample-efficiency), 51\% for estimating $\beta$ and 18.5\% for estimating the average VWAP cost in bps, given by $c_{\text{vwap}}(\alpha,\beta)=\frac{T}{2}\gamma v^{\alpha}+\eta v^{\beta}$ and $\nabla_{\alpha,\beta} c_{\text{vwap}}(\alpha^\star, \beta^\star)=[\frac{T}{2}\gamma v^{\alpha}\log(\alpha), \eta v^{\beta}\log(\beta)]$. On the other hand, also following Remark \ref{almgrenfit}, given $(\alpha^\star, \beta^\star)=(0.891, 0.600)$, suppose ones wants to estimate $(\gamma,\eta)$. Similarly one can show a 20.6\% variance reduction for estimating $\gamma$, 51.5\% for $\eta$. In the context of  \cite{schied2009risk}, a model of the form \eqref{power} with $\alpha^\star=\beta^\star=1$ and a utility $u(R)=-\exp(-AR)$ for some $A>0$ is assumed (e.g., we let $A=5$ below). Then the optimal adaptive liquidation strategy (notice this is a sell-program so $x_0=X$) is shown to be of $x_t=X\exp(-t\sqrt{\frac{\sigma^2 A}{2\eta}})$. One can check if $\eta$ is estimated around $\eta^\star$ with a 10\% error, then the optimal liquidation at $T$ would be off-target by around 6.88\%, whereas a 5\% estimated error on $\eta$ would reduce that error to around 3.36\%, an improvement close to 50\%.}
\end{example}

\subsection{Non-VWAP Executions }\label{nonvwappp}

In this section, we use the framework from \cite{almgren2001optimal} to briefly discuss the estimation for non-VWAP orders under \eqref{acdynamic}. In particular, given full grid observation $\boldsymbol{S}_{\text{full}}=\{S_{t_i}\}_{i\in [N]}$ and a sequence of trading rate $\{v_i\}_{i\in[N]}$ so that a trading rate of $v_i$ is executed during interval $(t_{i-1}, t_i]$ and  $\frac{T}{N}\sum_{i=1}^N v_i=X$. We assume trading rates and trading period is all fixed and we only investigate the sufficiency issue here. Thus, under \eqref{acdynamic} and $P_t=\frac{S_t-S_0}{S_0}$ we have 
\begin{equation}
    P_{t_i}=P_{t_{i-1}}+ g(v_i; \theta)\Delta t+h(v_i; \theta)-h(v_{i-1};\theta)+\sigma \int_{t_{i-1}}^{t_i}dW_s.
\end{equation}
Note that $\{v_i\}_{i\in[N]}$ needs not to have $N$ distinct values (i.e., $|\{v_i\}_i|< N$) as the same trading rate is allowed to prevail over several periods of $\Delta t$. To account for such strategy, suppose there are $N'$ distinct ($N'\leq N$) constant trading rates: $\{v'_{i'}\}_{i'\in[N']}$ with $v'_{i'}\neq v'_{j'}$ for all $i'\neq j'$ and each $v'_{i'}$ corresponds to some union of trading intervals where the trading rate is $v'_{i'}$. We represent such a union by a sequence of disjoint, non-adjacent intervals $I_{i'}=\cup_k (t^{s,i'}_k, t^{e,i'}_k]$. For example, if $v'_1=0.5$ and the order is executed with a trading rate of 0.5 during $(t_0,t_1]$, $(t_1,t_2]$ and $(t_5,t_6]$, then $I_{1}=(t_0,t_2] \cup (t_5, t_6]$. Note that the disjoint and non-adjacent property makes sure that 1) the representation of $I_{i'}$ is unique and 2) for any $t$ being $t_k^{s,i'}$ or $t_k^{e,i'}$ in the representation of $I_{i'}$, then 
\begin{enumerate}
    \item either $t=0$ or $t=T$
    \item or the interval $(t-\Delta t, t]$ and $(t, t+\Delta]$ has two different trading rates.
\end{enumerate}
In other words, if we view the time before $t=0$ and after $t=T$ as having 0 trading rate, then any $t$ being $t_k^{s,i'}$ or $t_k^{e,i'}$ in the representation of $I_{i'}$ implies $t$ connects two periods with different trading rates. In this setting, the likelihood can then be written as
\begin{align}\label{nonvwap}
    &p_\theta(\boldsymbol S_{\text{full}}|v,T) \nonumber\\
    \propto& \exp\Bigg(-\sum_{i=1}^N\frac{\Big((P_{t_i}-P_{t_{i-1}})-(g(v_i; \theta)\Delta t+h(v_i;\theta)-h(v_{i-1};\theta))\Big)^2}{2\sigma^2\Delta t}-\frac{(P_{T_{\text{post}}}-P_T+h(v_N;\theta))^2}{2\sigma^2(T_{\text{post}}-T)}\Bigg)\nonumber\\
    =&\exp\Bigg(-\sum_{i'\in[N']}\sum_{\{i: v_i=v'_{i'}\}}\frac{\Big((P_{t_i}-P_{t_{i-1}})-(g(v_i; \theta)\Delta t+h(v_i;\theta)-h(v_{i-1};\theta))\Big)^2}{2\sigma^2\Delta t}-\frac{(P_{T_{\text{post}}}-P_T+h(v_N;\theta))^2}{2\sigma^2(T_{\text{post}}-T)}\Bigg)\nonumber\\
    =&f_1 (\boldsymbol{S}_{\text{full}})\cdot\prod_{i'\in[N']}f^g\Big(\sum_{(t_k^{s,i'},t_k^{e,i'}]\in I_{i'}} (P_{t_k^{e,i'}}-P_{t_k^{s,i'}}) ;g(v'_{i'};\theta)\Big)\cdot\prod_{i'\in[N']}f^h\Big(\sum_{(t_k^{s,i'},t_k^{e,i'}]\in I_{i'}} (P_{t_k^{s,i'}+\Delta t}-P_{t_k^{s,i'}})\nonumber\\
    &+\sum_{(t_k^{s,i'},t_k^{e,i'}]\in I_{i'}} (P_{t_k^{e,i'}}-P_{t_k^{e,i'}+\Delta t}){\mathbf 1_{\{t_k^{e,i'}<T\}}}+(P_{T}-P_{T_{\text{post}}})\frac{\Delta t}{T_{\text{post}}-T}\mathbf 1_{\{t_k^{e,i'}=T\}} ;h(v'_{i'};\theta)\Big)\cdot f_2(\theta).
    \end{align}
    for some $f_1, f_2, f^g, f^h$. Here an indicator function $\mathbf 1_{A}(x)=1$ if $x\in A$ and 0 otherwise. In view of Theorem \ref{factor}, a sufficient statistic for \eqref{nonvwap} is $\{P_t: t=T_{\text{post}} \text{ or } t\in\{t_k^{s,i'},t_k^{s,i'}+\Delta t,t_k^{e,i'},t_k^{e,i'}+\Delta t \} \text{ for some } i'\in[N'], k\in[|I_{i'}|]\}.$ In other words,  the union of $P_t$'s that is either a \textit{node}, which connects two different trading rate periods, or immediately follows such a \textit{node}, constitutes a sufficient statistic of $\boldsymbol S_{\text{full}}$. This seems to suggest the price points more closely following trading rate changes are more important to the estimation. Using this characterization, we can also recover the sufficiency $P_{t_1}, P_{T}$ and $P_{T_{\text{post}}}$ in the VWAP case ($P_0=0$ is known). 

\subsection{Simulation} 
In this section, we verify the results from corollary \ref{coromain} and theorem \ref{acfinal}. For a fair comparison, we recover, to the best of our abilities, the IBM stock simulation result in Table 3 from \cite{almgren2005direct}. Several practical modifications in \cite{almgren2005direct} are carried out, so impact equations in (7) and (8) of \cite{almgren2005direct} become: 

\begin{equation*}
    \begin{aligned}
        & & \frac{I}{\sigma} & = \gamma_A T \cdot \text{sgn}(X_A) \left| \frac{X_A}{V_AT} \right|^{\alpha} \left( \frac{\Theta_A}{V_A} \right)^{\delta_A} + \langle\text{noise}\rangle \\
        & & \frac{1}{\sigma}(J - \frac{I}{2}) & = \eta_A \cdot \text{sgn}(X) \left| \frac{X_A}{V_AT} \right|^{\beta} + \langle\text{noise}\rangle 
    \end{aligned}
\end{equation*}
whereas our equation based on \eqref{iandj} gives
\begin{equation} \label{eq:fengpei}
    \begin{aligned}
        & & I & = \gamma T \cdot v^\alpha  + \langle\text{noise}\rangle \\
        & & J - \frac{I}{2} & = \eta \cdot v^\beta + \langle\text{noise}\rangle.
    \end{aligned}
\end{equation}
Here $\Theta_A$ is called \textit{shares outstanding} , $V_A$ is \textit{average daily volume} and $\frac{\Theta_A}{V_A}$ is called \textit{inverse turnover}. All concepts are in \cite{almgren2005direct} to model \textit{liquidity} factor into the market impact, which is not necessary for a single stock analysis. To avoid confusion, the parameters from Table 3 in \cite{almgren2005direct} which are different from ours are subscripted by $A$. Their transformation with our model parameters can be constructed as
\begin{equation*}
    v = \left| \frac{X_A}{V_AT}\right|, \gamma = \sigma \gamma_A \left( \frac{\Theta_A}{V_A} \right)^{\delta_A}, \eta = \sigma \eta_A
\end{equation*}
Following Table 3 of \cite{almgren2005direct}, the initial hyper-parameters for IBM stocks are given as 
\begin{equation*}
    V = 6,561,000, \Theta = 1,728,000,000, \sigma = 0.0157, T_{\text{post}} - T = 0.5 / 6.5 = 0.077 \text{ and } T = 0.2.
\end{equation*}
 With these transformations, we can reproduce values in Table 3 in Almgren for IBM using equation (\ref{eq:fengpei}) and our version of the parameters, up to a negligible error (i.e., no discrepancy up to the forth digit). The results are summarized in Table \ref{Tab:Tcr}.
\\

    \begin{minipage}[b]{0.4\textwidth}
    \centering
\scalebox{0.8}{\begin{tabular}[t]{lc|ccc}
		\toprule
		&   & \textbf{IBM}     \\
		\midrule
		Average Daily Volume & $V$ & 6.561M  \\
		Shares outstanding & $\Theta$ &   1.728B  \\
		Inverse turnover & $\Theta / V$ &   263.37  \\
		Daily volatility (\%) & $\sigma$ &  1.57\% \\
		Normalized trade size & $X/V$ &  0.1  \\
		\midrule
		Normalized permanent & $I/\sigma$ &   0.1265 \\
		Permanent price impact (bp) & $I$ &   19.86  \\
		\midrule
		Trade duration (days) & $T$ &  0.2 \\
		Normalized temporary & $K/\sigma$ &  0.0934 \\
		Temporary Impact cost (bp) & $K$ &  14.67 \\
		\midrule
		Realized cost (bp) & $J$ & 24.60\\
		\bottomrule
	\end{tabular}}

\captionof{table}{\small{Table 3 recreation from \cite{almgren2005direct}.}}
\label{Tab:Tcr}
\end{minipage}
\begin{minipage}[b]{0.46\textwidth}
    \centering
\includegraphics[width=0.9\textwidth]{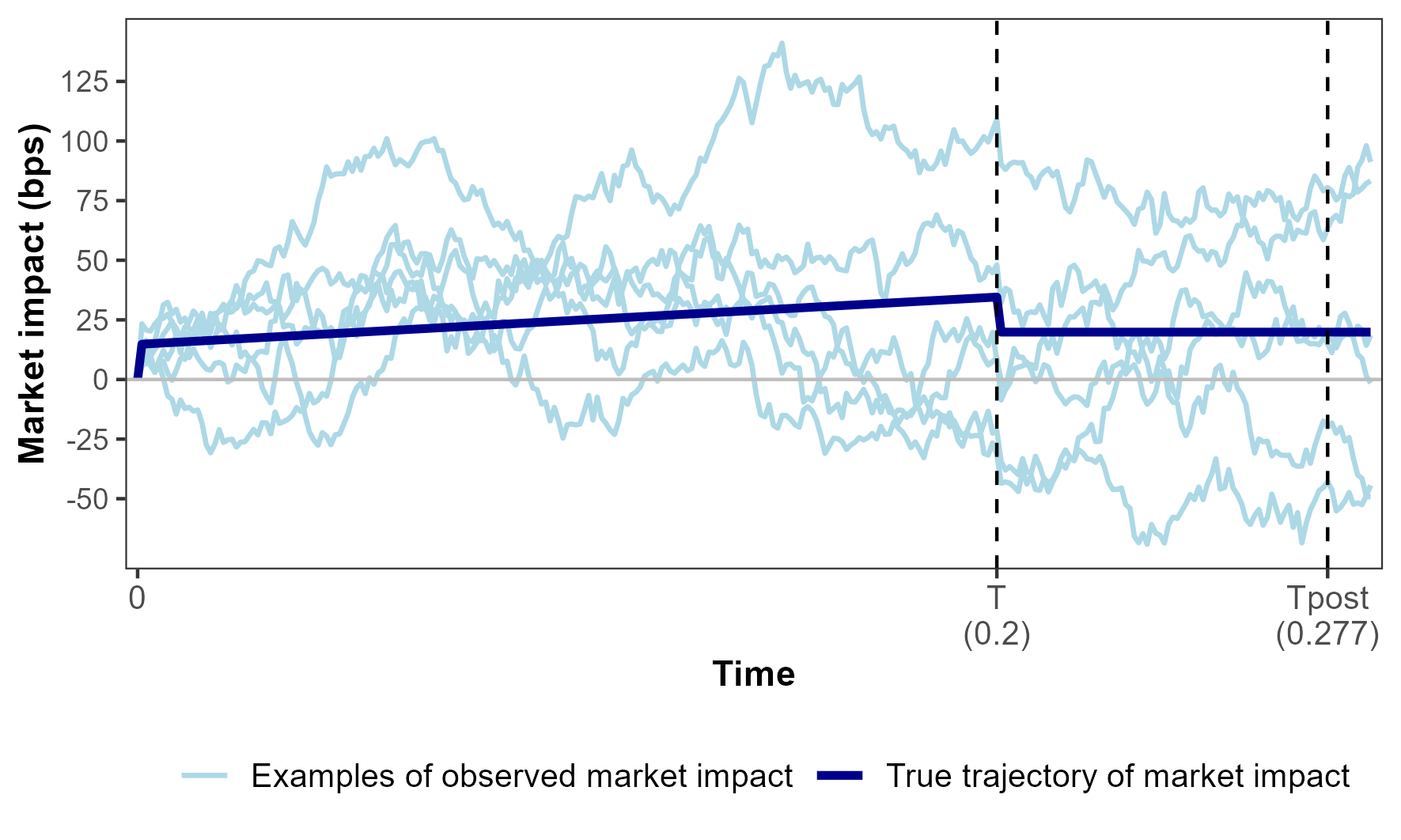}
\captionof{figure}{\small{Simulate Price Trajectory for Table \ref{Tab:Tcr}. Blue line represents expected average.}}
    \end{minipage}

For simulation, we let $X_A/V_A$ follow a uniform distribution on an interval from 0.05 to 0.15 with $T=0.2$ fixed. This would imply our parameter $v = \frac{X_A}{V_AT}$ follows a uniform distribution from 0.25 to 0.75. This specifies $G_{\text{order}}$. The ground truth parameters are set following Table 1 as $\alpha^\star = 1, \beta^\star = 0.6, \gamma^\star = \sigma \gamma_A \left( \frac{\Theta}{V} \right)^{\delta_A} = 0.314 \cdot 0.0157 \cdot 4.0285 = 0.01986, \eta^\star = \sigma \eta_A = 0.142 \cdot 0.0157 = 0.002229$. In the simulation, we fit all four parameters $\theta=(\gamma,\eta,\alpha,\beta)$ in \eqref{power}. Each path is generated by the Euler–Maruyama method with $\Delta = 0.001T$ and the stochastic integral in $J$ is approximated by the discrete sum. We consider five experiments to verify the findings in the previous sections:
\begin{enumerate}
    \item Using $(I,J)$, referred to as ``Almgren approach $(I,J)$"
    \item Using $(P_T, P_{T_{\text{post}}})$, referred to as ``Two points $(T,T_{\text{post}})$"
    \item Using $(P_{0.1T}, P_T, P_{T_{\text{post}}})$, referred to as ``Three points $(0.1T, T,T_{\text{post}})$", with $\frac{t}{T}=0.1\leq 0.25$.
    \item Using $(P_{0.5T}, P_T, P_{T_{\text{post}}})$, referred to as ``Three points $(0.5T, T,T_{\text{post}})$", with $\frac{t}{T}=0.5\geq 0.25$
    \item Using $(P_{0.1T}, P_{0.5T}, P_T, P_{T_{\text{post}}})$, referred to as ``Four points $(0.1T, 0.5T, T,T_{\text{post}})$".
\end{enumerate}
The results are presented in Table \ref{Tab:Tcr2222}. For each method, we report the experiment results for metaorder size $n=500$, $n=1000$ and $n=10000$. For each sample size $n$, we repeat the experiment 1000 times so that we can report 1) average estimates of each parameter denoted by ``\textbf{Avg estimate} $\hat\theta$, across 1000 simulations,  2) the theoretical standard deviation of each parameter, based on empirical Fisher information, denoted by ``\textbf{Theoretical} SD($\hat\theta$)" and 3) the empirical standard deviation of each parameter, based on bootstrap from 1000 simulations, denoted by ``\textbf{Empirical} SD($\hat\theta$)". The best method, for each $n$, is recorded in bold. As we can see, consistent with the results in section \ref{samplestrat}: 

\begin{enumerate}
\item The ``Three points method" with $t/T\leq 0.25$ outperform the Almgren approach, especially in terms of standard deviation, as indicated by theorem \ref{acfinal}.

\item The ``Two points method" and ``Three points method" with $t/T\geq 0.25$ do not visibly outperform the Almgren approach and 

\item The ``Three points method" perform almost exactly the same as the ``Four points method", as indicated by corollary \ref{coromain}.

\end{enumerate}

Finally, the \textbf{Theoretical} SD($\hat\theta$) provides a good estimate for \textbf{Empirical} SD($\hat\theta$), especially for large $n$, which can be useful for inference.
 
\begin{center}
\scalebox{0.65}{\begin{tabular}[t]{l|rrrr|rrrr|rrrr|}
 \toprule
	\multicolumn{1}{c}{\textbf{ }} & \multicolumn{4}{c}{\textbf{Avg estimate $\hat{\theta}$}} & \multicolumn{4}{c}{\textbf{Theoretical $SD(\hat{\theta})$}} & \multicolumn{4}{c}{\textbf{Empirical $SD(\hat{\theta})$}} \\
	\multicolumn{1}{c}{ } & \multicolumn{4}{c}{{\tiny (average over 1000 simulations)}} & \multicolumn{4}{c}{{\tiny (avg. of hessian implied SD)}} & \multicolumn{4}{c}{{\tiny (SD of estimate over 1000 sim)}} \\
	Method & $\alpha$ & $\beta$ & $\gamma$ & $\eta$ & $\alpha$ & $\beta$ & $\gamma$ & $\eta$ & $\alpha$ & $\beta$ & $\gamma$ & $\eta$\\
	\midrule
	\addlinespace[0.3em]
	\multicolumn{13}{l}{\textbf{$\bm{n=500}$}}\\
	\hspace{1em}Almgren apporach $(I, J)$ & 1.0641 & 0.6209 & 0.022222 & 0.002285 & 0.7777 & 0.3327 & 0.010810 & 0.000520 & 0.7825 & 0.3337 & 0.011847 & 0.000533\\
	\hspace{1em}Two points: $(T, T_{\text{post}})$ & 1.0694 & 0.6334 & 0.022755 & 0.002356 & 0.7752 & 0.4935 & 0.012370 & 0.000784 & 0.8197 & 0.4814 & 0.023695 & 0.000781\\
	\hspace{1em}Three points: $(0.1T, T, T_{\text{post}})$ & \textBF{1.0464} & \textBF{0.6161} & \textBF{0.021577} & \textBF{0.002267} & \textBF{0.6835} & \textBF{0.2282} & \textBF{0.009126} & \textBF{0.000355} & \textBF{0.6891} & \textBF{0.2235} & \textBF{0.010007} & \textBF{0.000357}\\
	\hspace{1em}Three points: $(0.5T, T, T_{\text{post}})$ & 1.0609 & 0.6248 & 0.022008 & 0.002320 & 0.7449 & 0.4164 & 0.010123 & 0.000656 & 0.7420 & 0.4160 & 0.010897 & 0.000669\\
	\hspace{1em}Four Point: $(0.1T, 0.5T, T, T_{\text{post}})$ & \textBF{1.0464} & \textBF{0.6161} & \textBF{0.021577} & \textBF{0.002267} & \textBF{0.6835} & \textBF{0.2282} & \textBF{0.009126} & \textBF{0.000355} & \textBF{0.6891} & \textBF{0.2235} & \textBF{0.010007} & \textBF{0.000357}\\
	\addlinespace[0.3em]
	\multicolumn{13}{l}{\textbf{$\bm{n=1000}$}}\\
	\hspace{1em}Almgren apporach $(I, J)$ & 1.0294 & 0.6215 & 0.020860 & 0.002273 & 0.5240 & 0.2324 & 0.006782 & 0.000362 & 0.5120 & 0.2316 & 0.006791 & 0.000367\\
	\hspace{1em}Two points: $(T, T_{\text{post}})$ & 1.0286 & 0.6257 & 0.020847 & 0.002312 & 0.5198 & 0.3401 & 0.006719 & 0.000537 & 0.5095 & 0.3329 & 0.006757 & 0.000537\\
	\hspace{1em}Three points: $(0.1T, T, T_{\text{post}})$ & \textBF{1.0372} & \textBF{0.6021} & \textBF{0.020831} & \textBF{0.002245} & \textBF{0.4628} & \textBF{0.1594} & \textBF{0.005990} & \textBF{0.000247} & \textBF{0.4676} & \textBF{0.1561} & \textBF{0.006152} & \textBF{0.000245}\\
	\hspace{1em}Three points: $(0.5T, T, T_{\text{post}})$ & 1.0249 & 0.6262 & 0.020774 & 0.002296 & 0.4995 & 0.2891 & 0.006450 & 0.000454 & 0.4948 & 0.2844 & 0.006514 & 0.000457\\
	\hspace{1em}Four Point: $(0.1T, 0.5T, T, T_{\text{post}})$ & \textBF{1.0372} & \textBF{0.6021} & \textBF{0.020831} & \textBF{0.002245} & \textBF{0.4628} & \textBF{0.1594} & \textBF{0.005990} & \textBF{0.000247} & \textBF{0.4676} & \textBF{0.1561} & \textBF{0.006152} & \textBF{0.000245}\\
	\addlinespace[0.3em]
	\multicolumn{13}{l}{\textbf{$\bm{n=10000}$}}\\
	\hspace{1em}Almgren apporach $(I, J)$ & 1.0111 & 0.5994 & 0.020039 & 0.002224 & 0.1604 & 0.0728 & 0.002020 & 0.000112 & 0.1634 & 0.0734 & 0.002035 & 0.000114\\
	\hspace{1em}Two points: $(T, T_{\text{post}})$ & 1.0110 & 0.6006 & 0.020037 & 0.002234 & 0.1601 & 0.1063 & 0.002016 & 0.000164 & 0.1633 & 0.1052 & 0.002033 & 0.000164\\
	\hspace{1em}Three points: $(0.1T, T, T_{\text{post}})$ & \textBF{1.0102} & \textBF{0.6000} & \textBF{0.020007} & \textBF{0.002230} & \textBF{0.1421} & \textBF{0.0503} & \textBF{0.001787} & \textBF{0.000078} & \textBF{0.1431} & \textBF{0.0505} & \textBF{0.001773} & \textBF{0.000076}\\
	\hspace{1em}Three points: $(0.5T, T, T_{\text{post}})$ & 1.0117 & 0.5996 & 0.020034 & 0.002232 & 0.1540 & 0.0903 & 0.001939 & 0.000139 & 0.1581 & 0.0919 & 0.001952 & 0.000142\\
	\hspace{1em}Four Point: $(0.1T, 0.5T, T, T_{\text{post}})$ & \textBF{1.0102} & \textBF{0.6000} & \textBF{0.020007} & \textBF{0.002230} & \textBF{0.1421} & \textBF{0.0503} & \textBF{0.001787} & \textBF{0.000078} & \textBF{0.1431} & \textBF{0.0505} & \textBF{0.001773} & \textBF{0.000076}\\
	\bottomrule
\end{tabular}}
 \captionof{table}{\small{Simulation results with true parameter $(\alpha^\star,\beta^\star,\gamma^\star,\eta^\star) = (1.0000, 0.6000, 0.01986, 0.002229)$.}}
 \label{Tab:Tcr2222}
\end{center}

\section{The Propagator Models}\label{s3p2}
As discussed in \cite{zarinelli2015beyond}, the linear permanent market impact predicted by the Almgren-Chriss model can deviate from the market impact trajectory for non-VWAP executions and fail to reproduce the concavity of the market impact curve \cite{zarinelli2015beyond}. The propagator models \cite{bouchaud2003fluctuations,gatheral2013dynamical,obizhaeva2013optimal}, on the other hand, can consistently recover features of market impact curve, including the \textit{square-root law}, by allowing non-linear \textit{instantaneous} function $f(\cdot)$ and a \textit{decay kernel} $G(\cdot)$ to capture the \textit{transient} component of market impact. The word \textit{transient} implies the trade impact is neither permanent nor temporary, but decays over time\cite{moro2009market,curato2017optimal}. As discussed in \cite{bouchaud2003fluctuations, busseti2012calibration}, such a transient assumption is consistent with real data and is in line with empirical findings on market efficiency and strong concavity of impact. In the framework of \eqref{main}, we have\footnote{The equation \eqref{propdynamic} is typically referred to as the continuous-time propagator model or the Gatheral (JG) model. The discrete-time propagator was first proposed by \cite{bouchaud2003fluctuations}.}
\begin{equation}\label{propdynamic}
    S_t=S_0+\int_0^t f(v) G(t-s)ds+\sigma\int_0^t dWs.
\end{equation}
Notable functional forms of instantaneous impact function $f$ and decay kernel $G$ \eqref{propdynamic} includes: 

\begin{enumerate}
\item power-law $f(v)\propto v^{\delta}$ and power-law decay $G(s)\propto s^{-\gamma}$.\footnote{Here $\propto$ indicates in proportional to.} As we shall see, this formulation recovers the square-root law when $\delta=\gamma=\frac{1}{2}$. When $\delta=1$, this corresponds to the transient linear price impact where the optimal execution strategy can be explicitly solved using Fredholm integral equation of the first kind \cite{gatheral2012transient}. For non-linear $f$, the optimal execution problem is harder to solve and generally requires discretization and iterative optimization  \cite{dang2017optimal,curato2017optimal}. When $G(s)=\mathbf 1_{\{0\}}(s)$, this recovers the temporary impact function $h$ of \eqref{acdynamic} in \cite{almgren2001optimal}. 

\item linear $f(v)\propto v$ and exponential decay $G(s)\propto e^{-\rho s}$. This is one of the first propagator models with transient impact proposed by \cite{obizhaeva2013optimal}. The original formulation focused on the modeling of shape/dynamics of the limit order book (LOB), which is later extended by \cite{alfonsi2010optimal} for generally-shaped LOB. The connection/equivalence of price impact reversion and volume impact reversion is further discussed in  \cite{alfonsi2010optimal1,gatheral2011exponential}. 

\item logarithmic $f(v) \propto \log(v/v_0)$ and $G(s)\propto l_0(l_0+s)^{-\gamma}$ or $G(s)\propto (l_0^2+s^2)^{-\gamma/2}$. Here $\gamma\approx\frac{1-\alpha}{2}$ is related to the exponent of auto-correlation among trade signs studied in \cite{bouchaud2003fluctuations}. The logarithmic impact has also been supported empirically to fit the data across different magnitudes better over the square root function. \cite{zarinelli2015beyond}. 

\end{enumerate}

There are many other forms of $f$ and $G$ (e.g., Gaussian kernel, trigonometric kernel in \cite{gatheral2012transient} as examples of where one can solve Fredholm equation) and propagator models have been extensively studied (e.g., \cite{gatheral2010no,gatheral2011exponential,gatheral2012transient,curato2017optimal,bouchaud2003fluctuations,bouchaud2006random,alfonsi2010optimal1}) to characterize conditions of the model which prevent price manipulation strategy or guarantee the stability of trades \cite{gatheral2010no} in optimal execution strategy. The various forms of $f$ and $G$ are  typically modeled separately. As a result, their regularity conditions in assumption \ref{modelspec}-\ref{nonsigfishh} need to be checked in a case-by-case basis, although they are satisfied for common propagator models.

\subsection{Sufficient Statistics for a Discrete Price Trajectory}

In this section, we fix $(v,T)$ by letting $G_{\text{order}}= \delta_{v,T}$ be a point mass to simplify the discussion. In the same spirit as Section \ref{sus}, we first study the experiment based on $\boldsymbol S_{\text{full}}=\{S_{t_i}\}_{i\in[N]}$ on a given grid $\{\Delta t, 2\Delta t,..., N\Delta t=T\}$. Note that in this setting we no longer include a $S_{T_{\text{post}}}$ with $T_{\text{post}}>T$ to measure the permanent impact, i.e., the height of plateau that peak impact relaxes to height, as its determination might be difficult \cite{zarinelli2015beyond}.  Formally, we have the following characterization of separability before we can state the theorem based on $\boldsymbol S_{\textup{full}}$, from a parameter estimation point of view as in Section \ref{sus},

\begin{definition}
The form of $f$ and $G$ in \eqref{propdynamic} is separable if they are modeled by separate parameters, i.e., there exist parameter spaces $\Theta_f$ and $\Theta_G$ such that $\Theta=\Theta_f \times \Theta_G$ and $f(\cdot; \theta)=f(\cdot, \Theta_f)$, $G(\cdot; \theta)=G(\cdot; \Theta_G)$ for all $\theta\in\Theta$.
\end{definition}

\begin{theorem}\label{minisuff}
Fix $(v,T)$ and consider \eqref{gassexp}. With a full-grid observations $\boldsymbol S_{\textup{full}}=\{S_{t_i}\}_{i\in[N]}$, for power-law kernel $G(s)\propto s^{-\gamma}, G(s)\propto l_0(l_0+s)^{-\gamma}$ where $\gamma\in\Gamma\subseteq(0,1)$ contains some open set or exponential kernel $G(s)\propto e^{-\rho s}$ where $\rho\in\Gamma\subseteq \mathbb R^{+}$ contains some open set, assume the instantaneous impact $f$ is positive (i.e., $f(v;\theta)\neq 0$ for $v>0$) and is separable from from the kernel $G$, then the \textbf{unique} sufficient statistic (up to a one-to-one transformation) for propagator model \eqref{propdynamic} is the full trajectory data $\boldsymbol{S}_{\textup{full}}$.
\end{theorem}
\begin{proof}
A sufficient statistic is minimally sufficient if it can be represented as a function of every other sufficient statistic, i.e., $\boldsymbol X$ is minimally sufficient if for every other sufficient statistic $\boldsymbol Y$, there exists some function $f$ such that $\boldsymbol X=f(\boldsymbol Y)$ (a.e. for all $\theta$). Thus, if we can show $\boldsymbol{S}_{\text{full}}$ is minimal sufficient, then for any sufficient statistic $\boldsymbol X$, we know both $\boldsymbol X$ is a function of $\boldsymbol{S}_{\text{full}}$ (by definition of a statistic) and $\boldsymbol{S}_{\text{full}}$ is a function of $\boldsymbol X$ (by minimal sufficiency of $\boldsymbol{S}_{\text{full}}$), thus proving $\boldsymbol{S}_{\text{full}}$ is the unique sufficient statistic, up to a one-to-one mapping.
Under \eqref{main}, we have $\mu(t,v;\theta)=f(v;\theta)\int_0^t G(s;\theta)ds$. Then we can write, as in \eqref{likeli}:
\begin{align}\label{mini}
   p_\theta(\boldsymbol{S}_{\text{full}}|v,T)\propto& \exp\Big(-\sum_{i=1}^N\frac{(S_{t_i}-S_{t_{i-1}}-(\mu(t_i,v;\theta)-\mu(t_{i-1},v;\theta)))^2}{2\sigma^2\Delta t}\Big) \nonumber\\
   =& \exp\Big(-\sum_{i=1}^N\frac{(S_{t_i}-S_{t_{i-1}}-f(v;\theta)\int_{t_{i-1}}^{t_i} G(s;\theta)ds)^2}{2\sigma^2\Delta t}\Big)\nonumber\\
   =& f_1(\boldsymbol{S}_{\text{full}}) \exp\Big(\Big(\sum_{i=1}^N(S_{t_i}-S_{t_{i-1}})\cdot \frac{f(v;\theta)}{\sigma^2 \Delta t}\int_{t_{i-1}}^{t_i} G(s;\theta)ds\Big) -f_2(\theta)\Big).
\end{align}
for some $f_1,f_2$. It is known that (see e.g., \cite{lehmann2006theory} or Theorem 3.19 \cite{keener2010theoretical}) that, for an exponential family with density
\begin{equation}\label{expfam}
  p_\theta(x)=\exp(\eta(\theta)\cdot T(x)-B(\theta))h(x),  
\end{equation}
where the $\eta(\theta), T(x)\in\mathbb R^k$ and $\eta(\theta)\cdot T(x)$ denote their dot product, the statistic $T(x)\in\mathbb R^K$ is minimal sufficient if (1) $\{T_i(x)\}_{1\leq i \leq k}$ is linearly independent and (2) the dimension of the convex hull of $\{\eta(\theta)\}_{\theta\in\Theta}$ has dimension $k$. The convex hull of a set $S$ is $\{\sum_{i=1}^n t_ix_i: t_i\geq0, x_i\in S, \forall i, \sum_{i=1}^n t_i=1, n>0\}$. An equivalent condition for (2) is: there exists $\{\theta_0,\theta_1,...,\theta_k\}\subseteq\Theta$ such that $\{\eta(\theta_i)-\eta(\theta_0)\}_{1\leq i \leq k}$ is linearly independent. Now, compare \eqref{mini} with \eqref{expfam}, we let $T(\boldsymbol S_{\textup{full}})\in\mathbb R^N$ with $T_i(S_{\textup{full}})=S_{t_i}-S_{t_{i-1}}$ for $i\in[N]$ and let $\eta(\theta)\in\mathbb R^N$ with $\eta_i(\theta)=\frac{f(v;\theta)}{\sigma^2 \Delta t}\int_{t_{i-1}}^{t_i}G(s;\theta)ds$ for $i\in[N]$. Clearly $\{T_i\}_{i\in[N]}$ are linearly independent, as there are no linear constraints among the price increments $S_{t_i}-S_{t_{i-1}}$. Thus (1) is satisfied. For (2), since $f(v;\theta)\neq 0$ and is separable from $G$, it suffices to show the convex hull of $\{\eta(\theta_G)\}_{\theta_G\in\Theta_G}\subseteq \mathbb R^N$ has dimension $N$ for the various kernel $G(\cdot)$ mentioned above, where we redefine $\eta_i(\theta_G)=\int_{t_{i-1}}^{t_i}G(s;\theta_G)ds$ using $G$ only. This is left in the Appendix \ref{condim}. Finally, the minimal sufficiency of $T(\boldsymbol S_{\textup{full}})=\{S_{t_i}-S_{t_{i-1}}\}_{i\in[N]}$ is equivalent to the minimal sufficiency of $\boldsymbol S_{\textup{full}}$ since $S_0$ is given.
\end{proof}

Theorem \ref{minisuff} shows, when fitting a general propagator model, only the entire price trajectory provides the whole information, not any summary statistics (e.g., cost of VWAP, or the three points as in Theorem \eqref{acfinal} for Almgren-Chriss model). To see more clearly why this should be the case, consider a simplified example.

\begin{example}\label{fullpathexp}
\textup{Suppose we are estimating a one-dimensional parameter $\theta\in\Theta\subseteq\mathbb R$ from \eqref{main} and $(v,T)$ is fixed. We can estimate $\theta$ using either the average cost of VWAP (up to a constant factor $ v$):
\begin{equation}\label{Jexam}
    J=\frac{1}{T}\int_0^T S_tdt \sim\mathcal N(\frac{1}{T}\int_0^T \mu_{\theta}(t,v)dt, \frac{\sigma^2 T}{3}),
\end{equation}
or the entire path $\boldsymbol S_{\text{full}}$. Using Lemma \ref{jacob}, the Fisher information $\mathcal I_J(\theta)$ is 
\begin{equation*}
    \mathcal{I}_J(\theta)=\frac{3}{\sigma^2 T^3}\Big(\int_0^T\frac{\partial \mu_\theta(t,v)}{\partial\theta}dt\Big)^2,
\end{equation*}
whereas, using Lemma \ref{jacob} and the independence of $\{S_{t_i}-S_{t_{i-1}}\}_{i\in[N]}$ (multivariate Gaussian), the Fisher information $\mathcal I_{\boldsymbol S_{\text{full}}}(\theta)$ is (let $\Delta t\downarrow 0$ and assume boundary condition $\frac{\partial \mu_\theta(0,v)}{\partial\theta}=0$):
\begin{equation}\label{asymfull}
 \lim_{\Delta t \rightarrow 0}\mathcal I_{\boldsymbol S_{\text{full}}}(\theta)=\frac{1}{\sigma^2}\int_0^T(\frac{\partial^2 \mu_\theta(t,v)}{\partial\theta\partial t})^2dt.
\end{equation}
 Let $h(t;\theta,v)\triangleq\frac{\partial \mu_\theta(t,v)}{\partial\theta}$, $h'(t;\theta,v)\triangleq \frac{dh(t;\theta,v)}{dt}$ and assume $h'(0;\theta,v)=0$. Then $\lim_{\Delta t\rightarrow 0}\mathcal I_{\boldsymbol S_{\text{full}}}(\theta)\geq \mathcal{I}_J(\theta)$ then immediately follows from the inequality:
 \begin{align*}
    \Big(\int_0^T h(t;\theta,v) dt\Big)^2=&\Big(\int_0^T \int_0^t h'(s;\theta,v)ds dt\Big)^2\nonumber\\
    =&\Big(\int_0^T \int_s^T h'(s;\theta,v)dt ds\Big)^2\nonumber\\
    =&\Big(\int_0^T (T-s) h'(s;\theta,v) ds\Big)^2 \leq \frac{T^3
    }{3} \int_0^T (h'(s;\theta,v))^2 ds,
 \end{align*}
where the second equality follows from the Fubini-Tonelli theorem and the last inequality follows from Cauchy–Schwarz inequality.}
\end{example}

\subsection{Estimation of Impact Function $f$}
In the previous section, Theorem \ref{minisuff} addresses the joint estimation of $f$ and $G$, but the proof also works for the estimation of $G$ given $f$. However, Theorem \ref{minisuff} can be adjusted for estimation of impact function $f$ given $G$, which typically can be calibrated from other methods (\cite{bouchaud2003fluctuations,busseti2012calibration}).  As discussed in section 2.1 of \cite{curato2017optimal}, one of the ``major attractions of the propagator models to practitioners" is that given $G(\cdot)$ and a large data collection of ``VWAP-like executions", the expected cost of VWAP $v\int_0^T S_tdt-XS_0$:
\begin{equation}\label{propcost}
    c_{\text{vwap}}(\theta)= vf(v;\theta)\int_0^T\int_0^tG(s)dsdt,
\end{equation}
can be used to estimate $f(v)$ and reflect the performance of execution algorithms. In this section, based on Theorem \ref{minisuff}, we quantify the efficiency of using VWAP cost to calibrate $f$, as well as that of using trajectory data, in the limit case when $\Delta t\rightarrow 0$. 

For fixed $(T,v)$, the information matrix based on VWAP cost alone only has rank 1, so for simplicity, we first assume we are calibrating an impact function  $f(v;\theta)$ characterized by a one-dimensional parameter $\theta\in\Theta\subseteq \mathbb R$. As in Remark \ref{tandv}, the general case could depend on the distribution of $(T,v)\sim G_{\text{order}}$ and we avoid this dependence to isolate the effect. For concreteness, one can suppose a power law impact $f(v)=v^\delta$ and we want to estimate $\delta\in\Gamma\subseteq (0,1]$. Interestingly, although considered unrealistic for certain cases (\cite{gatheral2010no}), it is recently argued in \cite{jusselin2020no} that power-law impact is the only impact function consistent with the no-arbitrage condition. 

\begin{definition}\label{disJ}
Fix $\boldsymbol S_{\text{full}}=\{S_{t_i}\}_{i\in\mathcal [N]}$ be price trajectory data. Let $\boldsymbol S_{\text{partial}}=\{S_{t}\}_{t\in\mathcal K}$ be a partial price trajectory data with $\mathcal K=\{t_0^{\mathcal K},t_1^{\mathcal K},t_2^{\mathcal K},...,t_{N_{\mathcal K}}^{\mathcal K}\}\subseteq [0,T]$ and $0=t_0^{\mathcal K}\leq t_1^{\mathcal K} \leq t_2^{\mathcal K}\leq...\leq t_{N_{\mathcal K}}^{\mathcal K}\leq T$. Let $J_{\text{discrete}}=\sum_{i=1}^{N_{\mathcal K}}S_{t_i}\Delta t$ be the discrete VWAP cost while $J=\int S_tdt$ is the true VWAP cost.
\end{definition}
% \begin{coro}
%     Let $\{J_n\}$
% \end{coro}

\begin{theorem}\label{calif}
In the same setting as Theorem \ref{minisuff}, suppose the decay kernel $G(\cdot)$ is fixed and the impact function $f(v;\theta)$ is parameterized by some $\theta\in\Gamma\subseteq \mathbb R$ where $\Gamma$ contains some open set. Then the sufficient statistics for estimating $f$ in \eqref{propdynamic} is $\phi(\boldsymbol S_{\text{full}})\triangleq\sum_{i=1}^N (S_{t_i}-S_{t_{i-1}})\int_{t_{i-1}}^{t_i} G(s)ds$, which implies
\begin{equation}\label{40}
    \mathcal I_{\boldsymbol S_{\text{full}}}(\theta)=\mathcal I_{\phi(\boldsymbol S_{\text{full}})}(\theta).
\end{equation}
Moreover, in the limit as $\Delta t\rightarrow 0$, we have 
\begin{equation}\label{41}
    \mathcal I_{\boldsymbol S_{\text{full}}}(\theta)\rightarrow \frac{(\frac{\partial f}{\partial \theta})^2}{\sigma^2}\int_0^T G^2(t)dt,
\end{equation}
and 
\begin{equation}\label{43}
    \mathcal I_{J_{\text{discrete}}}(\theta)\rightarrow \mathcal I_{J}(\theta) =\frac{(\frac{\partial f}{\partial \theta})^2}{\sigma^2}\cdot\frac{3}{T^3}\Big(\int_0^T G(t)(T-t)dt\Big)^2.
\end{equation}
Finally, we have 
\begin{equation}\label{42}
    \mathcal I_{\boldsymbol S_{\text{partial}}}(\theta)\rightarrow \frac{(\frac{\partial f}{\partial \theta})^2}{\sigma^2}\Big(\sum_{i=1}^{N_{\mathcal K}}\frac{(\int_{t_{i-1}^{\mathcal K}}^{t_{i}^{\mathcal K}}G(t)dt)^2}{t_{i}^{\mathcal K}-t_{i-1}^{\mathcal K}}\Big).
\end{equation}
\end{theorem}
\begin{proof}
The proof is left in Appendix \ref{dsfds}.
\end{proof}

As a straightforward result of Theorem \ref{calif}, we can compare the Fisher information for calibrating $f$ of using two points $(S_t,S_T)$ versus $J$.
\begin{coro}\label{twovscwap}
In the setting of Theorem \ref{calif}, we have 
\begin{equation}
    \mathcal I_{S_t,S_T}(\theta)= \frac{(\frac{\partial f}{\partial \theta})^2}{\sigma^2}\Big(\frac{(\int_0^t G(s)ds)^2}{t}+\frac{(\int_t^T G(s)ds)^2}{T-t}\Big).
\end{equation}
Consequently, we have $\mathcal I_{S_t,S_T}(\theta) -\mathcal I_J(\theta) \geq 0$ when 
\begin{equation}\label{comppppp}
    \Big(\frac{(\int_0^t G(s)ds)^2}{t}+\frac{(\int_t^T G(s)ds)^2}{T-t}\Big) \geq \frac{3}{T^3} \Big(\int_0^T G(s)(T-s)ds\Big)^2.
\end{equation}
\end{coro}

For calibrating $f$, using both Theorem \ref{calif} and Corollary \ref{twovscwap}, we can compute optimal sampling points $\boldsymbol S$ that maximize $\mathcal I_{\boldsymbol S}$ or compare its efficiency with an estimation based on VWAP cost. In Theorem \ref{acfinal}, we see that, for the Almgren-Chriss model, a sample $S_t$ early enough with $\frac{t}{T}\leq \frac{1}{4}$ can outperform the VWAP based method. For the propagator model, the comparison between $(S_t, S_T)$ and $J$ depends on the specific decay kernel $G$ as in \eqref{comppppp}. However, as we shall see in the following examples, the result for the propagator model bears a resemblance to the result for the Almgren-Chirss model in Theorem \ref{acfinal}.
\begin{example}\label{321comp}
\textup{For the decay kernel $G(s)=s^{-\gamma}$ with $\gamma=0.4$ \cite{bouchaud2003fluctuations}, based on \eqref{comppppp}, we have $\mathcal I_{S_t,S_T}(\theta)\geq \mathcal I_J(\theta)$ when $2.11\cdot 10^{-4} \leq \frac{t}{T} \leq 0.279$. For different values of $\gamma$, one can check $\mathcal I_{S_t,S_T}(\theta)\geq \mathcal I_J(\theta)$ when the condition in Table \ref{tab:my_labelgamma} is satisfied.
\begin{table}[htb]
    \centering
    \scalebox{0.85}{
    \begin{tabular}{|l|c|c|c|c|c|r|}
	\hline
	\text{} & $\gamma=0.35$ &$\gamma=0.45$&$\gamma=0.5$ &$\gamma=0.55$ &$\gamma=0.65$&$\gamma=0.75$ \\
	\hline
    $\tau= t/T$ & ${8.97\cdot 10^{-4}\leq \tau\leq 0.369}$ &${9.41\cdot 10^{-7}\leq \tau\leq 0.252}$&$\tau\leq \frac{1}{4}$&$\tau\leq 0.257$&$\tau\leq 0.279$&$\tau\leq 0.301$\\
	\hline
\end{tabular}}
    \caption{The range of values of $\tau=\frac{t}{T}$ where $\mathcal I_{S_t,S_T}(\theta)\geq \mathcal I_J(\theta)$. }
    \label{tab:my_labelgamma}
\end{table}
As shown in Table \ref{tab:my_labelgamma}, similar to Theorem \ref{acfinal}, the comparison of $\mathcal I_{S_t,S_T}(\theta)$ and $\mathcal I_J(\theta)$ only relies on $\tau=\frac{t}{T}$. In general, a smaller value of $\tau$ would help $(S_t,S_T)$ outperform $J$, as long as it is not unreasonably/unrealistically small (i.e. on the order of $10^{-4}$ or less for certain $\gamma$). Perhaps more interestingly, the cut-off point for being ``small" is around $\frac{1}{4}$ just as in Theorem \ref{acfinal}, although the precise value also changes with $\gamma$ (the exact value of $\frac{1}{4}$ is achieved for $\gamma=0.5$).}
\end{example}

\begin{example}\label{expyixia}
 For $G(s)=e^{-\rho s}$, the comparison of \eqref{comppppp} depends on specific values of $t$ and $T$, not just their ratio $\tau$. However, in the case where $t,T\rightarrow \infty$ but $\frac{t}{T}\rightarrow \tau$, we have $\mathcal I_{S_t,S_T}(\theta)\geq \mathcal I_J(\theta)$ as long as $$\tau \leq \frac{1}{3},$$ regardless of the value of $\rho$. The derivation is left in Appendix \ref{asdas}. 
\end{example}

\subsection{Sampling Strategy for Propagator Models: A Numerical Study} \label{subsec:path_numerical}
Theorem \ref{minisuff} shows that the full price trajectory $\boldsymbol S_{\text{full}}$ is the only sufficient statistic for a general propagator model. However, it does not specify how much information is lost when we only have partial observations $\boldsymbol S$, neither suggests a consistent sampling strategy over the cost-of-VWAP based estimation, as in Section \ref{samplestrat}. Unfortunately, unlike Theorem \ref{acfinal}, the optimal sampling strategy for \eqref{propdynamic} depends on the specific values of $t$ and $T$ (not just their ratios) and the specific form of $G$ and $f$. Thus, in section, we aim to gain some insights into questions (1) and (2) through numerical examples. We conduct numerical studies to compare different sampling strategies against the VWAP-cost-based estimation methods, in the flavor of Section \ref{samplestrat}. 
\subsubsection{Example 1: Power Law}
In the first example, we take the power-law kernel $G(s)=s^{-\gamma}$ with $\gamma=0.4$ \cite{bouchaud2003fluctuations,busseti2012calibration} and power-law impact $f(v)=v^{\delta}$ with $\delta=0.6$ \cite{almgren2005direct} (hence $\mu(v,t; \theta)= \frac{v^\delta t^{1-\gamma}}{1-\gamma}$). We shall compare the asymptotic variance for calibrating $\delta$, under different estimation methods. In this example, we assume $T=0.1$, $v=0.3$, $\sigma=1$ $\gamma^\star=0.4$ and $\delta^\star=0.6$. Assuming we are only estimating $\delta$ (i.e., all other parameters are given), then as computed in \eqref{asymfull}, as $\Delta t\downarrow0$, one can check the entry for $\delta$ in the Fisher information matrix $[\mathcal I_{\boldsymbol S_{\text{full}}}]_{\delta,\delta}$ is
\begin{equation*}
    [\mathcal I_{\boldsymbol S_{\text{full}}}]_{\delta,\delta} \rightarrow \frac{1}{\sigma^2}\int_0^T(\frac{\partial^2 \mu_\theta(t,v)}{\partial\delta\partial t})^2dt=\int_0^T t^{-2\gamma}v^{2\delta} \ln^2(v) dt=1.078.
\end{equation*}
We compare estimation methods based on the VWAP cost $J=v\int_0^T S_tdt-XS_0\sim\mathcal N(\frac{v^{1+\delta}T^{2-\gamma}}{(1-\gamma)(2-\gamma)},\frac{\sigma^2v^2T^3}{3})$ with methods based on price points $S_i$'s sampled from trajectory. Again, as in \eqref{Jexam}, we have
\begin{equation*}
    [\mathcal{I}_J]_{\delta,\delta}=\frac{3}{\sigma^2 T^3}\Big(\int_0^T\frac{\partial \mu_\theta(t,v)}{\partial\delta}dt\Big)^2=0.702,
\end{equation*}
which recovers just over 65.1\% information for estimating $\delta$ (i.e., $0.651\approx \frac{0.702}{1.078}$), which implies the asymptotic variance for estimating $\delta$ using $J$ is approximately $1.536\approx 1/0.651$ times larger than the asymptotic variance for estimating $\delta$ using $\boldsymbol S_{\text{full}}$. Next, based on results from Theorem \ref{acfinal}, we pick the following representative sampling times $t_1=\frac{T}{8}, t_2=\frac{T}{4}, t_3=\frac{5T}{8}$ in addition to $T$ and compare their Fisher information ratio  $\frac{[\mathcal I]_{\delta,\delta}}{[\mathcal I_{\boldsymbol S_{\text{full}}}]_{\delta,\delta}}$. The results are summarized in Table \ref{tab:my_label}.
\begin{table}[ht]
    \centering
    \scalebox{0.85}{
    \begin{tabular}{|l|c|c|c|c|c|c|r|}
	\hline
	\text{} & $S_{t_1},S_T$ & $S_{t_2},S_T$ &$S_{t_3},S_T$&$ S_{t_1},S_{t_2},S_T$ &$ S_{t_1},S_{t_3},S_T$&$ S_{t_2},S_{t_3},S_T$&$ S_{t_1},S_{t_2}, S_{t_3},S_T$\\
	\hline
    ${[\mathcal I]_{\delta,\delta}}/{[\mathcal I_{\boldsymbol S_{\text{full}}}]_{\delta,\delta}}$ & \textbf{0.689} & 0.657&0.595&\textbf{0.700}&\textbf{0.698}&\textbf{0.661}&\textbf{0.704}\\
	\hline
\end{tabular}}
\caption{Comparison of Fisher information for calibrating  power-law impact.
}
    \label{tab:my_label}
\end{table}

The entries with bold numbers indicate the method has a better performance for estimating $\delta$ than using $J$ (i.e., ratio greater than $0.651$). As expected, we see that adding more points always increases the accuracy of estimation, no matter which points you already include. In this example, similar to Theorem \ref{acfinal}, the inclusion of the earliest sample point $S_{t_1}$ brings most improvement (just the two points with $S_{t_1}$ and $S_T$ can outperform $J$) and any combination of three points can outperform the VWAP cost $J$. The method based on four points recovers about 70\% of the $\mathcal I_{\boldsymbol S_{\text{full}}}$. However, we do note that $\mathcal I_{\boldsymbol S_{\text{full}}}$ here are computed based on the ideal, limiting case where $\Delta t\downarrow0$ and $N\rightarrow \infty$ (the number observations approach infinity). Moreover, we note that this depends on the particular choice of points. For example, if we let $t_1=\frac{T}{4}, t_2=\frac{T}{2}, t_3=\frac{3T}{4}$ and $t_4=T$, then we would have $\frac{[\mathcal I]_{\delta,\delta}}{[\mathcal I_{\boldsymbol S_{\text{full}}}]_{\delta,\delta}}=0.662$, whereas if we set  $t_1=\frac{T}{20}, t_2=\frac{T}{10}, t_3=\frac{T}{3}$ and $t_4=T$, we would have $\frac{[\mathcal I]_{\delta,\delta}}{[\mathcal I_{\boldsymbol S_{\text{full}}}]_{\delta,\delta}}=0.749$ (here we see again the improvement related to the inclusion of earlier sample points). 

We again verify our results on simulation. The results are summarized in Table \ref{tab:well_specified}. The performance is evaluated using the squared error between the true cost and the estimated cost,
and the parenthesis shows the standard error of the mean error.
$J$ is the model based on the average cost of VWAP, 
others are based on the points along the price trajectory.
$t_1=\frac{T}{8}, t_2=\frac{T}{2}, t_3=\frac{5T}{8}$ and $t_4=T$. In the simulation study, we use similar settings as the above theoretical analysis with the power-law kernel $G(s)=s^{-\gamma}$ and power-law impact $f(v)=v^{\delta}$.
We set $T=0.1$, $v=0.3$, $\sigma=0.2$, $\gamma=0.4$ and $\delta=0.6$. The diffusion process is discretized into $1024$ time bins.
The data-generating model and the estimation model are in the same parametric format and only $\delta$ is estimated.
The first model is based on the average cost of VWAP as in \eqref{Jexam}.
Other models rely on the price trajectory as discussed in section \ref{samplestrat} with a subset or all points at 
$t_1=\frac{T}{8}, t_2=\frac{T}{2}, t_3=\frac{5T}{8}$ and $t_4=T$.
Each estimation has 300 samples, and the simulation is repeated 40,000 times.
The performance is evaluated using the squared error of the cost with the estimated $\hat\delta$.
Bold numbers in Table \ref{tab:well_specified} indicate the method has a better performance than using $J$, after a standard two-sample t-test.
The conclusion agrees with the theoretical analysis that the cost-based method is inferior to the price trajectory-based methods if earlier price trajectory points are selected (for example $t_1$ or $t_2$).
In addition, if more trajectory points are chosen or the earlier the points are selected, the performance will be better.
\begin{table}[ht]
\centering
\scalebox{0.85}{
\begin{tabular}{|l|c|}
\hline
Estimator &  mean error $\times 10^{-7}$ (SE $\times 10^{-9}$)  \\ \hline
$ J$ &  $3.980$ ($2.807$)  \\ \hline

$ S_{t_1}, S_T$ & $\textbf{3.769}$ ($2.648$) \\ \hline
$S_{t_2}, S_T$ & $\textbf{3.930}$ ($2.771$) \\ \hline
$S_{t_3}, S_T$ & $4.323$ ($3.054$) \\ \hline 

$S_{t_1}, S_{t_2}, S_T$ & $\textbf{3.706}$ ($2.604$)\\ \hline
$S_{t_1}, S_{t_3}, S_T$ & $\textbf{3.716}$ ($2.611$) \\ \hline
$S_{t_2}, S_{t_3}, S_T$ & $\textbf{3.905}$ ($2.752$) \\ \hline

$S_{t_1}, S_{t_2}, S_{t_3}, S_T$ & $\textbf{3.684}$ ($2.589 $) \\ \hline
\end{tabular}}
\caption{\small{Numerical comparison of models with  power-law impact using simulation.}
}
\label{tab:well_specified}
\end{table}

\subsubsection{Example 1: Exponential Kernel}
Next, we consider the setting of \cite{obizhaeva2013optimal}, where $f(v)\propto v$ and $G(s)\propto e^{-\rho s}$ so that $\mu(t,v;\theta)=\frac{cv}{\rho}(1-e^{-\rho t})$ for some $c$. For this example, we assume $T=0.1$, $v=0.3$, $\sigma=1$ $c^\star=0.01$ and $\rho^\star=0.15$. In this example, we estimate $c$ and $\rho$ jointly and we can compute $\mathcal I_{\boldsymbol S_{\text{full}}}$ now as (again we let $\Delta t\downarrow 0$): $$[\mathcal I_{\boldsymbol S_{\text{full}}}]_{\theta_i,\theta_j}=\frac{1}{\sigma^2}\int_0^T(\frac{\partial^2 \mu_\theta(t,v)}{\partial\theta_i\partial t})(\frac{\partial^2 \mu_\theta(t,v)}{\partial\theta_j\partial t})dt.$$
For joint estimation, we no longer simply compare the ratio of specific entries in $\mathcal I$ against $\mathcal I_{\boldsymbol S_{\text{full}}}$ as in the last example. Instead, we compute the matrix 2-norm (i.e., the spectral norm) of $\mathcal I^{1/2}_{\boldsymbol S_{\text{full}}}\cdot\mathcal I^{-1}\cdot\mathcal I^{1/2}_{\boldsymbol S_{\text{full}}}$. To see why, note that 
\begin{equation*}
    \|\mathcal I^{1/2}_{\boldsymbol S_{\text{full}}}\cdot\mathcal I_X^{-1}\cdot\mathcal I^{1/2}_{\boldsymbol S_{\text{full}}}\|_2 \leq1+ \epsilon \text{ }\Leftrightarrow\text{} \boldsymbol w^T \mathcal I_X^{-1}\boldsymbol w \leq (1+\epsilon) \boldsymbol w^T \mathcal I^{-1}_{\boldsymbol S_{\text{full}}}\boldsymbol w, \text{ } \forall \boldsymbol w
\end{equation*}
which, in view of \eqref{asymvar}, implies the asymptotic variance for estimating any function of $\theta$ using $X$ cannot be more than $(1+\epsilon)$ times more than the asymptotic variance estimated using $\boldsymbol S_{\text{full}}$ (i.e., smaller value of $\|\mathcal I^{1/2}_{\boldsymbol S_{\text{full}}}\cdot\mathcal I^{-1}\cdot\mathcal I^{1/2}_{\boldsymbol S_{\text{full}}}\|_2$ indicates a more accurate estimate, the optimal estimation based on $\boldsymbol S_{\text{full}}$ has $\|\mathcal I^{1/2}_{\boldsymbol S_{\text{full}}}\cdot\mathcal I^{-1}_{\boldsymbol S_{\text{full}}}\cdot\mathcal I^{1/2}_{\boldsymbol S_{\text{full}}}\|_2=1$). Moreover, same as the last example, we also compare with the VWAP cost $J$ on the asymptotic variance for estimating $c$ (i.e., calibrating impact function $f$. Here, as before, we use the ratio $\frac{[\mathcal I]_{c,c}}{[\mathcal I_{\boldsymbol S_{\text{full}}}]_{c,c}}$ to compare against the benchmark $\frac{[\mathcal I_J]_{c,c}}{[\mathcal I_{\boldsymbol S_{\text{full}}}]_{c,c}}=0.754$. ). We do not compute the joint estimation of $c,\rho$ under VWAP cost $J$ as the Fisher information based on $J$ alone has rank one. Again, we pick the following representative sampling times $t_1=\frac{T}{8}, t_2=\frac{T}{4}, t_3=\frac{5T}{8}$ in addition to $T$. The results are summarized in Table \ref{tab:my_label1}. 
\begin{table}[ht]
    \centering
    \scalebox{0.75}{
    \begin{tabular}{|l|c|c|c|c|c|c|r|}
	\hline
	\text{} & $S_{t_1},S_T$ & $S_{t_2},S_T$ &$S_{t_3},S_T$&$ S_{t_1},S_{t_2},S_T$ &$ S_{t_1},S_{t_3},S_T$&$ S_{t_2},S_{t_3},S_T$&$ S_{t_1},S_{t_2}, S_{t_3},S_T$\\
	\hline
    $\|\mathcal I^{1/2}_{\boldsymbol S_{\text{full}}}\cdot\mathcal I^{-1}\cdot\mathcal I^{1/2}_{\boldsymbol S_{\text{full}}}\|_2$ & 3.025 & 1.769 &1.426&1.732&1.219&1.137&1.122\\
	\hline
	${[\mathcal I]_{c,c}}/{[\mathcal I_{\boldsymbol S_{\text{full}}}]_{c,c}}$ & \textbf{1-1.258e-5} & \textbf{1-8.184e-6}&\textbf{1-5.590e-6}&\textbf{1-7.962e-6}&\textbf{1-3.376e-6}&\textbf{1-2.274e-6}&\textbf{1-2.051e-6}\\
	\hline
\end{tabular}}
    \caption{Comparison of Fisher information for calibrating  \cite{obizhaeva2013optimal}.}
    \label{tab:my_label1}
\end{table}

The numerical results from Table \ref{tab:my_label1} are in line with the findings from the first example, where adding more points always increases the accuracy of estimation. However, for this model, the inclusion of earlier points $S_{t_1}$ does not improve the estimation as much as the inclusion of later points $S_{t_3}$. This is due to the fact that, we are not just calibrating impact function $f$, but also the decay kernel $G$. As we can see, in this example the asymptotic variance of estimation (for any function of $\theta$) using four points method is within 12.2\% of the optimal one. On the other hand, for calibrating the impact function $f$, all sampling-based methods in the table are close to optimal (i.e. close to 1) and considerably outperform $J$ (i.e., 75.4\%).

\begin{remark}
Although we have shown that the full price trajectory data $\boldsymbol S_{\text{full}}$ is most efficient and adding more price points $S_t$ increases the efficiency, there seems to be, based on previous examples, a ``diminishing return" effect to adding price points, as the increase in efficiency tends to decrease in each addition. However, the quantification of such an effect or the characterization of efficiency for $n$ given points requires further research. 
\end{remark}

\subsection{Square-root Law: A Special Case}

A well-known, widely-used rule-of-thumb to produce a pre-trade estimate of cost is the \textit{square-root law}. As stated in \cite{toth2011anomalous}, the square-root law provides a remarkably good fit on empirical data (although it is argued in \cite{zarinelli2015beyond} that the logarithmic function provides a better fit across more orders of magnitude than the square root function). It is suggested that the square-root law remains a robust statistical phenomenon across a spectrum of traded instruments/markets and is roughly independent of trading period, order type, trade duration/rate, and stock capitalization \cite{zarinelli2015beyond}. In particular, under the framework \eqref{main}, the square-root law can be stated in terms of the impact \cite{zarinelli2015beyond}:
\begin{equation}\label{sq1}
    \mu(T,v)\propto (vT)^{\frac{1}{2}}=X^{\frac{1}{2}},
\end{equation}
 or the average cost per-share \cite{gatheral2010no}:
\begin{equation}\label{sq2}
    \frac{c_{\text{vwap}}}{X}\triangleq\frac{\mathbb E[v\int_0^TS_tdt-XS_0]}{X}\propto X^{\frac{1}{2}}.
\end{equation}
The original statement uses $X/V_D$ but we have scaled $X$ by $V_D$ (i.e., $V_D=1$, see remarks following \eqref{tandv}). As noted in \cite{zarinelli2015beyond}, under \eqref{sq1}, the market impact no longer depends on the specific trading rate $v$ or trade duration $T$, but only on their product $X=vT$. This phenomenon can be seen as a special case of propagator model under the power-law impact $f(v)\propto v^\delta$ and power-law decay $G(s)=s^{-\gamma}$, with the constraint that $\delta+\gamma=1$ (see \cite{gatheral2010no}). Under such model, \eqref{sq1} is consistent 
\begin{equation}\label{sq1prop}
    \mu(T,v;\theta)\propto \frac{v^{\delta}T^{1-\gamma}}{1-\gamma}=\frac{1}{\delta} (vT)^\delta
\end{equation}
and \eqref{sq2} is consistent with 
\begin{equation}\label{sq2prop}
   \frac{c_{\text{vwap}}}{X}=\frac{v\int_0^T\mathbb E[S_t]dt-XS_0}{X}\propto\frac{1}{\delta(1+\delta)} (vT)^{\delta},
\end{equation}
when $\delta=\gamma=\frac{1}{2}$, as a special case of $\gamma+\delta=1$. From an estimation point of view, if we want to calibrate such a propagator model, where the impact only depends on the product $vt$, by estimating $\delta$ under the constraint $\gamma+\delta=1$, we can compare the estimation method using sampled points $(S_{t}, S_T)$ based on \eqref{sq1prop} or using the cost of the VWAP $J=v\int_0^TS_tdt-XS_0$ based on \eqref{sq2prop}. Leveraging techniques from sections before, we have the following Corollary. Again we observe the power of an early sample around 20\%. Yet, just as in Table \ref{tab:my_labelgamma}, sampling too early  is not sufficiently helpful.
\begin{coro}\label{corosq}
Under the propagator model $f(v)\propto v^\delta$ and $G(s)\propto s^{\delta-1}$, if $\delta \leq 0.8$ and $0.005\leq vT\leq 1$, then the Fisher information matrix (one dimension, estimating $\delta$ only) satisfies 
\begin{equation*}
    \mathcal I_{S_t,S_T}(\delta) \geq \mathcal I_{J}(\delta),
\end{equation*}
whenever $1.2\%\leq t/T \leq 22.2\%$.
\end{coro}
\begin{proof}
See Appendix \ref{proofcorosq}.
\end{proof}

\section{Limitations}\label{limitation}
 In this section, we discuss some limitations and their practical implications on the aforementioned theorems. First,  theorems in Section \ref{s3} rely on Assumption \ref{modelspec}, which states, the stochastic law (or the underlying ``true" structure) governing the process of market impact can be correctly specified within the parameterized model. It is natural to discuss, when this assumption fails, how the results of previous theorems hold and what types of adjustments need to be made regarding the quantification of estimation ``efficiency".
\subsection{Model Misspecification}
Model misspecification occurs when the ``true" underlying process of market impact cannot be correctly specified by \eqref{main} (or a specific variant of \eqref{main}, e.g., Almgren Chriss model or the propagator model). In other words, given $(T,v)\sim G_{\text{order}}$, the joint distribution of the discretized price trajectory $\boldsymbol S\in\mathbb R^N$ follows distribution $F$; however, a $\theta^\star\in\Theta$ such that $F(\theta)$ equals $F$ does not exist, where $F(\theta^\star)$ is characterized by $S_t=S_0+\mu_{\theta^\star}(t,v) + \sigma W_t$ for all $t\in[0,T]$ and $(T,v)\sim G_{\text{order}}$. In this setting, under regularity conditions, it is shown in  \cite{white1982maximum,bishwal2007parameter} that the MLE (or QMLE, quasi-maximum likelihood estimator in this setting \cite{white1982maximum}) $\hat\theta_{\text{MLE}} = \mathop{\arg \max}_{\theta\in\Theta} \sum_{j} l(\boldsymbol{S}_j|\theta) $ estimates the parameter $\theta^\star_{\text{KL}}$, which minimizes the Kullback-Leibler (KL) divergence
\begin{align}\label{whitepaper}
    \theta^\star_{\text{KL}}=\mathop{\arg \min}_{\theta\in\Theta}D_{\text{KL}}(F\lVert F(\theta))
\end{align}
with consistency and asymptotic normality. As a widely used statistical distance measure, the KL divergence (also known as relative entropy) $D_{\text{KL}}(\mathbb P\lVert \mathbb Q)$ between two probability measure $\mathbb P$ and $\mathbb Q$ is defined as $$D_{\text{KL}}(\mathbb P\lVert \mathbb Q) = \int\log\frac{d\mathbb P}{d\mathbb Q}d\mathbb P,$$ where $\frac{d\mathbb P}{d\mathbb Q}$ is the Radon-Nikodym derivative of $\mathbb P$ w.r.t $\mathbb Q$. In other words, under model misspecification, the MLE estimates the parameters whose corresponding price trajectory distribution is the closest to the true one, in terms of the statistical distance measured by KL-divergence. Since $F$ is unknown, to gain some intuition on the property of $\theta^\star_{\text{KL}}$, we consider the following idealized example, where we observe the entire continuous path $\{S_t\}_{0\leq t \leq T}$.
\begin{example}
 Given $(v,T)$, suppose the true price process $\{S_t\}_{0\leq t \leq T}$ is governed by a diffusion process:
 \begin{equation}\label{realsde}
    dS_t= \mu^\star(S_t;t,v)dt + \sigma^\star(S_t;t,v)dW_t, \text{  } 0\leq t \leq T,
\end{equation}
for some $\mu^\star(\cdot),\sigma^\star(\cdot)$, which cannot be described by \eqref{main} for any $\theta\in\Theta$:
\begin{equation*}
    dS_t= \frac{\partial \mu_\theta(t,v)}{\partial t}dt + \sigma dW_t, \text{  } 0\leq t \leq T, \theta\in\Theta.
\end{equation*}
Then, using Girsanov theorem, \cite{mckeague1984estimation} showed the log-likelihood function of the process $\{S_t\}_{0\leq t \leq T}$, as an element of $C[0,T]$ ($C[0,T]$ represents the space of continuous function on $[0,T]$), can be written as 
\begin{align*}
    &l(\theta|\{S_t\}_{0\leq t \leq T})= \nonumber\\
    & -\frac{1}{2\sigma^2}\int_0^T \bigg(\frac{\partial \mu_\theta(t,v)}{\partial t}-\mu^\star(S_t;t,v)\bigg)^2dt +\frac{1}{2\sigma^2} \int_0^T\Big(\mu^\star(S_t;t,v)\Big)^2dt +\frac{1}{\sigma^2}\int_0^T{\frac{\partial \mu_\theta(t,v)}{\partial t}\sigma^\star(S_t;t,v)}dW_t.
\end{align*}
Given sufficient metaorder data, the sum of the log-likelihood can be maximized to obtain an optimal value $\hat\theta_{\text{MLE}}$ (if exists), which can be seen as an estimator for 
\begin{align*}
    \theta^\star =& \mathop{\arg \max}_{\theta\in\Theta}\mathbb E \Big[l(\{S_t\}_{0\leq t \leq T}|\theta)\Big] \nonumber\\
    =&\mathop{\arg \min}_{\theta\in\Theta}\int_0^T\mathbb E\bigg[\bigg(\frac{\partial \mu_\theta(t,v)}{\partial t}-\mu^\star(S_t;t,v)\bigg)^2\bigg]dt.
\end{align*}
Thus, in this example, under model misspecification, the $\theta^\star$ corresponds to a model whose drift term (i.e., $\frac{\partial \mu_\theta(t,v)}{\partial t}$, or the gradient of impact function $\mu_\theta(t,v)$ w.r.t $t$,), on average, best matches true drift term $\mu^\star(S_t; t,v)$ of $S_t$ in mean squared error over $[0,T]$. This example suggests the MLE from \eqref{main} seeks to recover an impact function that best describes the shape of the impact function by matching gradients (or how the gradient diminishes, in presence of concavity of impact function \cite{zarinelli2015beyond}), rather than the direct value the impact function. 
\end{example}

More importantly, how does model misspecification affect our discussion, largely based on the Fisher information, regarding the (asymptotic) ``efficiency" of MLE? When the model is correctly specified, as mentioned in \eqref{mleinv}, the asymptotic covariance matrix of $\hat\theta_{\text{MLE}}$ scaling with $\frac{1}{\sqrt{n}}$ (i.e., the number of metaorders) is the Fisher information matrix evaluated at the true model parameter $\theta^\star$,
\begin{equation*}
    \mathcal I(\theta^\star) =\mathcal A(\theta^\star)\triangleq\mathbb E\bigg[\Big(\frac{\partial l(\theta|\boldsymbol{S})}{\partial \theta}\Big)\Big(\frac{\partial l(\theta|\boldsymbol{S})}{\partial \theta}\Big)^T\bigg]
\end{equation*}
which can be equivalently defined in Hessian form, known as the information matrix equivalence theorem \cite{white1982maximum}, as
\begin{equation}\label{equivla}
    \mathcal A(\theta^\star) =\mathcal B(\theta^\star),
\end{equation}
where $[\mathcal B]_{ij}(\theta^\star)\triangleq-\mathbb E\Big[\frac{\partial^2 l(\theta|\boldsymbol{S})}{\partial \theta_i\partial \theta_j}\Big]$. Note both $\mathbb E$ above refer to expectation taken w.r.t. the true probability measure ($\mathbb E_{\theta^\star}$ when the model is correctly specified and is generally unknown under model misspecification). However, as shown in \cite{white1982maximum}, under model misspecification, \eqref{equivla} no longer holds and the (scaled) asymptotic variance as $\hat\theta_{\text{MLE}}$ approaches $\theta^\star_{\text{KL}}$ becomes
\begin{equation*}
 \mathcal C(\theta^\star_{\text{KL}}) = \mathcal B^{-1}(\theta^\star_{\text{KL}})\mathcal A(\theta^\star_{\text{KL}})\mathcal B^{-1}(\theta^\star_{\text{KL}}).
\end{equation*}
Moreover, \cite{white1982maximum} shows the discrepancies between empirical estimates of $\mathcal A(\theta^\star_{\text{KL}})$ %can be used to $\mathcal B(\theta^\star_{\text{KL}})$ 
and $\mathcal B(\theta^\star_{\text{KL}})$ can be used to
test model misspecification. However, this implies the conclusion drawn from our previous theorems no longer holds, as they rely on the comparison of the Fisher information of correctly specified models. In this case, it becomes less clear how to choose the appropriate models or consider estimation efficiency of MLE based on experiment design, as theoretical properties of $\mathcal A(\theta^\star_{\text{KL}})$ or $\mathcal B(\theta^\star_{\text{KL}})$ depends on the unknown, true distribution of the price trajectory. 

Thus, the presence of model misspecification will pose a limitation to the results of previous theorems. To better discuss the extent of such limitations and possible remedies, we revisit relevant concepts on model selection
from statistical learning/machine learning.

\subsection{Simulation}

In this section, we perform a simulation study for the model misspecification case. The estimator model follows the power-law kernel $G(s)=s^{-\gamma}$ and power-law impact $f(v)=v^{\delta}$ (so $\mu_\theta(t,v)= \frac{v^\delta t^{1-\gamma}}{1-\gamma}$),
where $\delta$ is the only parameter that needs to be estimated.
However, the sample-generating model does not have the same parametric format. The generating model has $G(s)=s^{-\gamma}$ and 
$f(v)= 1.5 \ln(1 + 0.5 v^2 + 0.7 v)$
(so $\mu^\star(t,v) = 1.5 \ln(1 + 0.5 v^2 + 0.7 v) \frac{t^{1-\gamma}}{1-\gamma}$).
We set $\sigma=0.2$, $\gamma=0.4$.
Each estimation has 300 samples, and each sample trajectory is assigned a random price trajectory length
$T\sim \mathrm{Uniform}(0.1, 0.15)$ and a random $v\sim \mathrm{Uniform}(0.3, 0.4)$. The simulation is repeated 40,000 times for each scenario.
The diffusion process is discretized into $1024$ time bins.
The first model is based on the average cost of VWAP as in \eqref{Jexam}.
Other models rely on the price trajectory as discussed in section \ref{samplestrat} with a subset or all points at 
$t_1=\frac{T}{8}, t_2=\frac{T}{2}, t_3=\frac{5T}{8}$ and $t_4=T$ relative to the corresponding path length $T$.
The error of estimation is the mean squared error of the cost across all sample paths.
Results are shown in Table \ref{tab:mis_specified}.
The conclusion is similar to the well-specified model in section \ref{subsec:path_numerical} and the simulation study therein that the cost-based method is worse than the price trajectory-based methods if earlier price trajectory points are selected (for example $t_1$ or $t_2$).
In addition, if more trajectory points are chosen or the earlier the points are selected, the performance will be better. The performance is evaluated using the squared error between the true cost and the estimated cost,
and the parenthesis shows the standard error of the mean error.
$X=J$ is the model based on the average cost of VWAP, 
others are based on the points along the price trajectory.
$t_1=\frac{T}{8}, t_2=\frac{T}{2}, t_3=\frac{5T}{8}$ and $t_4=T$.
Bold numbers indicate the method has a statistically significant better performance than using $J$ based on two sample t-tests.

\begin{table}[ht]
\centering
\scalebox{0.85}{
\begin{tabular}{|l|c|}
\hline
Estimator via $X$ &  mean error $\times 10^{-7}$ (SD $\times 10^{-9}$)  \\ \hline
$X = J$ &  $1.079$ ($4.831$)  \\ \hline

$X = S_{t_1}, S_T$ & $\textbf{1.049}$ ($4.612$) \\ \hline
$X = S_{t_2}, S_T$ & $\textbf{1.077}$ ($4.822$) \\ \hline
$X = S_{t_3}, S_T$ & $1.144$ ($5.294$) \\ \hline 

$X = S_{t_1}, S_{t_2}, S_T$ & $\textbf{1.038}$ ($4.538$) \\ \hline
$X = S_{t_1}, S_{t_3}, S_T$ & $\textbf{1.040}$ ($4.546$) \\ \hline
$X = S_{t_2}, S_{t_3}, S_T$ & $\textbf{1.073}$ ($4.788$) \\ \hline

$X = S_{t_1}, S_{t_2}, S_{t_3}, S_T$ & $\textbf{1.034}$ ($4.509$) \\ \hline
\end{tabular}}
\caption{Numerical comparison of models with misspecification.
}
\label{tab:mis_specified}
\end{table}

\subsection{Model Selection}
As in \eqref{whitepaper}, given $(T,v)\sim G_{\text{order}}$, we use $F$ to denote the true distribution of price trajectory $\boldsymbol S$ and $F(\theta)$ the one characterized under $S_t=S_0+\mu_{\theta}(t,v) + \sigma W_t$ for $S_t\in\boldsymbol S$ and $(T,v)\sim G_{\text{order}}$. Then, as discussed in the previous section, if we measure the risk as
\begin{equation*}
    R(\theta) = D_{\text{KL}}(F\lVert F(\theta)),
\end{equation*}
which quantifies the discrepancy between distributions fitted by the model and truth. Then, the model selection problem decomposes the risk as
\begin{align}\label{decompose}
    R(\theta)=&\underbrace{R(\theta)-\inf_{\theta\in\Theta}R(\theta)}_\text{estimation error} + \underbrace{\inf_{\theta\in\Theta}R(\theta)}_\text{approximation error}\nonumber\\
    =&\underbrace{R(\theta)-R(\theta^\star_{\text{KL}})}_\text{estimation error}+ \underbrace{R(\theta^\star_{\text{KL}})}_\text{approximation error}
\end{align}
In a correctly specified model, $\theta^\star_{\text{KL}} = \theta^\star$ and the estimation error is 0. The approximation error is directly related to the efficiency for estimating $\theta^\star$ and the results of our previous theorem, based on the Fisher information matrix, can be applied. When a model misspecification occurs, different factors/concerns can come into play and the landscape is less clear. For example,
\begin{itemize}
    \item If one is concerned with an accurate estimation of the price trajectory for various values of $(T,v)$, then the estimation using the price trajectory would increase the sample size, compared to using summary statistics, e.g., VWAP, which could lead to a positive effect on reducing estimation error. On the other hand, depending on how the distribution is misspecified, e.g., heavy tail versus light tail, the method using summary statistics could turn out to be more robust than using the entire trajectory. In this case, the previous comparison between using different summary statistics (VWAP or sampled price points) may or may not hold.
    \item If one is only concerned with an accurate estimation of summary statistics, e.g., VWAP. Then, this could lead to a reduced approximation error because KL divergence monotonically decreases under the non-invertible transformation (it is well known that KL divergence is invariant under invertible transformations, i.e., given random variable $X,Y$ and an invertible transformation $g(\cdot)$, we have $D_{\text{KL}}(X|Y)=D_{\text{KL}}(g(X)|g(Y))$. However, it can be shown that $D_{\text{KL}}(X|Y)\geq D_{\text{KL}}(g(X)|g(Y))$ when $g$ is non-invertible (see \cite{mena2018learning}). The transformation from the price trajectory $\boldsymbol S$ to summary statistics are typically non-invertible.). However, it could lead to an increase in approximation error due to reduced sample size or inefficiencies of estimation.
\end{itemize}

During practical implementation, we should take these issues into consideration. Typical methods include cross-validation upon different methods (or some hybrid of them) under the specified loss function or Bayesian inference based on carefully calibrated prior from previous data (or even online data, where one can update the posterior as the order is executed and adjust the execution). This is left for discussion and future research.

\bibliographystyle{apalike}
\bibliography{ref}

\newpage

\appendix
\section{Appendix}

\subsection{Asymptotic Inference}
\subsubsection{Proof of Basic Results and Review}
\begin{definition}[Differentiability in Quadratic Mean] A family of model $\{\mathbb P_\theta\}_{\Theta}$ with dominating measure $m$ is differentiable in quadratic mean (q.m.d.) if there exists a vector of measurable function $\dot{l}_\theta=(\dot{l}_{\theta,1},...,\dot{l}_{\theta,K})$ such that 
\begin{equation*}
    \int\bigg[\sqrt{p_{\theta+h}}-\sqrt{p_{\theta}}-\frac{1}{2}h^T\dot{l}_\theta\sqrt{p_{\theta}}\bigg]dm = o(\|h^2\|) \text{ as } h\rightarrow0.
\end{equation*}
    
\end{definition}

\begin{definition}[Regular Estimator]
    For a parametric family of probability measure $\{\mathbb P_{\theta}\}_\Theta$, a sequence of estimators $\{T_n\}_{n\geq0}$ is called $\textit{regular}$ at $\theta$ for estimating $\phi(\theta)$ if, for every $h$,
    \begin{equation*}
        \sqrt{n}\Big(T_n-\phi(\theta+\frac{h}{\sqrt{n}})\Big) \xrightarrow[\mathbb P_{\theta+\frac{h}{\sqrt{n}}}]{d} L_\theta,
    \end{equation*}
    i.e., the left hand side converges in distribution under $\mathbb P_{\theta+\frac{h}{\sqrt{n}}}$  to some random variable $L_\theta$ which do not depend on $h$.
\end{definition}

\begin{proof}[Proof of Proposition \ref{sdfadsfasdfas}]
    The proof mainly follows from basic results in \cite{van2000asymptotic} and \cite{bickel1993efficient}. Given assumption 3.3, we derive that $\theta\rightarrow\sqrt{p_\theta(\boldsymbol x)}$ is continuously for all $\boldsymbol x$. Assumption 4 ensures that $I_{\boldsymbol X}(\theta)$ is well-defined and continuous in $\theta$. Assumption 1 gives openness of $\Theta$. As as result, according to Lemma 7.6 in \cite{van2000asymptotic} or Proposition 1 in \cite{bickel1993efficient}, the model $\{\mathbb P_\theta\}_{\Theta}$ is q.m.d. with $\dot{l}_\theta$ given by $\frac{\dot{p}_\theta}{p_\theta}$. 

    Then, we can use Assumptions 2 and 3.4 to derive 
    \begin{equation*}
        |\log p_{\theta_1}(\boldsymbol x)-\log p_{\theta_2}(\boldsymbol x)|\leq L\|\theta_1-\theta_2\|,
    \end{equation*}
    for some $\theta_1,\theta_2$ in the neighborhood of $\theta^\star$. Assumption 3.1 gives consistency of $\hat\theta_n$ by Theorem 5.14 of \cite{van2000asymptotic}. These two results, in addition to assumption 3.5 and q.m.d. give \eqref{mleinv}, according to Theorem 7.12 in \cite{van2000asymptotic}. 

    Finally, \eqref{lammmmax} follows by Theorem 8.11 in \cite{van2000asymptotic}, using q.m.d. and assumption 4.
\end{proof}

The locally asymptotic minimax theorem is only one of many ways to argue the asymptotic optimality of MLE. For example, the  H\'ajek-LeCam convolution theorem and its variant show that for any \textit{regular} estimator, the MLE attains the ``best" limiting distribution asymptotically and achieves the lowest possible variance for any asymptotically regular sequence of the estimator. Moreover, any improvement over this limit distribution can only be made on a Lebesgue null set of parameters.  For a comprehensive review, see section 8 of \cite{van2000asymptotic}.

\subsubsection{Technical Lemma}\label{aiiii}
The appendix mainly consists of technical proofs and additional lemmas. The proof of the following support lemma can be found in \cite{zamir1998proof,zegers2015fisher}.
\begin{lemma}[Chain Rule for Fisher Information]\label{mutual} Let $X$ and $Y$ be two statistics of a given experiments $(\boldsymbol{S},\mathcal F, \{\mathbb P_\theta\}_{\theta\in\Theta})$, then 
\begin{align*}
    \mathcal I_{X,Y}(\theta)=\mathcal I_{X}(\theta)+ \mathcal I_{Y|X}(\theta),
\end{align*}
where $\mathcal I_{Y|X}(\theta)=\mathbb E_X[\mathcal I_{Y|X=x}(\theta)]$ is simply defined as
\begin{equation*}
    \mathcal I_{Y|X=x}(\theta)=\mathbb E_{\theta}\bigg[\Big(\frac{\partial l(Y|X=x,\theta)}{\partial \theta}\Big)\Big(\frac{\partial l(Y|X=x,\theta)}{\partial \theta}\Big)^T\bigg].
\end{equation*}
More generally, if we have $n$ statistics $X_1,X_2,...,X_n$, then 
\begin{equation*}
    \mathcal I_{X_1,X_2,...,X_n}(\theta)=\mathcal I_{X_1}(\theta)+\sum_{i=1}^{n-1}\mathcal I_{X_{i+1}|X_1,...,X_i}(\theta)
\end{equation*}
where $\mathcal I_{X_{i+1}|X_1,...,X_i}(\theta)=\mathbb E_{X_1,...,X_i}[\mathcal I_{X_{i+1}|X_1,...,X_i=x_1,...,x_i}(\theta)]$.
\end{lemma}

\subsection{Proofs}
\subsubsection{The Proof of Lemma \ref{data processing}}
\begin{proof}[Proof of Lemma \ref{data processing}]
The rigorous proof can be found in Theorem 7.2 of \cite{ibragimov2013statistical}. It can be checked that the assumptions in Theorem 7.2 are all satisfied. For a more concise proof, similar to the proof in \cite{zamir1998proof}, we can use Lemma \ref{mutual} to obtain:
\begin{equation}\label{mutualS}
    \mathcal I_{\boldsymbol{X}}(\theta)+\mathbb E_{\boldsymbol{X}}[\mathcal I_{\phi(\boldsymbol{X})|\boldsymbol{X}=\boldsymbol x}(\theta)]=\mathcal I_{\boldsymbol{X},\phi(\boldsymbol{X})}(\theta)=\mathcal I_{\phi(\boldsymbol{X})}(\theta)+\mathbb E_{\phi(\boldsymbol{X})}[\mathcal I_{\boldsymbol{X}|\phi(\boldsymbol{X})=\phi(\boldsymbol x)}(\theta)].
\end{equation}
Since $\phi(\boldsymbol{X})$ is a statistic of $\boldsymbol{X}$, it is clear that the conditional distribution of $\phi(\boldsymbol{X})$ given $\boldsymbol{X}$ is independent of $\theta$, implying that $\mathcal I_{\phi(\boldsymbol{X})|\boldsymbol{X}}(\theta)=0$. Consequently, \eqref{mutualS} now becomes 
\begin{equation*}
    \mathcal I_{\boldsymbol{X}}(\theta)=\mathcal I_{\phi(\boldsymbol{X})}(\theta)+\mathbb E_{\phi(\boldsymbol{X})}[\mathcal I_{\boldsymbol{X}|\phi(\boldsymbol{X})=\phi(\boldsymbol x)}(\theta)].
\end{equation*}
Since $\mathcal I_{\boldsymbol{X}|\phi(\boldsymbol{X})}(\theta)\succcurlyeq {0}$ (any Fisher information is positive semi-definite as it is the covariance matrix of the \textit{score} function), we have shown that
\begin{equation*}
   \mathcal I_{\boldsymbol{X}}(\theta)\succcurlyeq\mathcal I_{\phi(\boldsymbol{X})}(\theta).
\end{equation*}
Now, if $\phi(\boldsymbol{X})$ is a sufficient statistic, then the conditional distribution of $\boldsymbol{X}$ given $\phi(\boldsymbol{X})$ is independent of $\theta$ which leads to $\mathcal I_{\boldsymbol{X}|\phi(\boldsymbol{X})}(\theta)=0$ and the equality of $\mathcal I_{\boldsymbol{X}}(\theta)=\mathcal I_{\phi(\boldsymbol{X})}(\theta)$. On the other hand, if $\mathcal I_{\boldsymbol{X}}(\theta)=\mathcal I_{\phi(\boldsymbol{X})}(\theta)$, we must have $E_{\phi(\boldsymbol{X})}[\mathcal I_{\boldsymbol{X}|\phi(\boldsymbol{X})=\phi(\boldsymbol s)}(\theta)]=0$, but this does not necessarily imply $\phi(\boldsymbol{X})$ is a sufficient statistic. An interesting counter-example where an insufficient statistic preserves Fisher information can be found in \cite{kagan2005sufficiency}. To complete the proof, an additional condition on the continuous differentiability of density is enough to  establish the sufficiency of $\phi(\boldsymbol{X})$, the detailed proof again is shown in \cite{kagan2005sufficiency}.
\end{proof}
\subsubsection{The Proof of Lemma \ref{data refine}}
\begin{proof}[Proof of Lemma \ref{data refine}] The first claim follows directly from Lemma \ref{mutual} since $\mathcal I_{\boldsymbol Y|\boldsymbol X=\boldsymbol x}(\theta) \succcurlyeq {0}$ and $\mathbb E_{\boldsymbol X}[\mathcal I_{\boldsymbol Y|\boldsymbol X=\boldsymbol x}(\theta)] \succcurlyeq {0}$. For the second claim, notice by Lemma \ref{mutual},
\begin{equation*}
    \mathcal I_{\boldsymbol Y}(\theta)+ \mathbb E_{\boldsymbol Y}[\mathcal I_{\boldsymbol X|\boldsymbol Y=\boldsymbol y}(\theta)]=\mathcal I_{\boldsymbol X,\boldsymbol Y}(\theta)=\mathcal I_{\boldsymbol X}(\theta)+ \mathbb E_{\boldsymbol X}[\mathcal I_{\boldsymbol Y|\boldsymbol X=\boldsymbol x}(\theta)].
\end{equation*}
Since $\boldsymbol Y=\phi(\boldsymbol X)$ and $\phi(\cdot)$ is a bijective mapping, it is easy to check that $\mathbb E_{\boldsymbol X}[\mathcal I_{\boldsymbol Y|\boldsymbol X=\boldsymbol x}(\theta)]=E_{\boldsymbol Y}[\mathcal I_{\boldsymbol X|\boldsymbol Y=\boldsymbol y}(\theta)]=0$ which implies $\mathcal I_{\boldsymbol X}(\theta)=\mathcal I_{\boldsymbol Y}(\theta)$.
\end{proof}

\subsection{The Joint Distribution of $(I,J)$ in Almgren-Chriss Model} 

In this derivation, we fix $(v,T)$ as constant and derive the distribution of $(I,J)$ conditional on $(v,T)$. Notice that, after the affine transformation $P_t=\frac{S_t-S_0}{S_0}$, we have 
\begin{align*}
     P_t=&g(v)t+h(v)+\sigma W_t, \text { when $t\leq T$} \nonumber\\
     P_t=&g(v)T+\sigma W_t, \text{ when $t>T$. }
\end{align*}
Following the notations in \cite{almgren2005direct}, we let 
$\bar S=\frac{\int_0^t S_t\dot{x}_tdt}{X}=\frac{\int_0^T S_tdt}{T}$. First, following \eqref{iandj}, we have
\begin{align*}
    I=P_{T_\text{post}}=Tg(v)+\sigma W_{T_{\text{post}}},  \implies &I\sim \mathcal N(Tg(v),\sigma^2T_\text{post})\nonumber\\
    J=\frac{\bar S-S_0}{S_0}=\frac{1}{T}\int_0^T P_tdt=\frac{T}{2}g(v)+h(v)+\sigma \frac{1}{T}\int_0^T W_tdt, \implies& J\sim \mathcal N \big(\frac{T}{2}g(v)+h(v),\sigma^2\frac{T}{3}\big),
\end{align*}
where we know $\int_0^T W_tdt=\int_0^T (T-t)dW_t\sim\mathcal N(0,\int_0^T (T-t)^2dt)=\mathcal N(0,\frac{T^3}{3})$ by It\^{o}'s Lemma. Since $I$ and $J$ are clearly jointly normal and we have known their mean and variance, we only need to calculate their correlation to determine the joint distribution. First, we can calculate $\text{Cov}(\int_0^T W_tdt, W_{T_\text{post}})$:
\begin{align*}
    \text{Cov}(\int_0^T W_tdt, W_{T_\text{post}})=&\text{Cov}(\int_0^T W_tdt, W_T+ W_{T_\text{post}}-W_T)\nonumber\\
    =&\text{Cov}(\int_0^T W_tdt, W_T)\nonumber\\
    =&\mathbb EW_T\int_0^T W_tdt=\int_0^T\mathbb E[W_TW_t]dt=\int_0^T tdt=\frac{T^2}{2}.
\end{align*}
On the other hand, again using It\^{o}'s Lemma, we have  $\text{Var}(\int_0^T W_tdt)=\frac{T^3}{3}$ and $\text{Var}(W_{T_\text{post}})=T_\text{post}$. Thus, it can be checked that $\text{Corr}(\int_0^T W_tdt, W_{T_\text{post}})=\sqrt{\frac{3T}{4T_\text{post}}}$. This allows characterizing the joint distribution of $I$ and $J$ completely, equivalent as in equation (2) of \cite{almgren2005direct}, by noting that there exist two independent standard normal random variables $\xi_1,\xi_2$ where
\begin{align}\label{bigau}
        W_{T_\text{post}}=&\sqrt{T_\text{post}}\xi_1  \nonumber\\
        \int_0^T W_tdt=&\frac{T^2}{2\sqrt{T_\text{post}}}\xi_1+  T\sqrt{\frac{T}{12}(4-3\frac{T}{T_\text{post}})} \cdot \xi_2
        \end{align}
which can then be shown to be equivalent as \eqref{iandj}.

\subsubsection{Proof of Lemma \ref{powerregular}}
Based on assumption \ref{modelspec}, \ref{orderindepG} on the lower and upper boundedness of $v$ and $T$, and the specific form of \eqref{power}, one can routinely check assumption 3.1,3.3,3.4 are satisfied. For assumption 3.4, we note that, by \eqref{iandj} the  $\Sigma_{I,J}(T)$ in Gaussian variable \eqref{gassexp} is given by
\begin{equation*}
     \Sigma_{I,J}(T)\triangleq\sigma^2\begin{bmatrix}
 T_\text{post}& -\frac{T_\text{post}-T}{2} \\
-\frac{T_\text{post}-T}{2}  & \frac{T_\text{post}}{4}-\frac{T}{6}
\end{bmatrix}.
\end{equation*}
Define the ratio $\tau = \frac{T_{\text{post}}}{T}-1$, we note $\tau$ is random but $\tau>0$. Then, one can write
\begin{equation*}
     \Sigma_{I,J}(T)\triangleq\sigma^2T\begin{bmatrix}
 1+\tau& -\frac{\tau}{2} \\
-\frac{\tau}{2}  & \frac{1+\tau}{4}-\frac{1}{6}
\end{bmatrix}\triangleq \sigma^2T\Sigma_\tau.
\end{equation*}
The smallest eigenvalue $\lambda_{1}(\tau)$ of $\Sigma_\tau$ can be explicitly computed by $a-\sqrt{a^2-b}$ where $a=\frac{5(1+\tau)}{8}-\frac{1}{12}$ is the mean of diagonal elements and $b=\frac{(1+\tau)^2}{4}-\frac{(1+\tau)}{6}-\frac{\tau^2}{4}$ is the determinant. For $\tau>0$, $\lambda_1(\tau)$ is lower bounded by $\lambda_1(0)$ which is a positive constant. Thus, by assumption \ref{orderindepG}, we have 
\begin{equation*}
    \Sigma_{I,J}(T)=\sigma^2T\Sigma_\tau\succcurlyeq \sigma^2T_L\lambda_1(0)\textbf{I}.
\end{equation*}
For assumption 3.5, the uniqueness of $\theta^\star$ can be violated only if there exists $\theta'=(\gamma',\eta',\alpha',\beta')$ such that
\begin{align}\label{powercheckuni}
    \gamma' v^{\alpha'} T=&\gamma^\star v^{\alpha^\star} T\nonumber\\
    \eta' v^{\beta'} = &\eta^\star v^{\beta^\star},
\end{align}
for all $(v,T)$, $G_{\text{order}}$ almost surely. This is impossible when the support of the marginal distribution of $v$ from $G_{\text{order}}$ has infinite cardinality (or finite and sufficiently large). For assumption, recall from \eqref{FI-IJ} that,
\begin{align*}
    \mathcal I_{I,J}(\theta)=&\mathbb E_{(v,T)\sim G_{\text{order}}}\bigg[\mathcal J^T_{I,J}(\theta, T,v)\cdot\sigma^{-2}\begin{bmatrix}
 \frac{1}{T}\frac{T_\text{post}/4-T/6}{T_\text{post}/3-T/4}& \frac{1}{T}\frac{T_\text{post}/2-T/2}{T_\text{post}/3-T/4} \\
\frac{1}{T}\frac{T_\text{post}/2-T/2}{T_\text{post}/3-T/4}  &  \frac{1}{T}\frac{T_\text{post}}{T_\text{post}/3-T/4}
\end{bmatrix}\cdot\mathcal J_{I,J}(\theta, T,v)\bigg],\\
\mathcal J_{I,J}(\theta, T,v)=&\begin{bmatrix}
 Tv^\alpha& 0 &T\gamma v^{\alpha}\ln(v)&0\\ 
0  &  v^\beta&0&\eta v^{\beta}\ln(v)
\end{bmatrix}.
\end{align*}

In order to show $\mathcal I_{I,J}(\theta)\succ 0$, it suffices to show that, for any $\boldsymbol x\in\mathbb R^K$ with $\|\boldsymbol x\|=1$, we have
\begin{equation*}
    \boldsymbol x^T \mathcal I_{I,J}(\theta) \boldsymbol x >0.
\end{equation*}
The non-singularity of $\Sigma_{I,J}$ and form of $\mathcal I_{I,J}(\theta)$ ensures $\boldsymbol x^T \mathcal I_{I,J}(\theta) \boldsymbol x \geq 0$ for all such $\boldsymbol x$. If there exists one such $\boldsymbol x$ that $\boldsymbol x^T \mathcal I_{I,J}(\theta) \boldsymbol x =0$, we must have 
\begin{equation*}
    \boldsymbol x^T \mathcal J_{I,J}(\theta, T,v)\Sigma_{I,J}^{-1}(T)\mathcal J_{I,J}(\theta, T,v)\boldsymbol x =0
\end{equation*}
for all $(v,T)\sim G_{\text{order}}$ almost surely. To see why, simply note that, for a random variable $X\geq 0$, $\mathbb E[X]=0$ if and only if $\mathbb P(X=0)=1$. Due to the non-singularity of $\Sigma_{I,J}(T)$, this requires $\mathcal J_{I,J}(\theta, T,v)\boldsymbol x=0$ for all $(v,T)\sim G_{\text{order}}$ almost surely:
\begin{align*}
    x_1 Tv^\alpha+x_3T\gamma v^\alpha \ln(v) =&0 \nonumber\\
    x_2 v^\beta + x_4\eta v^\beta \ln(v) =0 
\end{align*}
for all $(v,T)\sim G_{\text{order}}$ almost surely. This is not possible for any $\boldsymbol x=1$, when the support of marginal distribution of $v$ from $G_{\text{order}}$ has infinite cardinality (e.g., $x_1+x_3\ln(v)=0$ for all $v$ iff $x_1=x_3=0$).

\subsubsection{Derivation of $\mathcal I_{P_T,P_{T_\text{post}}}(\theta)$ in Almgren-Chriss Model}
In this derivation, we fix $(v,T)$ as constant. By Lemma \ref{data refine}, we have $\mathcal I_{P_T,P_{T_\text{post}}}(\theta)=\mathcal I_{P_{T_\text{post}},P_{T_\text{post}}-P_T}(\theta)$. To calculate $\mathcal I_{P_{T_\text{post}},P_{T_\text{post}}-P_T}(\theta)$, note that:
\begin{align*}
    \begin{pmatrix}P_{T_\text{post}}\\P_T-P_{T_\text{post}}\end{pmatrix} \sim\mathcal N
        \begin{pmatrix}
        \begin{pmatrix}Tg(v;\theta)\\h(v;\theta)\end{pmatrix},
          \sigma^2\begin{bmatrix}
 T_\text{post}& -(T_\text{post}-T) \\
-(T_\text{post}-T) & T_\text{post}-T
\end{bmatrix}
    \end{pmatrix}.
  \end{align*}
The rest follows from standard calculation based on Lemma \ref{jacob} and let $(v,T)\sim G_{\text{order}}$.

\subsubsection{Proof of Theorem \ref{acfinal}}\label{proofofthmmain}
\begin{proof}[Proof of Theorem \ref{acfinal}]

Throughout the proof, we assume $(v,T)$ is fixed, i.e., $G_{\text{order}}$ is a point mass. Results for generic $G_{\text{order}}$ follow trivially since, for any random variable, $X\geq Y$ implies $\mathbb E[X]\geq \mathbb E [Y]$. 

Following \eqref{suff}, we can see that, $(P_T,\frac{P_{t'_1}}{t'_1}-\frac{P_{T_{\text{post}}}-P_T}{T_{\text{post}-T}})$ is sufficient for experiment based on $\boldsymbol{S}$, thus we also have
\begin{equation}\label{2to3}
    \mathcal I_{\boldsymbol{S}}(\theta)=\mathcal I_{P_T,\phi_{\tau_1}(P_{t_1},P_T, P_{T_{\text{post}}})}(\theta).
\end{equation}
where $\phi_{\tau_1}(P_{t_1},P_T, P_{T_{\text{post}}})=\frac{P_{t_1}}{t_1}-\frac{P_{T_{\text{post}}}-P_T}{T_{\text{post}-T}}$ and $\tau_1$ is fixed for the experiment. Together with \eqref{3point}, we can see that, for $(P_t,P_T,P_{T_\text{post}})$,
\begin{equation}\label{3is2}
    \mathcal I_{P_t,P_T,P_{T_\text{post}}}(\theta)=\mathcal I_{P_T,\phi_\tau(P_t,P_T,P_{T_\text{post}})}(\theta).
\end{equation}
Now, we can check
\begin{align}\label{newmu}
    \begin{pmatrix}P_T\\t\cdot\phi_\tau(P_t,P_T,P_{T_\text{post}})\end{pmatrix} \sim\mathcal N
        \begin{pmatrix}
        \bar{\mu}(\theta),
          \bar{\Sigma},
    \end{pmatrix} 
    \end{align}
where
    \begin{align}
 \bar{\mu}(\theta)\triangleq\begin{pmatrix}Tg(v;\theta)+h(v;\theta)\\tg(v;\theta)+(1+\frac{t}{T_{\text{post}}-T})h(v;\theta)\end{pmatrix},  \bar\Sigma\triangleq\sigma^2\begin{bmatrix}
 T& t \\
t  & t+\frac{t^2}{T_{\text{post}}-T}
\end{bmatrix}.
  \end{align}
 Recall the definition in \eqref{iandj} where  
 \begin{align*}
 {\mu}(\theta)\triangleq\begin{pmatrix}Tg(v;\theta)\\h(v;\theta)\end{pmatrix}.
  \end{align*}
Now define the matrix 
\begin{align*}
   A\triangleq\begin{bmatrix}
 1& 1 \\
\frac{t}{T}  & 1+\frac{t}{T_{\text{post}}-T}
\end{bmatrix}.
\end{align*}
It is straightforward to check that for all $t\in(0,T)$ and $T_\text{post}>T$, $A$ is non-singular and 
\begin{equation*}
    \bar{\mu}(\theta)=A\mu(\theta).
\end{equation*}
Thus, if we define $U\in\mathbb R^2$ as
\begin{equation}
    U\triangleq A^{-1} \begin{bmatrix}
 P_T \\
t\cdot\phi_\tau(P_t,P_T,P_{T_\text{post}}) 
\end{bmatrix},
\end{equation}
we can check that 
\begin{equation}\label{udis}
    U\sim\mathcal N \big(\mu(\theta),A^{-1}\bar{\Sigma}(A^{-1})^T\big)
\end{equation}
Since $A$ is non-singular, we know from Lemma \ref{data refine} and \eqref{3is2} that
\begin{align}
  \mathcal I_U(\theta)=\mathcal I_{P_T,t\cdot\phi_\tau(P_t,P_T,P_{T_\text{post}})}(\theta)=\mathcal I_{P_t,P_T,P_{T_\text{post}}}(\theta).
\end{align}
Thus, it suffices to calculate $\mathcal I_U(\theta)$, for which we use Lemma \ref{jacob} on \eqref{udis} to obtain:
\begin{equation*}
  \mathcal I_{P_t,P_T,P_{T_{\text{post}}}}(\theta)=\mathcal I_U(\theta)=\mathcal J^T_{ I,J}(\theta) \big(A^T\bar\Sigma^{-1}A\big)\mathcal J_{ I,J}(\theta),
\end{equation*}
where we can carry out direct calculations to get 
\begin{equation*}
    A^T\bar\Sigma^{-1}A=\sigma^{-2}\begin{bmatrix}
 \frac{1}{T}& \frac{1}{T} \\
\frac{1}{T}  & \frac{1}{t}+\frac{1}{T_{\text{post}}-T}
\end{bmatrix}.
\end{equation*}
Recall from equation \eqref{FI-IJ} that
\begin{equation*}
    \mathcal I_{I,J}(\theta)=\mathcal J^T_{I,J}(\theta)\cdot\sigma^{-2}\begin{bmatrix}
 \frac{1}{T}\frac{T_\text{post}/4-T/6}{T_\text{post}/3-T/4}& \frac{1}{T}\frac{T_\text{post}/2-T/2}{T_\text{post}/3-T/4} \\
\frac{1}{T}\frac{T_\text{post}/2-T/2}{T_\text{post}/3-T/4}  &  \frac{1}{T}\frac{T_\text{post}}{T_\text{post}/3-T/4}
\end{bmatrix}\cdot\mathcal J_{ I,J}(\theta).
\end{equation*}
In order to show $\mathcal I_U(\theta)\succcurlyeq\mathcal I_{I,J}(\theta)$ for generic form of $h,g$, it suffices to show that
\begin{equation}\label{comp}
    \begin{bmatrix}
 \frac{1}{T}& \frac{1}{T} \\
\frac{1}{T}  & \frac{1}{t}+\frac{1}{T_{\text{post}}-T}
\end{bmatrix}\succcurlyeq \begin{bmatrix}
 \frac{1}{T}\frac{T_\text{post}/4-T/6}{T_\text{post}/3-T/4}& \frac{1}{T}\frac{T_\text{post}/2-T/2}{T_\text{post}/3-T/4} \\
\frac{1}{T}\frac{T_\text{post}/2-T/2}{T_\text{post}/3-T/4}  &  \frac{1}{T}\frac{T_\text{post}}{T_\text{post}/3-T/4}
\end{bmatrix}.
\end{equation}
Define $\tau_1\triangleq\frac{T_\text{post}}{T}$ and $\tau_2\triangleq\frac{T}{t}$ (notice $\tau_1,\tau_2>1$), showing \eqref{comp} is equivalent as showing:
\begin{equation*}
    B\triangleq\begin{bmatrix} 1-
 \frac{\tau_1/4-1/6}{\tau_1/3-1/4}& 1-\frac{\tau_1/2-1/2}{\tau_1/3-1/4} \\
1-\frac{\tau_1/2-1/2}{\tau_1/3-1/4}  & \tau_2+\frac{1}{\tau_1 -1}-\frac{\tau_1}{\tau_1/3-1/4}
\end{bmatrix} \succcurlyeq {0}.
\end{equation*}
To check the positive semi-definiteness of the above matrix, we use Sylvester's criterion, which reduces to checking:
\begin{equation}\label{check}
  \begin{cases}
			[B]_{11}=1-
 \frac{\tau_1/4-1/6}{\tau_1/3-1/4}\geq 0, & \text{}\\
            [B]_{22}=\tau_2+\frac{1}{\tau_1 -1}-\frac{\tau_1}{\tau_1/3-1/4} \geq 0 & \text{} \\
            \text{det}(B) \geq {0} & \text{}
		 \end{cases}.
		 \end{equation}
For $[B]_{11}$, it is easy to check that $[B]_{11}>0$ whenever $\tau_1>1$, which is satisfied by definition of $\tau_1\triangleq\frac{T_{\text{post}}}{T}$. For $\text{det}(B)$, we wanted to check for \begin{equation}\label{det}
  \text{det}(B)= (1-
 \frac{\tau_1/4-1/6}{\tau_1/3-1/4})(\tau_2+\frac{1}{\tau_1 -1}-\frac{\tau_1}{\tau_1/3-1/4})-(1-\frac{\tau_1/2-1/2}{\tau_1/3-1/4})^2 \geq 0.
\end{equation}
By multiplying a strictly positive term $12^2(\tau_1/3-1/4)^2$, one can check \eqref{det} is greater than 0 if and only if 
\begin{equation*}
    (\tau_2(\tau_1-1)+1)(4\tau_1-3)-12\tau_1(\tau_1-1)-(2\tau_1-3)^2 \geq 0
\end{equation*}
which simplifies to
\begin{equation*}
    (\tau_2-4)(4\tau_1^2-7\tau_1+3)\geq 0.
\end{equation*}
One can check $4\tau_1^2-7\tau_1+3 > 0$ for all $\tau_1>1$. Thus, $\text{det} (B) \geq 0 $ if and only if $\tau_2\geq 4$.
So far, the validity of \eqref{check} requires $\tau_1>1$ and $\tau_2\geq 4$. Finally, for $[B]_{22}$, by multiplying a strictly positive term $12(\tau_1-1) (\tau_1/3-1/4)$, it can be checked that $[B]_{22} \geq 0$ if and only if 
\begin{equation}\label{quad1}
    4(\tau_2-3)\tau_1^2+(16-7\tau_2)\tau_1+3(\tau_2-1)\geq 0.
\end{equation}
For $\tau_2\geq 4$, \eqref{quad1} is a quadratic function in $\tau_1$ with positive quadratic coefficient. The discriminant of \eqref{quad1} is 
\begin{equation}\label{discri}
    \Delta=(16-7\tau_2)^2-48(\tau_2-3)(\tau_2-1),
\end{equation}
which satisfies $\Delta<0$ (no roots) if and only if $4<\tau_2<28$. Thus, for $4<\tau_2<28$, $[B]_{22} > 0$ for all $\tau_1 >1$. When $\tau_2=4$, we have $\Delta=0$ and $[B]_{22} \geq 0$ for all $\tau_1 >1$. Thus,  On the other hand, when $\tau_2\geq 28$, \eqref{quad1} has two roots (possibly with multiplicity):
\begin{equation}\label{roots}
    \frac{7\tau_2-16 \pm \sqrt{\Delta}}{8\tau_2-24}.
\end{equation}
However, it can be checked that, when $\tau_2\geq 28$, $(\tau_2-8)^2-(\tau_2-4)(\tau_2-28)=16(\tau_2-3)>0$, thus we have
\begin{equation*}
    \sqrt{\Delta}=\sqrt{(\tau_2-4)(\tau_2-28)}\leq \sqrt{(\tau_2-8)^2} =\tau_2-8.
\end{equation*}
Consequently, for $\tau_2 \geq 28$, we have for \eqref{roots} that
\begin{equation*}
   \frac{7\tau_2-16 + \sqrt{\Delta}}{8\tau_2-24} \leq 1,
\end{equation*}
which implies the roots of \eqref{quad1} is no greater than 1 and $[B]_{22} > 0$ for $\tau_1>1$. Thus, in summary, \eqref{check} is satisfied as long as $\tau_1>1$ and $\tau_2\geq 4$. The rest follows by letting $(v,T)\sim G_{\text{order}}$.
\end{proof}

\subsubsection{Technical Proof in Theorem \ref{minisuff}}\label{condim}
Fix $(v,T)$. We prove the case for the kernel $G(s)\propto s^{-\gamma}$ and the case for $G(s)\propto l_0(l_0+s)^{-\gamma}$ or $G(s)\propto e^{-\rho s}$ is similar. For simplicity, we assume $G(s)=s^{-\gamma}$ so that $\eta_i(\gamma)=\int_{t_{i-1}}^{t_i}G(s;\gamma)ds=\frac{1}{1-\gamma}(t_i^{1-\gamma}-t_{i-1}^{1-\gamma})=\frac{(\Delta t)^{1-\gamma}}{1-\gamma}(i^{1-\gamma}-(i-1)^{1-\gamma})$. We want to show that the convex hull of $\{\eta(\gamma)\}_{\gamma\in\Gamma}$ has dimension $N$, or equivalently: there exists $\{\gamma_0,\gamma_1,...,\gamma_N\}\subseteq\Theta$ such that $\{\eta(\gamma_i)-\eta(\gamma_0)\}_{1\leq i \leq N}$ is linearly independent. Define $a_i(\gamma)\triangleq \frac{(\Delta t)^{1-\gamma}}{1-\gamma}(i^{1-\gamma}-(i-1)^{1-\gamma})$, we then want to show there exists $(\gamma_0,\gamma_1,...,\gamma_N)\in\Gamma^{N+1}$, such that the matrix
\begin{equation}\label{matrixgamma}
    A(\gamma_0,\gamma_1,...,\gamma_N)\triangleq \begin{bmatrix} 
    a_{1}(\gamma_1)-a_{1}(\gamma_0) & \dots  & a_{1}(\gamma_N)-a_{N}(\gamma_0)\\
    \vdots & \ddots & \vdots\\
    a_{N}(\gamma_1)-a_{N}(\gamma_0) & \dots  & a_{N}(\gamma_N)-a_{N}(\gamma_0)
    \end{bmatrix}
\end{equation}
has full rank $N$. To show this, it suffices to show there exists $(\gamma_0,\gamma_1,...,\gamma_N)\in\Gamma^{N+1}$ such that
\begin{equation}\label{detwant}
  \textup{det} (A (\gamma_0,\gamma_1,...,\gamma_N))\neq 0 .
\end{equation}
 To show this, note that the dimension of the convex hull of $\{\eta(\gamma)\}_{\gamma\in\Gamma}$ is also the dimension of the smallest affine set (affine hull) that contains $\{\eta(\gamma)\}_{\gamma\in\Gamma}$. There are two scenarios:
 
 \textbf{Scenario 1} If $\{\eta(\gamma)\}_{\gamma\in\Gamma}$ cannot be contained in a linear hyper-plane of dimension $N-1$, then the convex hull of $\{\eta(\gamma)\}_{\gamma\in\Gamma}$ has dimension $N$, therefore there must exist $(\gamma_0,\gamma_1,...,\gamma_N)\in\Gamma^{N+1}$ that satisfies \eqref{detwant}.
 
 \textbf{Scenario 2} If $\eta(\gamma)$ can be contained in a linear hyper-plane of dimension $N-1$, then there exists some $\boldsymbol w\in\mathbb R^{N+1}=(w_0,w_1,...,w_N)$ such that $\boldsymbol w \neq \boldsymbol 0$ and 
 \begin{equation}\label{keycontra}
    b(\gamma)\triangleq\sum_{i=1}^N w_i a_i(\gamma)-w_0=0
 \end{equation}
for all $\gamma\in\Gamma$. First, we note that $a_i(\gamma)$ is an analytic function on not just $\Gamma$, but also $(-\infty,1)$. It follows from \cite{mityagin2020zero} that the zero set of $b(\gamma)$:
\begin{equation*}
    B=\{\gamma\in(-\infty,1): b(\gamma)=0\}
\end{equation*}
either has Lebesgue measure 0 or is the entire $(-\infty,1)$. First, we assume, for the sake of contradiction, that $B=(-\infty,1)$, then
\begin{equation}\label{contr1}
    b(\gamma)=0 \Leftrightarrow \sum_{i=1}^N w_i(t_i^{1-\gamma}-t_{i-1}^{1-\gamma})=w_0(1-\gamma),
\end{equation}
for all $\gamma\in(-\infty,1)$.
Then we can define $i^\star=\max\{i: 0\leq i \leq N, w_i\neq 0\}$. Note that $i^\star$ is well-defined since $\boldsymbol w \neq \boldsymbol 0$. Moreover, $i^\star \geq 1$ since otherwise we would would have 
\begin{equation*}
    0=\sum_{i=1}^N w_i(t_i^{1-\gamma}-t_{i-1}^{1-\gamma})\neq w_0(1-\gamma),
\end{equation*}
which contradicts \eqref{contr1}. Thus, we can rewrite \eqref{contr1} as
\begin{equation}\label{contr2}
    \sum_{i=1}^{i^\star} w_i(t_i^{1-\gamma}-t_{i-1}^{1-\gamma})= w_0(1-\gamma),
\end{equation}
for all $\gamma\in(-\infty,1)$.Divide both side by $t_{i^\star}^{1-\gamma}$ and let $\gamma\downarrow-\infty$, the left side of \eqref{contr2} goes to $w_{i^\star}$ while the right side of \eqref{contr2} becomes
\begin{equation*}
    \lim_{\gamma\downarrow-\infty} \frac{w_0(1-\gamma)}{t_{i^\star}^{1-\gamma}},
\end{equation*}
which cannot converge to a constant other than 0 (i.e., not $w_{i^\star}$), regardless of the value of $t_{i^\star}$. Thus, we arrive at a contradiction and we conclude that $B$ cannot be $(-\infty,1)$. Thus, $B$ must have Lebesgue measure 0. However, we assumed that $\Gamma$ contains an open set, which implies the Lebesgue of $\Gamma$ on $\mathbb R$ is greater than 0. As a result, there exists $\gamma\in\Gamma$ (in fact there exists a set of positive measures with such points) such that $b(\gamma)\neq 0$, which contradicts \eqref{keycontra}. Thus, \textbf{Scenario 2} is not possible, we must only have \textbf{Scenario 1}, which completes the proof.

\subsubsection{Proof of Theorem \ref{calif}} \label{dsfds}

Fix $(v,T)$. Denote $\Delta S_i\triangleq S_{t_i}-S_{t_{i-1}}$ and $\Delta G_i=\int_{t_{i-1}}^{t_i}G(s)ds$. It follows from \eqref{mini} and Theorem \ref{factor} that $$\phi(\boldsymbol S_{\text{full}})\triangleq\sum_{i=1}^N (S_{t_i}-S_{t_{i-1}})\int_{t_{i-1}}^{t_i} G(s)ds=\sum_{i=1}^N \Delta S_i\Delta G_i,$$ 
is a sufficient statistic. Since $\phi(\boldsymbol S_{\text{full}}) \sim \mathcal N(f(v;\theta)\sum_{i=1}^N(\Delta G_i)^2, \sigma^2\Delta t\sum_{i=1}^N(\Delta G_i)^2)$, as $\Delta t\rightarrow 0$, it follows from Lemma \ref{jacob} that 
\begin{equation}\label{fullintee}
    \mathcal I_{\boldsymbol S_{\text{full}}}(\theta)=\mathcal I_{\phi(\boldsymbol S_{\text{full}})}(\theta)=\frac{(\frac{\partial f}{\partial \theta})^2}{\sigma^2}\sum_{i=1}^N\frac{\Delta G^2_i}{\Delta t}\rightarrow \frac{(\frac{\partial f}{\partial \theta})^2}{\sigma^2}\int_0^T G^2(t)dt,
\end{equation}
which proves \eqref{41}. Now, based on definition \ref{disJ}, we have
\begin{equation}\label{Jnew}
   J=\sum_{i=1}^{N}S_{t_i}\Delta t=\sum_{i=1}^{N}\Delta S_i(T-t_i+\Delta t).
\end{equation}  
Now, using Lemma \ref{mutual}, we have
\begin{equation}\label{Jfull}
    \mathcal I_J(\theta)=\mathcal I_{\phi(\boldsymbol S_{\text{full}})}(\theta)-\mathbb E_{x\sim J}[\mathcal I_{\phi(\boldsymbol S_{\text{full}})|J=x}(\theta)].
\end{equation}
However, based on \eqref{Jnew}, the joint distribution of $(\phi(\boldsymbol S_{\text{full}}), J)$ is 
\begin{align*}
    \begin{pmatrix}\phi(\boldsymbol S_{\text{full}})\\J\end{pmatrix}
    =&\begin{pmatrix}\sum_{i=1}^N \Delta S_i\Delta G_i\\\sum_{i=1}^{N}\Delta S_i(T-t_i+\Delta t)\end{pmatrix}\nonumber\\ \sim&\mathcal N
        \begin{pmatrix}
        f(v;\theta)\begin{pmatrix}\sum_{i=1}^N (\Delta G_i)^2\\\sum_{i=1}^N(T-t_i+\Delta t)\Delta G_i\end{pmatrix},
          \sigma^2\Delta t\begin{bmatrix}
 \sum_{i=1}^N (\Delta G_i)^2& \sum_{i=1}^N(T-t_i+\Delta t)\Delta G_i \\
\sum_{i=1}^N(T-t_i+\Delta t)\Delta G_i & \sum_{i=1}^N(T-t_i+\Delta t)^2
\end{bmatrix}
    \end{pmatrix}.
  \end{align*}
Then, standard Gaussian calculation gives conditional distribution:
\begin{align*}
    \phi(\boldsymbol S_{\text{full}})|J=x\sim &\mathcal N \Bigg(f(v;\theta)\sum_{i=1}^N (\Delta G_i)^2+\frac{\sum_{i=1}^N(T-t_i+\Delta t)\Delta G_i}{\sum_{i=1}^N(T-t_i+\Delta t)^2}\Big(x-f(v;\theta)\cdot\sum_{i=1}^N(T-t_i+\Delta t)\Delta G_i\Big),\nonumber\\
    &\sum_{i=1}^N (\Delta G_i)^2-\frac{(\sum_{i=1}^N(T-t_i+\Delta t)\Delta G_i)^2}{\sum_{i=1}^N(T-t_i+\Delta t)^2}\Bigg).
\end{align*} 
Now, based on Lemma \ref{jacob}, 
we have, as $\Delta t\rightarrow 0$,
\begin{align*}
    \mathcal I_{\phi(\boldsymbol S_{\text{full}})|J}(\theta)=&\frac{(\frac{\partial f}{\partial \theta})^2}{\sigma^2}\Big(\sum_{i=1}^N \frac{(\Delta G_i)^2}{\Delta t}-\frac{(\sum_{i=1}^N(T-t_i+\Delta t)\Delta G_i)^2}{\sum_{i=1}^N(T-t_i+\Delta t)^2\Delta t}\Big)\nonumber\\
    \rightarrow& \frac{(\frac{\partial f}{\partial \theta})^2}{\sigma^2}\Big(\int_0^T G^2(t)dt-\frac{(\int_0^T(T-t)G(t)dt)^2}{\int_0^t(T-t)^2dt}\Big)\nonumber\\
    =& \frac{(\frac{\partial f}{\partial \theta})^2}{\sigma^2}\Big(\int_0^T G^2(t)dt-\frac{3}{T^3}(\int_0^T(T-t)G(t)dt)^2\Big),
\end{align*}
which, combined with \eqref{Jfull} and \eqref{41}, gives \eqref{43}:
\begin{equation*}
    \mathcal I_{J}(\theta)\rightarrow \frac{(\frac{\partial f}{\partial \theta})^2}{\sigma^2}\cdot\frac{3}{T^3}\Big(\int_0^T G(t)(T-t)dt\Big)^2.
\end{equation*}
The fact that $\mathcal I_{J}(\theta)=\frac{(\frac{\partial f}{\partial \theta})^2}{\sigma^2}\cdot\frac{3}{T^3}\Big(\int_0^T G(t)(T-t)dt\Big)^2$ can be directly calculated using Proposition \ref{fishtvtv}, and note that by Fubini theorem that $\int_0^T\int_0^t G(s)ds=\int_0^T G(t)(T-t)dt$.

Finally, to prove \eqref{42}, we first write $\boldsymbol S_{\text{partial}}$ as $\boldsymbol S$ for ease of notation. Note that it follows from Lemma \ref{data refine} that 
\begin{equation}\label{kkkk}
    \mathcal I_{\boldsymbol S}(\theta)=\mathcal I_{S_{t_1^{\mathcal K}}-S_{t_0^{\mathcal K}},S_{t_2^{\mathcal K}}-S_{t_1^{\mathcal K}},...,S_{t_{N_{\mathcal K}}^{\mathcal K}}-S_{t_{{N_{\mathcal K}}-1}^{\mathcal K}}}(\theta).
\end{equation}
Now denote $\Delta S_i^{\mathcal K}=S_{t_i^{\mathcal K}}-S_{t_{i-1}^{\mathcal K}}$ for $1\leq i \leq {N_{\mathcal K}}$. Since $\Delta S_i^{\mathcal K}$ is independent with $\Delta S_j^{\mathcal K}$ for $i\neq j$, it follows that \eqref{kkkk} becomes \begin{equation}\label{ttp4}
    \mathcal I_{\boldsymbol S}(\theta)=\mathcal I_{\{\Delta S_i^{\mathcal K}\}_{1\leq i \leq n}}(\theta)=\sum_{i=1}^{N_{\mathcal K}} \mathcal I_{\Delta S_i^{\mathcal K}}(\theta).
\end{equation}
On the other hand, fix $1\leq i \leq {N_{\mathcal K}}$, define $T_i=t_{i}^{\mathcal K}-t_{i-1}^{\mathcal K}$, $\Delta_i=\frac{T_i}{N}$ and $\Delta S_{ij}=j\Delta_i$ so that we have 
\begin{equation*}
    \Delta S_i^{\mathcal K}=\sum_{j=1}^N \Delta S_{ij}.
\end{equation*}
We further define similarly $\Delta G_{ij}$ and
\begin{equation*}
    \phi^{\mathcal K}_{i}(\boldsymbol S)\triangleq\sum_{j=1}^N\Delta S_j\Delta G_j,
\end{equation*}
which, similar as \eqref{fullintee}, can be shown to satisfy as $N\rightarrow \infty$,
\begin{equation}\label{ttp1}
    \mathcal I_{\phi^{\mathcal K}_{i}(\boldsymbol S)}(\theta)=\frac{(\frac{\partial f}{\partial \theta})^2}{\sigma^2}\int_{t_{i-1}^{\mathcal K}}^{t_{i}^{\mathcal K}}G^2(s)ds.
\end{equation}
Now, fix $1\leq i \leq n$, as in \eqref{Jfull}, base on Lemma \ref{mutual}, we have 
\begin{equation}\label{ttp2}
    \mathcal I_{\Delta S_i^{\mathcal K}}(\theta)=\mathcal I_{\phi^{\mathcal K}_{i}(\boldsymbol S)}-\mathbb E_{x\sim \Delta S_i^{\mathcal K}}[\mathcal I_{\phi^{\mathcal K}_{i}(\boldsymbol S)|\Delta S_i^{\mathcal K}= x}(\theta)].
\end{equation}
Again, we calculate the joint distribution
\begin{align*}
    \begin{pmatrix}\phi_i^{\mathcal K}(\boldsymbol S)\\\Delta S_i^{\mathcal K}\end{pmatrix}
    =&\begin{pmatrix}\sum_{j\in [N]}\Delta S_{ij}\Delta G_{ij}\\\sum_{j\in [N]}\Delta S_{ij}\end{pmatrix}\nonumber\\ \sim&\mathcal N
        \begin{pmatrix}
        f(v;\theta)\begin{pmatrix}\sum_{j\in [N]} (\Delta G_{ij})^2\\\sum_{j\in [N]}\Delta G_{ij}\end{pmatrix},
          \sigma^2\Delta t\begin{bmatrix}
 \sum_{j\in [N]} (\Delta G_{ij})^2& \sum_{j\in [N]}\Delta G_{ij} \\
\sum_{j\in [N]}\Delta G_{ij} & \sum_{j\in [N]}1
\end{bmatrix}
    \end{pmatrix}.
  \end{align*}
which further gives a conditional distribution 
\begin{align*}
    \phi_i^{\mathcal K}(\boldsymbol S)|\Delta S_i^{\mathcal K}=x\sim &\mathcal N \Bigg(f(v;\theta)\sum_{j\in [N]} (\Delta G_{ij})^2+\frac{\sum_{j\in [N]} \Delta G_{ij}}{\sum_{j\in [N]} 1}\Big(x-f(v;\theta)\cdot\sum_{j\in [N]} \Delta G_{ij}\Big),\nonumber\\
    &\sum_{j\in [N]} (\Delta G_{ij})^2-\frac{(\sum_{j\in [N]} \Delta G_{ij})^2}{\sum_{j\in [N]} 1}\Bigg).
\end{align*} 
Again, following Lemma \ref{jacob}, 
    we have, as $ N\rightarrow \infty$,
\begin{align}\label{ttp3}
    \mathcal I_{\phi_i^{\mathcal K}(\boldsymbol S)|\Delta S_i^{\mathcal K}}(\theta)=&\frac{(\frac{\partial f}{\partial \theta})^2}{\sigma^2}\Big(\sum_{j\in [N]}\frac{(\Delta G_{ij})^2}{\Delta t}-\frac{(\sum_{j\in [N]}\Delta G_{ij})^2}{\sum_{j\in [N]}\Delta t}\Big)\nonumber\\
    \rightarrow& \frac{(\frac{\partial f}{\partial \theta})^2}{\sigma^2}\Big(\int_{t_{i-1}^{\mathcal K}}^{t_{i}^{\mathcal K}} G^2(t)dt-\frac{(\int_{t_{i-1}^{\mathcal K}}^{t_{i}^{\mathcal K}}G(t)dt)^2}{t_{i}^{\mathcal K}-t_{i-1}^{\mathcal K}}\Big).
\end{align}
Now, since $\mathcal I_{\boldsymbol S_{\text{partial}}}(\theta)$ does not depend on $N$, we can combine \eqref{ttp1}, \eqref{ttp2},\eqref{ttp3} and \eqref{ttp4} to show \eqref{42}:
\begin{equation*}
    \mathcal I_{\boldsymbol S_{\text{partial}}}(\theta)= \frac{(\frac{\partial f}{\partial \theta})^2}{\sigma^2}\Big(\sum_{i=1}^{n}\frac{(\int_{t_{i-1}^{\mathcal K}}^{t_{i}^{\mathcal K}}G(t)dt)^2}{t_{i}^{\mathcal K}-t_{i-1}^{\mathcal K}}\Big).
\end{equation*}

\subsubsection{Derivation of Example \ref{expyixia}}\label{asdas}

To calculate \eqref{comppppp}, notice when $G(s)=e^{-\rho s}$, we have 
\begin{equation}\label{8left}
    \frac{(\int_0^t G(t)dt)^2}{t}+\frac{(\int_t^T G(t)dt)^2}{T-t}=\frac{1}{t}\Big(\frac{1-e^{-\rho t}}{\rho}\Big)^2+\frac{1}{T-t}\Big(\frac{e^{-\rho t}-e^{-\rho T}}{\rho}\Big)^2
\end{equation}
and 
\begin{equation}\label{8right}
    \Big(\int_0^T G(t)(T-t)dt\Big)^2=\frac{3}{T^3}\Big(\frac{\rho T+e^{-\rho T}-1}{\rho^2}\Big)^2.
\end{equation}
To show \eqref{8left} is greater than \eqref{8right}, it suffices to show
\begin{equation}\label{8final}
    \frac{(\frac{1-e^{-\rho t}}{\rho})^2}{t/T}+\frac{(\frac{e^{-\rho t}-e^{-\rho T}}{\rho})^2}{1-t/T}\geq \frac{3}{T^2}\Big(\frac{\rho T+e^{-\rho T}-1}{\rho^2}\Big)^2
\end{equation}
Now, let $\frac{t}{T}\rightarrow \tau$ while $t, T \rightarrow \infty$, \eqref{8final} becomes 
\begin{equation*}
    \frac{1}{\tau \rho^2} \geq \frac{3}{\rho^2}.
\end{equation*}
or simply $\tau\leq \frac{1}{3}$.

\subsubsection{Proof of Corollary \ref{corosq}}\label{proofcorosq}
Under the setting of Corollary \ref{corosq}, we have 
\begin{equation*}
   \mu(T,v)=\frac{c}{\delta} (vT)^\delta
\end{equation*}
and \begin{equation*}
    \frac{J}{X}\sim\mathcal{N}(\frac{c}{\delta(1+\delta)}(vT)^\delta, \frac{\sigma^2 T}{3})
\end{equation*}
for some $c$. Under the setting of Corollary \ref{corosq}, for $S_t, S_T$, we have
\begin{gather*}
\mu(\delta) = \left[\begin{array}{c}
    c \frac{1}{\delta} (vt)^\delta  \\
    c \frac{1}{\delta} (vT)^\delta 
\end{array} \right]\\
S_t \sim \mathcal{N} \left(c \frac{1}{\delta} (vt)^\delta, \sigma^2 t \right) \\
S_T \sim \mathcal{N} \left(c \frac{1}{\delta} (vT)^\delta, \sigma^2 T \right) \\
\Sigma =
\frac{1}{\sigma^2 t(T-t)} \left[\begin{array}{cc}
    T & -t \\
    -t & t
\end{array}\right]
\end{gather*}
According to Lemma \ref{jacob}, the Fisher information with respect to $\delta$ is,
\begin{align}
&\mathcal{I}_{S_t, S_T}(\delta) 
= \mathcal{J}(\delta) ^T \Sigma^{-1} \mathcal{J}(\delta) \\
=& \frac{c^2}{\sigma^2 t(T-t)} 
\Bigg(
T \Big(-\frac{1}{\delta^2} (vt)^\delta + \frac{\ln(vt)}{\delta} (vt)^\delta\Big)^2
+ t \Big(-\frac{1}{\delta^2} (vT)^\delta + \frac{\ln(vT)}{\delta} (vT)^\delta \Big)^2 \\
&-2t \Big(-\frac{1}{\delta^2} (vt)^\delta + \frac{\ln(vt)}{\delta} (vt)^\delta \Big)
    \Big(-\frac{1}{\delta^2} (vT)^\delta + \frac{\ln(vT)}{\delta} (vT)^\delta \Big)
\Bigg)
\end{align}
where the Jacobian matrix is,
\begin{equation}
\mathcal{J}(\delta)=
\left[\begin{array}{c}
    \frac{\partial}{\partial \delta} c \frac{1}{\delta} (vt)^\delta \\
    \frac{\partial}{\partial \delta} c \frac{1}{\delta} (vT)^\delta
\end{array}\right]
= c\left[\begin{array}{c}
     -\frac{1}{\delta^2} (vt)^\delta + \frac{\ln(vt)}{\delta} (vt)^\delta \\
     -\frac{1}{\delta^2} (vT)^\delta + \frac{\ln(vT)}{\delta} (vT)^\delta
\end{array}\right]
\end{equation}

For $J$, consider
\begin{gather*}
\frac{J}{X} \sim \mathcal{N} \left(c \frac{1}{\delta(1+\delta)} (vT)^\delta, \frac{\sigma^2 T}{3} \right)
\end{gather*}
Applying Lemma \ref{jacob}, the Fisher information with respect to $\delta$ is,
\begin{equation*}
\mathcal{I}_J(\delta)
= \frac{3 c^2}{\sigma^2 T}
\left(
-\frac{2\delta+1}{\delta^2 (1+\delta)^2} (vT)^\delta
+ \frac{\ln(vT)}{\delta(1+\delta)} (vT)^\delta
\right)^2.
\end{equation*}
Finally, it follows from an ad-hoc numerical analysis, which does not depend on $c$ (as long as it is not 0), that for $0.005\leq vT\leq 1$ and $\delta\leq 0.8$, we have 
\begin{equation*}
\mathcal{I}_J(\delta)\leq\mathcal{I}_{S_t, S_T}
\end{equation*}
as long as $1.2\%\leq t/T \leq 22.2\%$.

\newpage

\end{document}